\documentclass[11pt, letterpaper]{article}

\usepackage[T1]{fontenc}
\usepackage{xspace}
\usepackage[colorlinks=true,linkcolor=blue,citecolor=blue]{hyperref}
\usepackage{physics}
\usepackage{mathtools}
\usepackage{enumitem}
\usepackage{amsmath,amsfonts,amssymb}
\usepackage{multirow,array}
\usepackage{amsthm}
\usepackage[usenames,dvipsnames,svgnames,table]{xcolor}
\usepackage{thmtools, thm-restate}
\usepackage{braket}
\usepackage{wrapfig}
\usepackage{tocloft}
\usepackage{stmaryrd}
\usepackage[utf8]{inputenc}
\usepackage{graphicx}
\usepackage{setspace}
\usepackage{changepage}
\usepackage{hyphenat}
\usepackage{float}
\usepackage[linesnumbered,ruled,vlined]{algorithm2e}
\usepackage{setspace}
\usepackage{etoolbox}
\usepackage{tcolorbox}
\usepackage{tikz,tikz-qtree}
\usepackage{pgfmath, pgfplots}
\usetikzlibrary{arrows, intersections, calc, decorations.pathreplacing, calligraphy, positioning, quantikz2, backgrounds}
\usepackage{stackengine}
\usepackage{cancel}
\usepackage{complexity}
\usepackage{dsfont}
\usepackage[labelfont=bf, font=small, width=1.2\textwidth]{caption}
\usepackage[labelfont=bf, font=small]{subcaption}
\stackMath
\usepackage{ytableau}
\ytableausetup{boxsize=0.5em, aligntableaux=center}
\usepackage[capitalise,nameinlink]{cleveref}

\newlang{\StateHSP}{StateHSP}
\newlang{\HSP}{HSP}
\newlang{\LWE}{LWE}
\newlang{\LPN}{LPN}
\newlang{\SkewLPN}{SkewLPN}
\newlang{\negl}{negl}
\newlang{\HiddenCut}{HiddenCut}

\def\C{{\mathbb{C}}} 
\def\Z{{\mathbb{Z}}} 
\newcommand{\bbS}{\mathbb{S}} 

\renewcommand{\Pr}{\mathop{\mathbb{P}\/}}
\renewcommand{\E}{\mathop{\mathbb{E}\/}}

\newcommand{\Var}{\mathop{\bf Var\/}}
\newcommand{\Cov}{\mathop{\bf Cov\/}}

\newcommand{\Ber}{\mathrm{Ber}} 

\NewCommandCopy{\dashl}{\l} 
\renewcommand{\l}{\ell} 

\newcommand{\B}{\{0,1\}}

\newcommand{\calF}{\mathcal{F}}

\newcommand{\calH}{\mathcal{H}}

\newcommand{\calP}{\mathcal{P}}

\newcommand{\U}{\mathrm{U}}             

\DeclareMathOperator{\cyc}{\mathrm{cyc}}
\DeclareMathOperator{\Irrop}{Irr}

\newcommand{\Haar}{\mathrm{Haar}}     
\newcommand{\Irr}[1]{\Irrop\left[#1\right]}       
\newcommand{\dtr}{\mathrm{D}_{\mathrm{tr}}}   
\newcommand{\swapjohn}{\mathsf{SWAP}}  
\newcommand{\anc}{\mathrm{Anc}}

\ytableausetup{boxsize=0.5cm, centertableaux}
\newlength{\gridboxsize}
\newtheoremstyle{thrmstyle}
  {0pt} 
  {0pt} 
  {\em} 
  {} 
  {\bfseries} 
  {.} 
  {.5em} 
  {} 
\theoremstyle{thrmstyle}
\newtheorem{theorem}{Theorem}
\newtheorem*{theorem*}{Theorem}
\newtheorem{corollary}{Corollary}
\newtheorem{lemma}{Lemma}
\newtheorem{proposition}{Proposition}
\newtheorem{notation}{Notation}
\newtheorem{fact}{Fact}[section]

\newtheorem{definition}{Definition}
\newtheorem*{definition*}{Definition}
\newtheorem{customdefinitionaux}{Definition}
\newenvironment{customdefinition}[1]
  {\renewcommand\thecustomdefinitionaux{#1}\customdefinitionaux}
  {\endcustomdefinitionaux}
\newtheorem{customtheoremaux}{Theorem}
\newenvironment{customtheorem}[1]
  {\renewcommand\thecustomtheoremaux{#1}\customtheoremaux}
  {\endcustomtheoremaux}
\newtheorem{customcorollaryaux}{Corollary}
\newenvironment{customcorollary}[1]
  {\renewcommand\thecustomcorollaryaux{#1}\customcorollaryaux}
  {\endcustomcorollaryaux}

\newtheoremstyle{definition_new}
  {2pt} 
  {2pt} 
  {} 
  {} 
  {\bfseries} 
  {.} 
  {.5em} 
  {} 
\theoremstyle{definition_new}

\newtheorem*{construction*}{Construction}

\newtheorem*{observation*}{Observation}

\newtheorem*{remark}{Remark}
\newtheorem{step}{Step}
\renewenvironment{proof}{\emph{\bfseries Proof.}}{\qed\vspace{2pt}}

\usepackage{titlesec}
\titleformat*{\section}{\large\bfseries}
\titlespacing*{\section}{0pt}{2\baselineskip}{1\baselineskip}
\titleformat*{\subsection}{\normalsize\bfseries}
\titleformat*{\subsubsection}{\normalsize\bfseries}
\titleformat*{\paragraph}{\normalsize\bfseries}
\titleformat*{\subparagraph}{\normalsize\bfseries}

\definecolor{colorone}{rgb}{0.92, 0.45, 0.4}
\definecolor{colortwo}{rgb}{0.15, 0.85, 0.1}

\newcommand{\abelianactionmini}{\tikz[scale=0.5]{
    \def\colsep{0.4}
    \foreach \col in {1,5,6,8} {
        \node (Z\col) at (\col * \colsep + 0.02, 0) {$1$};
        \node[blue, font=\tiny] (I\col) at (\col * \colsep, -1.3) {$\col$};
        \draw[-{Stealth[length=1mm]}, blue, thin] (I\col) -- (Z\col);
    }
    \foreach \col in {2,3,4,7} {
        \node (Z\col) at (\col * \colsep, 0) {$0$};
    }
}
}

\newcommand{\abelianactionzero}{\tikz[scale=0.5]{
        \def\nRows{4}
        \def\nCols{8}
        \def\rdelta{0.8}
        \def\cdelta{0.8}
        \def\lwidth{6}
        \foreach \row in {1,...,\nRows} {
            \draw[line width=\lwidth pt, line cap=round, colorone!35] (1*\cdelta, -\row*\rdelta) -- (\cdelta*\nCols/2, -\row*\rdelta);
            \draw[line width=\lwidth pt, line cap=round, colortwo!35] (\cdelta*\nCols/2 + \cdelta*1, -\row*\rdelta) -- (\cdelta*\nCols, -\row*\rdelta);
        }

        \foreach \row in {1,...,\nRows} {
            \foreach \col in {1,...,\nCols} {
                \node[circle, draw, fill=black, inner sep=0.9pt] (B\row\col) at (\cdelta*\col, -\row*\rdelta) {};
            }
        }
}}

\newcommand{\abelianactionone}{\tikz[scale=0.5]{
        \def\nRows{4}
        \def\nCols{8}
        \def\rdelta{0.8}
        \def\cdelta{0.8}
        \def\lwidth{6}
        \foreach \row in {1,...,\nRows} {
            \draw[line width=\lwidth pt, line cap=round, colorone!35] (\cdelta, -\row*\rdelta) -- (\cdelta*\nCols/2, -\row*\rdelta);
            \draw[line width=\lwidth pt, line cap=round, colortwo!35] (\cdelta*\nCols/2 + \cdelta*1, -\row*\rdelta) -- (\cdelta*\nCols, -\row*\rdelta);
        }
        \foreach \row in {1,...,\nRows} {
            \foreach \col in {1,...,\nCols} {
                \node[circle, draw, fill=black, inner sep=0.9pt] (B\row\col) at (\cdelta*\col, -\row*\rdelta) {};
            }
        }
        \foreach \col in {1, 5, 6, 8}{
            \foreach \row in {1,...,\nRows}{
                \ifodd\row
                \draw[-{Latex[length=1.4mm]}, bend right=45, blue, line width=0.7pt] (\cdelta*\col,-\row*\rdelta) to (\cdelta*\col,-\row*\rdelta-1*\rdelta);
                \else
                \draw[-{Latex[length=1.4mm]}, bend right=45, blue, line width=0.7pt] (\cdelta*\col,-\row*\rdelta) to (\cdelta*\col,-\row*\rdelta+1*\rdelta);
                \fi
            }
            \node[blue, font=\tiny] at (\cdelta*\col, -\rdelta*\nRows - 0.5) {$R_\col$};
        }
}}

\newcommand{\abelianactiontwo}{\tikz[scale=0.5]{
        \def\nRows{4}
        \def\nCols{8}
        \def\rdelta{0.8}
        \def\cdelta{0.8}
        \def\lwidth{6}
        \foreach \row in {1,...,\nRows} {
            \ifodd\row
                \draw[line width=\lwidth pt, line cap=round, colorone!35] (\cdelta, -\row*\rdelta) -- (\cdelta*2, -\row*\rdelta-1*\rdelta) -- (\cdelta*\nCols/2, -\row*\rdelta-1*\rdelta);
                \draw[line width=\lwidth pt, line cap=round, colortwo!35] (\cdelta*\nCols/2 + \cdelta*1, -\row*\rdelta) -- (\cdelta*\nCols/2 + \cdelta*2, -\row*\rdelta) -- (\cdelta*\nCols/2 + \cdelta*3, -\row*\rdelta-\rdelta) -- (\cdelta*\nCols, -\row*\rdelta);
            \else
                \draw[line width=\lwidth pt, line cap=round, colorone!35] (\cdelta, -\row*\rdelta) -- (\cdelta*2, -\row*\rdelta+1*\rdelta) -- (\cdelta*\nCols/2, -\row*\rdelta+1*\rdelta);
                \draw[line width=\lwidth pt, line cap=round, colortwo!35] (\cdelta*\nCols/2 + \cdelta, -\row*\rdelta) -- (\cdelta*\nCols/2 + \cdelta*2, -\row*\rdelta) -- (\cdelta*\nCols/2 + \cdelta*3, -\row*\rdelta+\rdelta) -- (\cdelta*\nCols, -\row*\rdelta);
            \fi
        }

        \foreach \row in {1,...,\nRows} {
            \foreach \col in {1,...,\nCols} {
                \node[circle, draw, fill=black, inner sep=0.9pt] (B\row\col) at (\cdelta*\col, -\row*\rdelta) {};
            }
        }
}}

\newcommand{\abelianactionthree}{\tikz[scale=0.5]{
        \def\nRows{4}
        \def\nCols{8}
        \def\rdelta{0.8}
        \def\cdelta{0.8}
        \def\lwidth{6}
        \foreach \row in {1,...,\nRows} {
            \draw[line width=\lwidth pt, line cap=round, colorone!35] (\cdelta, -\row*\rdelta) -- (\cdelta*\nCols/2, -\row*\rdelta);
            \draw[line width=\lwidth pt, line cap=round, colortwo!35] (\cdelta*\nCols/2 + \cdelta*1, -\row*\rdelta) -- (\cdelta*\nCols, -\row*\rdelta);
        }
        \foreach \row in {1,...,\nRows} {
            \foreach \col in {1,...,\nCols} {
                \node[circle, draw, fill=black, inner sep=0.9pt] (B\row\col) at (\cdelta*\col, -\row*\rdelta) {};
            }
        }
        \foreach \col in {1,...,4}{
            \foreach \row in {1,...,\nRows}{
                \ifodd\row
                \draw[-{Latex[length=1.4mm]}, bend right=45, blue, line width=0.7pt] (\cdelta*\col,-\row*\rdelta) to (\cdelta*\col,-\row*\rdelta-1*\rdelta);
                \else
                \draw[-{Latex[length=1.4mm]}, bend right=45, blue, line width=0.7pt] (\cdelta*\col,-\row*\rdelta) to (\cdelta*\col,-\row*\rdelta+1*\rdelta);
                \fi
            }
        }
}}

\newcommand{\abelianactionfour}{\tikz[scale=0.5]{
        \def\nRows{4}
        \def\nCols{8}
        \def\rdelta{0.8}
        \def\cdelta{0.8}
        \def\lwidth{6}
        \foreach \row in {1,...,\nRows} {
            \draw[line width=\lwidth pt, line cap=round, colorone!35] (\cdelta, -\row*\rdelta) -- (\cdelta*\nCols/2, -\row*\rdelta);
            \draw[line width=\lwidth pt, line cap=round, colortwo!35] (\cdelta*\nCols/2 + \cdelta*1, -\row*\rdelta) -- (\cdelta*\nCols, -\row*\rdelta);
        }
        \foreach \row in {1,...,\nRows} {
            \foreach \col in {1,...,\nCols} {
                \node[circle, draw, fill=black, inner sep=0.9pt] (B\row\col) at (\cdelta*\col, -\row*\rdelta) {};
            }
        }
        \foreach \col in {5,...,8}{
            \foreach \row in {1,...,\nRows}{
                \ifodd\row
                \draw[-{Latex[length=1.4mm]}, bend right=45, blue, line width=0.7pt] (\cdelta*\col,-\row*\rdelta) to (\cdelta*\col,-\row*\rdelta-1*\rdelta);
                \else
                \draw[-{Latex[length=1.4mm]}, bend right=45, blue, line width=0.7pt] (\cdelta*\col,-\row*\rdelta) to (\cdelta*\col,-\row*\rdelta+1*\rdelta);
                \fi
            }
        }
}}

\newcommand{\abelianactionfive}{\tikz[scale=0.5]{
        \def\nRows{4}
        \def\nCols{8}
        \def\rdelta{0.8}
        \def\cdelta{0.8}
        \def\lwidth{6}
        \foreach \row in {1,...,\nRows} {
            \draw[line width=\lwidth pt, line cap=round, colorone!35] (\cdelta, -\row*\rdelta) -- (\cdelta*\nCols/2, -\row*\rdelta);
            \draw[line width=\lwidth pt, line cap=round, colortwo!35] (\cdelta*\nCols/2 + \cdelta*1, -\row*\rdelta) -- (\cdelta*\nCols, -\row*\rdelta);
        }
        \foreach \row in {1,...,\nRows} {
            \foreach \col in {1,...,\nCols} {
                \node[circle, draw, fill=black, inner sep=0.9pt] (B\row\col) at (\cdelta*\col, -\row*\rdelta) {};
            }
        }
        \foreach \col in {1,...,8}{
            \foreach \row in {1,...,\nRows}{
                \ifodd\row
                \draw[-{Latex[length=1.4mm]}, bend right=45, blue, line width=0.7pt] (\cdelta*\col,-\row*\rdelta) to (\cdelta*\col,-\row*\rdelta-1*\rdelta);
                \else
                \draw[-{Latex[length=1.4mm]}, bend right=45, blue, line width=0.7pt] (\cdelta*\col,-\row*\rdelta) to (\cdelta*\col,-\row*\rdelta+1*\rdelta);
                \fi
            }
        }
}}

\newcommand{\cartoonone}{\tikz[scale=0.7,baseline=0pt]{
        \def\cdelta{0.7}
        \def\nCols{8}
        \def\arch{0.6}
        \def\thickness{6pt}
        \def\dotsize{1pt}
        \draw[line width=8pt, line cap=round, colorone!35] (1*\cdelta, 0) -- (\cdelta*\nCols/2, 0);
        \draw[line width=8pt, line cap=round, colortwo!35] (\cdelta*\nCols/2 + 1*\cdelta, 0) -- (\cdelta*\nCols, 0);
        \foreach \col in {1,...,\nCols} {
                \node[circle, draw, fill=black, inner sep=\dotsize] (B\col) at (\cdelta*\col, 0) {};
        }
        \node[purple] at (2.5*\cdelta, 0.35) {\tiny $|\phi_1\rangle$};
        \node[teal] at (6.5*\cdelta, -0.35) {\tiny $|\phi_2\rangle$};
}}

\newcommand{\cartoontwo}{\tikz[scale=0.7]{
        \def\cdelta{0.7}
        \def\nCols{8}
        \def\arch{0.6}
        \def\thickness{6pt}
        \def\dotsize{1pt}
            \draw[line width=\thickness, line cap=round, colorone!35] (1*\cdelta, 0)..controls ++(\cdelta,\arch)..(3*\cdelta,0)..controls ++(\cdelta, \arch)..(5*\cdelta, 0)..controls ++(\cdelta, \arch)..(7*\cdelta, 0);
            \draw[line width=\thickness, line cap=round, colortwo!35] (2*\cdelta, 0)..controls ++(\cdelta,-\arch)..(4*\cdelta,0)..controls ++(\cdelta, -\arch)..(6*\cdelta, 0)..controls ++(\cdelta, -\arch)..(8*\cdelta, 0);
        \foreach \col in {1,...,\nCols} {
                \node[circle, draw, fill=black, inner sep=\dotsize] (B\col) at (\cdelta*\col, 0) {};
        }
        \node[purple] at (4*\cdelta, 1.25*\arch) {\tiny $|\phi_1\rangle$};
        \node[teal] at (5*\cdelta, -1.25*\arch) {\tiny $|\phi_2\rangle$};
}}

\newcommand{\cartoonthree}{\tikz[scale=0.7]{
        \def\cdelta{0.7}
        \def\nCols{8}
        \def\arch{0.6}
        \def\thickness{6pt}
        \def\dotsize{1pt}
            \draw[line width=\thickness, line cap=round, colorone!35] (1*\cdelta, 0) -- (2*\cdelta, 0)..controls ++(1.5*\cdelta, \arch)..(5*\cdelta, 0) -- (6*\cdelta, 0);
            \draw[line width=\thickness, line cap=round, colortwo!35] (3*\cdelta, 0) -- (4*\cdelta, 0)..controls ++(1.5*\cdelta, -\arch)..(7*\cdelta, 0) -- (8*\cdelta, 0);
        \foreach \col in {1,...,\nCols} {
                \node[circle, draw, fill=black, inner sep=\dotsize] (B\col) at (\cdelta*\col, 0) {};
        }
        \node[purple] at (3.5*\cdelta, 1.25*\arch) {\tiny $|\phi_1\rangle$};
        \node[teal] at (5.5*\cdelta, -1.25*\arch) {\tiny $|\phi_2\rangle$};
}}

\newcommand{\cartoonfour}{\tikz[scale=0.7]{
        \def\cdelta{0.7}
        \def\nCols{8}
        \def\arch{0.6}
        \def\thickness{6pt}
        \def\dotsize{1pt}
            \draw[line width=\thickness, line cap=round, colorone!35] (2*\cdelta,0)..controls ++(1*\cdelta, \arch)..(4*\cdelta,0) -- (6*\cdelta, 0);
            \draw[line width=\thickness, line cap=round, colortwo!35] (1*\cdelta, 0)..controls ++(\cdelta,-\arch)..(3*\cdelta,0)..controls ++(2*\cdelta, -\arch)..(7*\cdelta, 0) -- (8*\cdelta, 0);
        \foreach \col in {1,...,\nCols} {
                \node[circle, draw, fill=black, inner sep=\dotsize] (B\col) at (\cdelta*\col, 0) {};
        }
        \node[purple] at (4*\cdelta, 1*\arch) {\tiny $|\phi_1\rangle$};
        \node[teal] at (5*\cdelta, -1.25*\arch) {\tiny $|\phi_2\rangle$};
}}

\newcommand{\qubitgrid}{\tikz[scale=0.6]{
        \def\nRows{4}
        \def\nCols{8}
        \def\rowsep{1}
        \def\colsep{1.2}
        \foreach \row in {1,...,\nRows} {
            \draw[line width=9pt, line cap=round, colorone!35] (1, -\row*\rowsep) -- (\nCols/2, -\row*\rowsep);
            \draw[line width=9pt, line cap=round, colortwo!35] (\nCols/2 + 1, -\row*\rowsep) -- (\nCols, -\row*\rowsep);
        }

        \foreach \row in {1,...,\nRows} {
            \foreach \col in {1,...,\nCols} {
                \node[circle, draw, fill=black, inner sep=1.5pt] (B\row\col) at (\col, -\row*\rowsep) {};
            }
            \node[purple] at (\nCols/4 + 0.5, -\row*\rowsep) {\tiny $|\phi_1\rangle$};
            \node[teal] at (3*\nCols/4 + 0.5, -\row*\rowsep) {\tiny $|\phi_2\rangle$};
        }
        \foreach \col in {1,...,4}{
            \draw[-{Latex[length=2mm]}, bend right=30, blue, line width=0.9pt] (\col,-1*\rowsep) to (\col,-2*\rowsep);
            \draw[-{Latex[length=2mm]}, bend right=30, blue, line width=0.9pt] (\col,-2*\rowsep) to (\col,-4*\rowsep);
            \draw[-{Latex[length=2mm]}, bend right=30, blue, line width=0.9pt] (\col,-4*\rowsep) to (\col,-3*\rowsep);
            \draw[-{Latex[length=2mm]}, bend right=30, blue, line width=0.9pt] (\col,-3*\rowsep) to (\col,-1*\rowsep);
        }
        \node[blue, font=\smaller] at (2.5,-\nRows*\rowsep-0.7*\rowsep) {$\bigotimes_{j\in C}R_j(\sigma)$};
        \foreach \col in {5, ..., 8}{
            \draw[-{Latex[length=2mm]}, bend right=30, blue, line width=0.9pt] (\col,-1*\rowsep) to (\col,-3*\rowsep);
            \draw[-{Latex[length=2mm]}, bend right=30, blue, line width=0.9pt] (\col,-3*\rowsep) to (\col,-1*\rowsep);
        }
        \node[blue, font=\smaller] at (6.5,-\nRows*\rowsep-0.7*\rowsep) {$\bigotimes_{j\in \overline{C}}R_j(\pi)$};

        \draw [decorate, decoration={brace, amplitude=10pt, mirror}, black] (0.5, -0.9*\rowsep) -- (0.5, -\nRows*\rowsep-0.1*\rowsep) node[midway, right, xshift=-48pt, black, font=\small] {$t$ copies};
        \draw [decorate, decoration={brace, amplitude=5pt}, black] (0.9, -0.5*\rowsep) -- (\nCols/2+0.1, -0.5*\rowsep) node[midway, black, yshift=12pt, font=\small] {$C$};
        \draw [decorate, decoration={brace, amplitude=5pt}, black] (\nCols/2 + 0.9, -0.5*\rowsep) -- (\nCols+0.1, -0.5*\rowsep) node[midway, black, yshift=12pt, font=\small] {$\overline{C}$};
    }
}

\newcommand{\swaptest}{\begin{tikzpicture}[scale=0.5]
                  \tikzset{
                  operator/.style = {draw,thick,fill=green!5,minimum width=1.5em,minimum height=1.5em},
                  phase/.style = {draw,fill,shape=circle,minimum size=3.5pt,inner sep=0pt},
                  surround/.style = {fill=blue!10,thick,draw=black,rounded corners=2mm},
                  meter/.style={draw, inner sep=3, rectangle, font=\vphantom{A}, minimum width=2em, minimum height=1.5em, line width=.8, path picture={\draw[black] ([shift={(0.1,0.1)}]path picture bounding box.south west) to[bend left=60] ([shift={(-0.1,.1)}]path picture bounding box.south east);\draw[black,-latex] ([shift={(0,.1)}]path picture bounding box.south) -- ([shift={(.3,-.1)}]path picture bounding box.north);}}
                  }
                  \def\rowsep{0.3cm};
                  \def\colsep{0.4cm};
                  \def\swaproundfrac{1.0};
                  \matrix[row sep=\rowsep, column sep=\colsep, ampersand replacement=\&] (circuit) { 
                      \node (q1) {$\ket{0}$};  
                  \&   [-0.4*\colsep] \node[operator] (H1) {$H$}; 
                  \&   \node[phase] (ctrl) {};                         
                  \&   \node[operator] (H2) {$H$};
                  \& [-0.4*\colsep] \node[meter] (end1) {}; \\
                  \node (q2) {$\ket{\alpha}$};                                      
                  \& \node[\empty] (a1) {};
                  \& \node[\empty] (b1) {};                                  
                  \& \node[\empty] (c1) {};
                  \&\coordinate (end2);\\
                  \node (q3) {$\ket{\beta}$};                                  
                  \&\node[\empty] (a2) {};
                  \&\node[\empty] (b2) {};                           
                  \&\node[\empty] (c2) {};
                  \&\coordinate (end3);\\
                  };
                  \node at (end1) [above=0.8*\rowsep, font=\smaller] {$\B$};
                  \node[\empty] (i1) at ($(a1)!0.5!(b1)$) {};
                  \node[\empty] (j1) at ($(b1)!0.5!(c1)$) {};
                  \node[\empty] (i2) at ($(a2)!0.5!(b2)$) {};
                  \node[\empty] (j2) at ($(b2)!0.5!(c2)$) {};
                  \node (i1right) at ($(i1)+(\swaproundfrac*\colsep,0)$) {};
                  \node (i2right) at ($(i2)+(\swaproundfrac*\colsep,0)$) {};
                  \node (j1left) at ($(j1)+(-\swaproundfrac*\colsep,0)$) {};
                  \node (j2left) at ($(j2)+(-\swaproundfrac*\colsep,0)$) {};
                  \node[scale=0.75] (int1) at (intersection of i1right--j2left and i2right--j1left) {}; 
                  \node[draw=black,thick,fill=white,minimum width=2.2*\colsep, minimum height=4.3*\rowsep] (swap1) at (int1) {};
                  \draw[thick] (q1) -- (H1) -- (ctrl) -- (H2) -- (end1)  
                  (q2) -- (i1.center)..controls (i1right) .. (int1)
                  (int1) .. controls (j2left) .. (j2.center) -- (end3)
                  (q3) -- (i2.center)..controls (i2right) .. (int1.center) .. controls (j1left) .. (j1.center) -- (end2);
                  \draw[thick] (ctrl) -- (swap1);
                  \draw[decorate,decoration={brace,mirror}]
                    ($(q2)+(-20pt,10pt)$)
                    to node[midway,left,font=\smaller] {$\ket{\psi}$}
                    ($(q3)+(-20pt,-10pt)$);
\end{tikzpicture}
}

\newcommand{\StateHSPFourierSampling}{
     \begin{tikzpicture}[scale=0.5]
                \tikzset{
                operator/.style = {draw,thick,fill=green!5,minimum width=2.5em,minimum height=2.5em},
                phase/.style = {draw,fill,shape=circle,minimum size=3.5pt,inner sep=0pt},
                surround/.style = {fill=blue!10,thick,draw=black,rounded corners=2mm},
                meter/.style={draw, inner sep=3, rectangle, font=\vphantom{A}, minimum width=2em, minimum height=1.5em, line width=.8, path picture={\draw[black] ([shift={(0.1,0.1)}]path picture bounding box.south west) to[bend left=60] ([shift={(-0.1,.1)}]path picture bounding box.south east);\draw[black,-latex] ([shift={(0,.1)}]path picture bounding box.south) -- ([shift={(.3,-.1)}]path picture bounding box.north);}}
                }
                \def\rowsep{0.3cm};
                \def\colsep{0.6cm};
                \def\swaproundfrac{1.0};
                \matrix[row sep=\rowsep, column sep=\colsep,ampersand replacement=\&] (circuit) { 
                     \node (q1) {};  
                \&   [-0.4*\colsep] \node[operator] (H1) {$\calF_G^{-1}$}; 
                \&   \node (ctrl) {};                         
                \&   \node[operator] (H2) {$\calF_G$};
                \& [-0.6*\colsep] \node[meter] (end1) {}; \\
                \node (q2) {};                                      
                \& \node[\empty] (a1) {};
                \& \node[\empty] (b1) {};                               
                \& \node[\empty] (c1) {};
                \&\coordinate (end2);\\
                };
                \node at (end1) [above=0.8*\rowsep, font=\smaller] {$\;\;\;\;\lambda\in\mathrm{Irr}[G]$};
                \node at (q1) [left=0.05*\colsep] {$\C^{G}: \ket{\mathbf{0}}$};
                \node at (q2) [left=0.0*\colsep] {$\C^d: \ket{\psi}$};
                \node[draw=black,thin,rounded corners,dashed,fill=blue!5!white,minimum width=4*\rowsep, minimum height=7*\rowsep] (spot) at ($(b1)!0.5!(ctrl)$) {};
                \node[draw=black,thick,fill=white,minimum width=3.0*\rowsep, minimum height=3.0*\rowsep] (gate1) at (b1) {$R$};
                \draw[thick] (q1) -- (H1) -- (ctrl.center) -- (H2) -- (end1)  
                (q2) -- (gate1) -- (end2);
                \node[phase] at (ctrl) {}; 
                \node[font=\smaller] at (spot) [above=3.4*\rowsep] {$U_{R}$};
                \draw[thick] (ctrl.center) -- (gate1);
    \end{tikzpicture}
}

\newcommand{\HiddenCutCircuit}{
    \begin{tikzpicture}[scale=0.5]
            \tikzset{
            operator/.style = {draw,thick,fill=green!5,minimum width=1.5em,minimum height=1.5em},
            phase/.style = {draw,fill,shape=circle,minimum size=2.5pt,inner sep=0pt},
            surround/.style = {fill=blue!10,thick,draw=black,rounded corners=2mm},
            meter/.style={draw, inner sep=3, rectangle, font=\vphantom{A}, minimum width=2em, minimum height=1.5em, fill=white, line width=.8, path picture={\draw[black] ([shift={(0.1,0.1)}]path picture bounding box.south west) to[bend left=60] ([shift={(-0.1,.1)}]path picture bounding box.south east);\draw[black,-latex] ([shift={(0,.1)}]path picture bounding box.south) -- ([shift={(.3,-.1)}]path picture bounding box.north);}}
            }
            \def\rowsep{0.3cm};
            \def\colsep{0.6cm};
            \def\swaproundfrac{0.85};

            \foreach \step in {0,...,5}{
                \begin{scope}[shift={(\step * 0.2,-\step * 0.15)}]
                    \matrix[row sep=\rowsep, column sep=\colsep,ampersand replacement=\&] (circuit) { 
                    \node[minimum height=1.5em] (q1) {};  
                    \&   \node[operator] (H1) {$H$}; 
                    \&   \node[phase] (ctrl) {};                         
                    \&   \node[operator] (H2) {$H$};
                    \&   \node[meter] (end1) {}; \\
                    \node[minimum height=1.5em] (q2\step) {};                                      
                    \& \node[\empty] (a1) {};
                    \& \node[\empty] (b1) {};                                  
                    \& \node[\empty] (c1) {};
                    \& \coordinate (end2);\\
                    \node[minimum height=1.5em] (q3\step) {};                                  
                    \&\node[\empty] (a2) {};
                    \&\node[\empty] (b2) {};                           
                    \&\node[\empty] (c2) {};
                    \&\coordinate (end3);\\
                    \node[minimum height=1.5em] (q4\step) {};                                      
                    \& \node[\empty] (a3) {};
                    \& \node[\empty] (b3) {};                                  
                    \& \node[\empty] (c3) {};
                    \&\coordinate (end4);\\
                    \node[minimum height=1.5em] (q5\step) {};                                      
                    \& \node[\empty] (a4) {};
                    \& \node[\empty] (b4) {};                                  
                    \& \node[\empty] (c4) {};
                    \& \coordinate (end5);\\
                    };
                    \foreach \k in {1,...,4} {
                        \node[\empty] (i\k) at ($(a\k)!0.5!(b\k)$) {};
                        \node[\empty] (iright\k) at ($(i\k)+(\swaproundfrac*\colsep,0)$) {};
                        \node[\empty] (j\k) at ($(b\k)!0.5!(c\k)$) {};
                        \node[\empty] (jleft\k) at ($(j\k)+(-\swaproundfrac*\colsep,0)$) {};
                    }
                    \node[scale=0.75] (int1) at (intersection of iright1--jleft2 and iright2--jleft1) {}; 
                    \node[scale=0.75] (int2) at (intersection of iright3--jleft4 and iright4--jleft3) {};
                    \node[draw=black,thick,fill=white,minimum width=1.75*\colsep, minimum height=9.7*\rowsep] (swap) at ($(int1)!0.5!(int2)$) {};
                    \draw[] (q1) -- (H1) -- (ctrl.center) -- (H2) -- (end1)  
                    (q2\step) -- (i1.center)..controls (iright1) .. (int1)
                    (int1) .. controls (jleft2) .. (j2.center) -- (end3)
                    (q3\step) -- (i2.center)..controls (iright2) .. (int1.center) .. controls (jleft1) .. (j1.center) -- (end2)
                    (q4\step) -- (i3.center)..controls (iright3) .. (int2)
                    (int2) .. controls (jleft4) .. (j4.center) -- (end5)
                    (q5\step) -- (i4.center)..controls (iright4) .. (int2.center) .. controls (jleft3) .. (j3.center) -- (end4);
                    \draw[thick] (ctrl.center) -- (swap);
                \end{scope}
            }
            \begin{scope}[shift={(10,0)}]
            \foreach \j in {2,...,5} {
                \draw[line width=4pt, line cap=round, red!35] (q\j0.center) -- ($(q\j0)!0.4!(q\j5)$);
                \draw[line width=4pt, line cap=round, green!35] ($(q\j0)!0.6!(q\j5)$) -- (q\j5.center);
                \foreach \step in {0,...,5}{
                    \node[circle, fill=black, minimum size=0.25em, inner sep=0pt] at ($(q\j0)!{\step/5}!(q\j5)$) {};
                } 
            }
        \end{scope}
        \node at (q1) [left=1em, yshift=0.5em, font=\smaller] {$\ket{0^n}$};
        \node at (end1) [above=2em,xshift=-0.5em, font=\smaller] {$\B^n$};
        \draw[decorate,decoration={brace,mirror}]
        ($(q20)+(-8pt,0)$)
        to node[midway,left,font=\smaller] {$t$ copies}
        ($(q50)+(-8pt,0)$);
        \draw[decorate,decoration={brace,mirror}]
        ($(q50)+(-3pt,-8pt)$)
        to node[midway,below,font=\smaller, xshift=-3pt] {$n$}
        ($(q55)+(-3pt,-8pt)$);
    \end{tikzpicture}
}

\usepackage[left=1in,right=1in,top=1in,bottom=1in]{geometry}
\usepackage[skip=10pt, indent=0pt]{parskip}

\allowdisplaybreaks

\title{\vspace{-40pt}\large\bf The state hidden subgroup problem \\ and an efficient algorithm for locating unentanglement \vspace{-10pt}}

\author{
\normalsize Adam Bouland\thanks{\href{mailto:abouland@stanford.edu}{abouland@stanford.edu}} \\ \small Stanford University
\and
\normalsize Tudor Giurgic\u{a}-Tiron\thanks{\href{mailto:tgt@stanford.edu}{tgt@stanford.edu}} \\ \small Stanford University, QuICS
\and
\normalsize John Wright\thanks{\href{mailto:jswright@berkeley.edu}{jswright@berkeley.edu}} \\ \small UC Berkeley
}
\date{}

\begin{document}

\maketitle
\vspace{-20pt}
\begin{abstract}\noindent
We study a generalization of entanglement testing which we call the ``hidden cut problem.'' Taking as input copies of an $n$-qubit pure state which is product across an unknown bipartition, the goal is to learn precisely \emph{where} the state is unentangled, i.e.\  to determine which of the exponentially many possible cuts separates the state. We give a polynomial-time quantum algorithm which can find the cut using $O(n/\epsilon^2)$ many copies of the state, which is optimal up to logarithmic factors. Our algorithm also generalizes to learn the entanglement structure of arbitrary product states. In the special case of Haar-random states, we further show that our algorithm requires circuits of only constant depth. 
To develop our algorithm, we introduce a state generalization of the hidden subgroup problem (StateHSP) which might be of independent interest, in which one is given a quantum state invariant under an unknown subgroup action, with the goal of learning the hidden symmetry subgroup. We show how the hidden cut problem can be formulated as a StateHSP with a carefully chosen Abelian group action. We then prove that Fourier sampling on the hidden cut state produces similar outcomes as a variant of the well-known Simon's problem, allowing us to find the hidden cut efficiently. Therefore, our algorithm can be interpreted as an extension of Simon's algorithm to entanglement testing. We discuss possible applications of StateHSP and hidden cut problems to cryptography and pseudorandomness.
\end{abstract}

\section{Introduction}
Detecting the entanglement properties of states is a central theme in quantum information. 
In the standard formulation of product testing \cite{montanaro2013survey}, the goal is to determine if a state is product vs.\  far from product across a fixed bipartition, given as input copies of the state. The well-known ``SWAP test'' \cite{gottesman2001quantum} provides a fundamental algorithmic primitive for entanglement testing in this setting. 
Variations of state product and separability testing have found many applications in quantum complexity theory \cite{gharibian2008strong,gutoski2013quantum}, and many problems which admit states as inputs can be reduced to questions of detecting entanglement --- such as the proof that $\textsf{QMA}(k)=\textsf{QMA}(2)$ \cite{harrow2013testing}.
Recently a number of works have shown upper and lower bounds \cite{badescu2020lower,soleimanifar2022testing,flammia2024quantum} for estimating various measures of entanglement/separability for quantum states.

In this work we study a generalization of product testing which we call the ``hidden cut problem,''
which removes the assumption of a pre-defined bipartition. 
Given as input copies of an $n$-qubit pure state, the task is not to determine {\em if} the state is unentangled,
but rather to determine {\em where} it is unentangled. 
In other words, given a state which is promised to be a product of two $(n/2)$-qubit states across some unknown ``cut'' (i.e.\  a bipartition of the $n$ qubits), the goal is to learn the precise location of the cut. The combinatorics of set partitions makes the problem non-trivial: even though one can efficiently check any candidate cut via SWAP test on two state copies, there are exponentially many possible bipartitions of $n$ qubits into two sets of size $n/2$, since $\frac{1}{2}\binom{n}{n/2}=2^{\Theta(n)}$; therefore, a brute-force search is inefficient. See \Cref{fig:cartoonbipartitions} for an illustration.
We define the hidden cut problem more formally as follows:
\begin{figure}
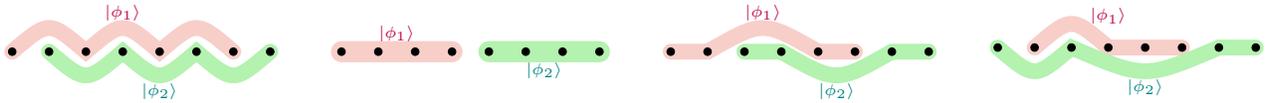

    \begin{equation*}
        \hspace*{-0.1in}
        \begin{array}{ccccccc}
            \vcenter{\hbox{\cartoontwo}} & & \vcenter{\hbox{\cartoonone}} & & \vcenter{\hbox{\cartoonthree}} & & \vcenter{\hbox{\cartoonfour}} \\
        \end{array}
    \end{equation*}\vspace*{-15pt}
    \caption{A cartoon illustrating four out of the $\frac{1}{2}\binom{8}{4}=35$ possible distinct ways (i.e.\  `cuts') in which a pure state on $n=8$ qubits can be formed as a product of two states $\ket{\phi_1},\ket{\phi_2}$ on $\frac{n}{2}=4$ qubits. The factor states $\ket{\phi_1}$ and $\ket{\phi_2}$ are depicted as contiguous clouds. The goal of the hidden cut problem is to determine which cut separates the product state.}
    \label{fig:cartoonbipartitions}
\end{figure}
\begin{restatable}[Hidden cut problem]{definition}{defHiddenCut}\label{def:HiddenCut}
    Let $\ket{\psi}\in (\C^{2})^{\otimes n}$ be a state on $n$ qubits (where $n$ is even) which is promised to be a product across an unknown cut $C\in\binom{[n]}{n/2}$, denoted $\ket{\psi}=\ket{\phi_1}_C\otimes \ket{\phi_2}_{\overline{C}}$, such that the factor states $\ket{\phi_{1}},\ket{\phi_2}\in(\C^{2})^{\otimes n/2}$ are $\epsilon$-far from being product states. The hidden cut problem asks to identify the cut $C$, given copies of the state $\ket{\psi}$.
\end{restatable}

Note that in order for the cut to be well-defined, we require that the two unentangled factor states are far from product states themselves. 
Additionally, while we have defined the hidden cut problem to refer to equal-sized bipartitions of the qubits, one can of course generalize the problem to unequally sized cuts or to a product of more than two states, which we will address later in the paper.

A natural first question asks whether the hidden cut problem is even information-theoretically solvable given polynomially many copies of the input state. The answer is yes --- a closely related problem has been analyzed in the property testing literature under the name of ``multipartite productness''. If we interpret the hidden cut problem as a search task (which cut separates the state?), then multipartite productness is defined as the associated decision task (is there a cut which separates the state?). In \cite{harrow2017sequential}, Harrow, Lin, and Montanaro showed that only $O(n/\epsilon^2)$ state copies are information-theoretically required to decide if a state is multipartite product or $\epsilon$-far from any such state, and this bound was recently shown to be optimal up to log factors by Jones and Montanaro \cite{jones2024testing}. 
A simple modification of Harrow, Lin and Montanaro's argument shows the cut can also be located using only $O(n/\epsilon^2)$ copies of the input state (see \Cref{sec:infotheoretic}).
However, these property testing algorithms are computationally inefficient, and work by combining an exponentially long sequence of ``gentle measurements'' on the input state to try out the different possible cuts.
This approach takes exponential time, and there is no obvious way to do better.

Our main result is a positive answer to this question: we construct an efficient algorithm for the hidden cut problem, which can learn the cut using $O(n/\epsilon^2)$ copies of the state and polynomial time. This is exponentially faster than prior property testing algorithms:
\begin{restatable}[Hidden cut algorithm]{theorem}{thmmain}\label{thm:main}
There is an efficient quantum algorithm for the hidden cut problem on $n$ qubits which uses $O(n/\epsilon^2)$ copies of the state and runs in polynomial time, requiring circuits of depth $O(n^2)+O(\log\epsilon^{-1})$ which act coherently on $O(\epsilon^{-2})$ state copies at a time.
\end{restatable}

The number of copies used by our algorithm is optimal up to log factors, in light of Jones and Montanaro's decision lower bound \cite{jones2024testing}.
Our algorithm works in an entirely different way than the prior information-theoretic approaches. In particular, we show the hidden cut problem can formulated as a quantum state version of the well-known Abelian hidden subgroup problem (HSP), and then give an efficient algorithm to solve this quantum state HSP via a generalization of Simon's algorithm \cite{simon1997power}. Our work can thus be interpreted as an extension of Simon's algorithm to entanglement testing.

In addition to being a natural question in entanglement testing, the hidden cut problem is also motivated by questions in quantum pseudorandomness and pseudoentanglement, since hiding the location of a separating cut could be a natural mechanism of hiding entanglement.
A number of works have recently explored creating quantum pseudorandom states \cite{ji2018pseudorandom} with low entanglement \cite{aaronson2024quantum}.
One natural recursive candidate construction\footnote{We thank Henry Yuen for raising this question.} would consist of building larger states from products of two smaller pseudorandom states across a random partition.
This could potentially result in a pseudorandom state construction with no entanglement across some partition, and which naturally lifts pseudorandom state construction on $n$ qubits to pseudorandom state constructions on $n'>n$ qubits.
Our result shows this approach does not work, as there is an efficient algorithm to locate unentanglement.
Interestingly, our algorithm does not rule out the possibility of more general pseudorandom state constructions with low, but nonzero entanglement across hidden cuts (see Discussion section).

Our algorithm also admits a number of generalizations and improvements. First, our algorithm generalizes to the case in which the state is a product of two or more unentangled states which are not necessarily of the same size. Specifically, we naturally define the more generic ``hidden many-cut problem'', in which the input state is allowed to be a product of several unentangled substates. We will show that the same algorithm can solve this version of the problem with minimal modifications, as long as the individual subsystems are entangled enough:
\begin{restatable}[Algorithm for the many-cut problem  --- informal]{corollary}{cormanycut}\label{cor:manycut}
    The same algorithm from \Cref{thm:main} solves the ``hidden many-cut problem'' in which an $n$-qubit input state is product across an arbitrary set partition $C_1\sqcup\dots\sqcup C_m = [n]$: 
    \begin{equation}
        \ket{\psi}=\ket{\phi_1}_{C_1}\otimes \dots \otimes \ket{\phi_m}_{C_m}\,,
    \end{equation}
    such that each factor state $\ket{\phi_k}$ (where $k \in [m]$ indexes the $m$ parts) is at least $\epsilon$-far from any separable state on $\abs{C_k}$ qubits. The algorithm identifies the set partition with the same resource and runtime requirements as in \Cref{thm:main}.
\end{restatable}

Second, a stronger promise about the internal entanglement structure of the input states can significantly reduce the runtime requirements of the algorithm. For example, if one assumes the input state is a product of two Haar random states, the algorithm works with constant-depth circuits, acting on only two copies at a time:
\begin{restatable}[Hidden cut algorithm with Haar-random states]{theorem}{thmhaar}\label{thm:haar}
    Under the stronger promise of Haar-random factor states, the hidden cut can be found by the same algorithm with only $O(n)$ copies of the state, involving circuits of constant depth which coherently access only two state copies at a time.
\end{restatable}

This highly efficient version of our algorithm still works when the factor states are not genuinely Haar, but rather computationally indistinguishable from Haar (i.e. pseudorandom states, for which known efficient constructions exist \cite{ji2018pseudorandom}). We conjecture that the scope of this algorithm can be further extended to other families of states which obey a strong entanglement volume law. For these reasons, this highly efficient version of our algorithm could be of particular relevance to near-term experiments.

\subsection{A hidden subgroup problem for states}

We will now describe a conceptual framework which will motivate our design of the algorithm for the hidden cut problem. As with many other quantum algorithms, the key is to make critical use of the symmetries of the input states. The conceptual contribution is to recognize that the hidden cut problem, as well as potentially many other problems with state inputs, can be formulated within a quantum state generalization of the well-known hidden subgroup problem.

We start by recalling that a core algorithmic design principle for state-input problems is accounting for the symmetries of the global state $\ket{\psi}^{\otimes t}$ made up of copies of the input state $\ket{\psi}$ (see e.g.\ \cite{montanaro2013survey, o2016efficient}). In the hidden cut problem, we observe that the global state will have various internal symmetries which are determined by the location of the hidden cut. In other words, we can define a group action on the global state such that each hidden cut will correspond to a unique subgroup of hidden symmetries. This is reminiscent of the  {\em hidden subgroup problem} (HSP), a central framework in quantum algorithms and complexity \cite[Section 5.4.3]{nielsen2010quantum}, in which the task is also to identify a hidden subgroup $H$ given a function on a parent group $G$ which is invariant under $H$. However, there is a fundamental difference as the HSP takes as input a \emph{function} with specific subgroup symmetries. In contrast, the hidden cut problem takes as input \emph{quantum states} with particular sets of symmetries.

Motivated by this observation, we define a quantum state version of the HSP, which we call the {\em state hidden subgroup problem} (StateHSP). This problem takes as input (copies of) a state which admits an efficient action of a finite group, such that the state is preserved by an unknown subgroup; the task is once again to identify the symmetry subgroup:
\begin{figure}
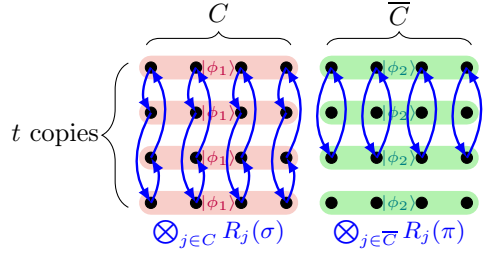

    \centering
    \hspace*{-0.6in}\qubitgrid
    \caption{An illustration of the permutational symmetries of the $tn$-qubit global state $\ket{\psi}^{\otimes t}$ when $\ket{\psi}$ is product across a cut $C\subset[n]$, depicted for $n=8$ and $t=6$. The same cross-copy permutation applied to all qubits inside each side of the cut leaves the global state invariant; to each hidden cut corresponds a hidden subgroup isomorphic to $S_t^{\times 2}$ inside the group $S_t^{\times n}$. In \Cref{sec:symmetricgroupHSP}, we show that Fourier sampling with this group action fails to identify the cut.}
    \label{fig:example}
\end{figure}
\begin{restatable}[The state hidden subgroup problem (StateHSP) --- informal]{definition}{defStateHSP}\label{def:StateHSP} Let $G$ be a finite group with a unitary representation $R:G\rightarrow \mathrm{U}(d)$. The goal is to identify the unknown hidden subgroup $H<G$, given access to the representation and (copies of) a quantum state $\ket{\psi} \in \C^d$ with the following properties:
\begin{itemize}[itemsep=-0.5em, leftmargin=*]
    \vspace*{-0.5em}
    \item $\ket{\psi}$ is invariant under the action of the subgroup $H$, i.e.\  for all $h\in H$, $R(h)\ket{\psi} = \ket{\psi}$.
    \item $\ket{\psi}$ is acted on nontrivially by elements outside the subgroup: for any $g\notin H$, $\abs{\mel{\psi}{R(g)}{\psi}} \leq 1-\epsilon$.
\end{itemize}
\end{restatable}
StateHSP can be interpreted as a generalization of the standard HSP in the following concrete sense: the canonical approach to the standard hidden subgroup problem already involves the construction of states with subgroup symmetries in the form of {\em coset states} \cite{childs2010quantum}. Whereas the coset states live in the regular representation of the group, StateHSP generalizes the state problem to arbitrary group representations. We will elaborate in more technical detail in \Cref{sec:StateHSP}.
We hope the StateHSP problem might be of independent interest, as it is a natural generalization of the HSP. 

In order to apply this framework to the hidden cut problem, one must (a) find an appropriate group action such that the hidden cut problem is formulated as a StateHSP, and (b) find an algorithm to solve the corresponding StateHSP. 
For the latter problem, we will later show that techniques for solving the standard HSP  -- such as an appropriate generalization of Fourier sampling -- can be ported over to the StateHSP setting. 
But first we need to explain why the hidden cut problem is a StateHSP in the first place.

\subsection{The Hidden cut problem as an Abelian StateHSP}
\begin{figure}
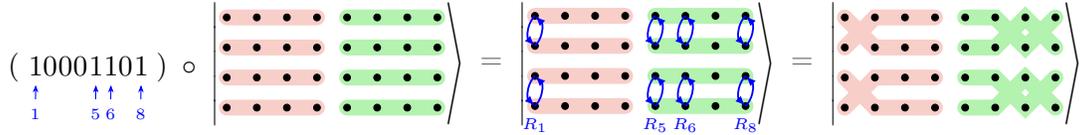
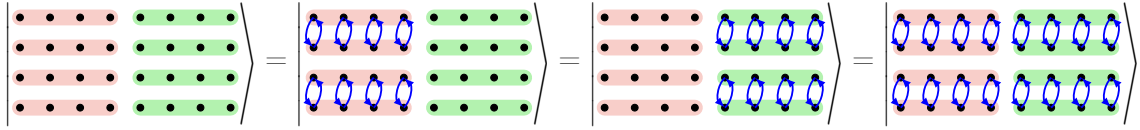

\centering
\begin{subfigure}[t]{\linewidth}
\begin{equation*}(\vcenter{\vspace*{0.22in}\hbox{\abelianactionmini}}\hspace*{-0.06in})\;\circ\;\resizebox{3pt}{27pt}{$\Biggr\vert$}\vcenter{\vspace*{-0.045in}\hbox{\abelianactionzero}}\resizebox{7pt}{27pt}{$\Biggr\rangle$} \vcenter{\vspace*{-0.04in}\hbox{\;\;=\;\;}} \resizebox{3pt}{27pt}{$\Biggr\vert$}\hspace*{-0.07in}\vcenter{\vspace*{0.08in}\hbox{\abelianactionone}}\hspace*{-0.06in}\resizebox{7pt}{27pt}{$\Biggr\rangle$} \vcenter{\vspace*{-0.04in}\hbox{\;\;=\;\;}}\resizebox{3pt}{27pt}{$\Biggr\vert$}\vcenter{\vspace*{-0.045in}\hbox{\abelianactiontwo}}\resizebox{7pt}{27pt}{$\Biggr\rangle$}
\end{equation*}
\vspace*{-0.2in}
\caption{
    The action of a group element $\mathbf{x}\in\Z_2^n$ on the global input state.
}
\end{subfigure}

\vspace*{0.0in}
\begin{subfigure}[t]{\linewidth}
\begin{equation*}
    \resizebox{3pt}{27pt}{$\Biggr\vert$}\vcenter{\vspace*{-0.045in}\hbox{\abelianactionzero}}\resizebox{7pt}{27pt}{$\Biggr\rangle$}
    \vcenter{\vspace*{-0.04in}\hbox{\;=\;}}\resizebox{3pt}{27pt}{$\Biggr\vert$}\vcenter{\vspace*{-0.045in}\hbox{\abelianactionthree}}\resizebox{7pt}{27pt}{$\Biggr\rangle$}\vcenter{\vspace*{-0.04in}\hbox{\;=\;}}\resizebox{3pt}{27pt}{$\Biggr\vert$}\vcenter{\vspace*{-0.045in}\hbox{\abelianactionfour}}\resizebox{7pt}{27pt}{$\Biggr\rangle$}\vcenter{\vspace*{-0.04in}\hbox{\;=\;}}\resizebox{3pt}{27pt}{$\Biggr\vert$}\vcenter{\vspace*{-0.045in}\hbox{\abelianactionfive}}\resizebox{7pt}{27pt}{$\Biggr\rangle$}
\end{equation*}
\vspace*{-0.2in}
\caption{
    The subgroup operations preserving the global input state.
}
\end{subfigure}
\caption{The Abelian group action of the group $G=\Z_2^n$ on the $tn$-qubit global state $\ket{\psi}^{\otimes t}$ which underlies our hidden cut algorithm, depicted for $n=8$ and $t=4$. To each hidden cut $C\subset[n]$ corresponds an order-four hidden subgroup $H_C \simeq \Z_2^2$ which leaves the state invariant. The subgroup is generated by the two equivalent $n$-bit strings which encode the cut and its complement: $H_C=\langle1^C0^{\overline{C}},\;0^C1^{\overline{C}}\rangle$.}\label{fig:AbelianGroupAction}
\end{figure}
In order to describe the hidden cut problem as a StateHSP, we must first find an appropriate group action on the global state $\ket{\psi}^{\otimes t}$ such that each hidden cut corresponds to a specific subgroup. In the hidden cut problem we make no assumption about the internal structure of the factor states $\ket{\phi_{1,2}}$ beyond high entanglement; theferore, the symmetries of the problem lie in acting across, not along, the copies of the input state. Any $t$-fold global state $\ket{\psi}^{\otimes t}$ has a trivial permutational symmetry group $\bbS_t$ which permutes the $t$ copies. However, since each copy is internally separable along a hidden cut $C\subset[n]$, there is a larger permutational symmetry group which leaves the global state invariant. In particular, permutations which act simultaneously on all qubits within each side of the cut also preserve the global state (see \Cref{fig:example}). This would formulate the hidden cut problem as a StateHSP over the parent group $G=(\bbS_t)^{\times n}$, such that the hidden subgroups are promised to be isomorphic to $(\bbS_t)^{\times 2}$. This would seem to be the most general set of permutational symmetries of the global state.
Furthermore, in the standard HSP, Fourier sampling is efficiently implementable given known circuits for the non-Abelian quantum Fourier transform on the symmetric group \cite{beals1997quantum}, so there is hope this group action could result in an algorithm for hidden cuts. However, as we will show in \Cref{sec:symmetricgroupHSP}, this Fourier sampling method fails to find the hidden cut, for similar reasons that Fourier sampling fails to solve the standard HSP over the symmetric group \cite{moore2008symmetric}. Thus, the most obvious StateHSP approach to the hidden cut problem does not work.

Our key observation is that a much simpler Abelian group action can be used to define a StateHSP for the hidden cut problem. As it turns out, it is possible to restrict the permutational symmetries to a subset isomorphic to the group $G=\Z_2^n$, by considering simple SWAPs of pairs of qubits. Concretely, we consider dividing the $t$ copies of the input state (assuming $t$ is even) into pairs; the action of the $i$-th entry of the $n$-bit string $x\in \{0,1\}^n$ is to SWAP the $i$-th qubits inside each pair (see \Cref{fig:AbelianGroupAction}). The key point is that swapping all the qubits within each side of the cut leaves the global state invariant. Therefore, to each possible hidden cut $C\subset[n]$ corresponds a hidden subgroup of order four isomorphic to $\Z_2^2$, which contains all operations acting simultaneously on all qubits on either side of the cut.
For example, if the cut is between the first and second $n/2$ qubits, the hidden subgroup is the group with elements $\{0^n,0^{n/2}1^{n/2},1^{n/2}0^{n/2},1^n\}$ under bitwise addition mod 2, because this subgroup of SWAP operations preserves the paired copies of the input state by exchanging the left/right halves of the paired states.
Therefore, this choice of group action successfully formulates the hidden cut problem as an Abelian StateHSP instance.

\subsection{Solving the Abelian StateHSP via Fourier Sampling}

Having identified the hidden cut problem as an Abelian StateHSP over $G=\Z_2^n$, it remains to show how to efficiently solve it. Recall that standard Abelian HSP instances can be efficiently solved by Fourier sampling. We will show that a generalization of Fourier sampling can be transplanted to the StateHSP problem, resulting in an efficient algorithm to find the hidden cut, or more generally to solve any Abelian StateHSP (see \Cref{fact:StateHSPFourierSamplingAmplified}). The algorithm follows a familiar Fourier sampling workflow: we first prepare an equal superposition of group elements, then apply the controlled group action to the input state(s), and finally take a Fourier transform followed by a measurement on the group register. The circuit implementation of this approach is particularly simple, with the added benefit of parallelization over the $n$ ancillary qubits which make up the $\Z_2^n$ group register (see \Cref{fig:GFouriersampling}). 

The main question we need to answer next is how the output of this state Fourier sampling circuit relates to that of the equivalent standard HSP algorithm, i.e.\  the standard hidden subgroup problem defined with the same parent group, and the same set of valid hidden subgroups.
The technical sections of this paper focus on precisely understanding the output distribution of Fourier samples arising from the hidden cut problem. A key observation is the way in which this measurement outcome distribution depends on the number of copies of the input states $t$ one uses to produce each Fourier sample. Specifically, increasing the number of copies $t$ behaves as a form of orthogonality amplification, resulting in the output distribution approaching the ``ideal'' distribution induced by the associated standard Abelian HSP. 
In other words, the ability to act coherently on multiple copies at a time makes the hidden cut problem behave more like the corresponding standard Abelian HSP. To see why, consider all the states obtained by group action from the initial input state, i.e.\  the group orbit of the initial state. A group element can either be inside the hidden subgroup (in which case it preserves the input state), or outside the hidden subgroup (in which case it does not), meaning that each distinct state in the orbit corresponds to a coset of the hidden subgroup. Acting coherently on several copies of the state at the same time exponentially suppresses the inner product between the states along the orbit of the group action. Intuitively, this effectively orthogonalizes the orbit states; the case of orthogonal coset states is precisely the regime of the standard HSP. This effect is crucial to our algorithm, because it essentially means the hidden cut problem can be reduced to an Abelian HSP, from the point of view of Fourier sampling.

Concretely, an input involving $t$ copies of a specific state $\ket{\psi}$ will induce a specific distribution $\mathrm{P}_\mathrm{StateHSP_{\psi,t}}[\mathbf{y}]$ over the measurement outcomes $\mathbf{y}\in\Z_2^n$ of the Fourier sampling circuit. In \Cref{sec:HiddenCut} we describe the technical error analysis which allows us to appropriately choose the number of state copies $t$. Specifically, this number of copies is chosen such that the output distribution of the hidden cut Fourier sampling circuit $\mathrm{P}_\mathrm{StateHSP_{\psi,t}}$ becomes negligibly close to the equivalent standard HSP outcome distribution $\mathrm{P}_\mathrm{HSP}$ in a multiplicative sense:

\begin{equation}\label{eq:relativenegligibleerror}
     \mathrm{P}_\mathrm{StateHSP_{\psi,t}}[\mathbf{y}] = \mathrm{P}_\mathrm{HSP}[\mathbf{y}] \left(1 + \negl(n)\right),\quad\text{ for all }\mathbf{y}\in\Z_2^n\,.
\end{equation}

Here, $\mathrm{P}_\mathrm{HSP}[\mathbf{y}]$ denotes the probability to obtain outcome $\mathbf{y}\in\Z_2^n$ via Fourier sampling in the associated standard HSP with the same group and subgroup specifications as our StateHSP. This associated HSP Fourier sampling distribution is particularly simple:

\begin{equation}\label{eq:SimonslikeDistribution}
        \mathrm{P}_\mathrm{HSP}[\mathbf{y}] = \left\{
        \begin{array}{lr}
            2^{-n+2}\hspace*{10pt} & \text{if }\mathbf{y}\cdot1^C0^{\overline{C}}=\mathbf{y}\cdot 0^C1^{\overline{C}}=0\;\mathrm{mod}\;2,\vspace*{10pt}\\
            0 & \text{otherwise.}
        \end{array}\right.
\end{equation}

Specifically, this means that the associated HSP distribution  $\mathrm{P}_\mathrm{HSP}$ is uniformly supported on the Boolean subspace of dimension $n-2$ which is orthogonal to the two equivalent bit strings\footnote{The notation $1^C0^{\overline{C}}$ denotes the $n$-bit string with 1's in the positions in $C\subseteq[n]$ and 0's elsewhere.} $1^C0^{\overline{C}}$, $0^C1^{\overline{C}}$ which encode the hidden cut $C$. We remark that this is a variation of the classic Simon's problem \cite{simon1997power}, in which the Fourier samples are also uniform over the subspace orthogonal to a secret string. The fact that we obtain a multiplicative error in the output distribution of Simon's algorithm means that we will never observe a string outside the orthogonal subspace when obtaining the Fourier samples.

Once the number of copies $t$ is chosen such that the hidden cut problem returns similar outcomes as the associated Simon-like HSP, the original logic of Simon's algorithm allows us to efficiently find the hidden cut: after obtaining $O(n)$ independent Fourier samples, one has collected a complete basis for the supporting subspace with high probability, from which the secret string $1^C0^{\overline{C}}$ which encodes the cut $C$ can be learned by solving a simple Boolean linear algebra problem of size $n$. Therefore, our algorithm can be viewed as an extension of Simon's algorithm to entanglement testing, since finding the hidden cut reduces to solving a Simon-like Abelian HSP over the group $G=\Z_2^n$.

As will be detailed in \Cref{sec:HiddenCut}, a straightforward application of Abelian Fourier sampling to the hidden cut problem succeeds in finding the cut, however a number $O(n^2/\epsilon^2)$ of copies are required. A key observation is that it is possible to further reduce the requirement by a factor of $n$, down to the optimal $O(n/\epsilon^2)$ asymptotic of \Cref{thm:main}, by an adaptive modification of the Fourier sampling procedure. Specifically, we show how this can be achieved by changing the initializion of the ancillary group register. Whereas the standard Fourier sampling approach involves starting with a uniform superposition over all group elements, in our second adaptive algorithm we will start with a superposition over the $\Z_2^n$ elements which are orthogonal to previous samples. We will show how this choice boosts the probability of measuring new linearly independent samples, such that a number of copies $t=O(1/\epsilon^2)$ at each sampling round suffices to produce valid measurements from the cut subspace with constant success probability.

\begin{figure}
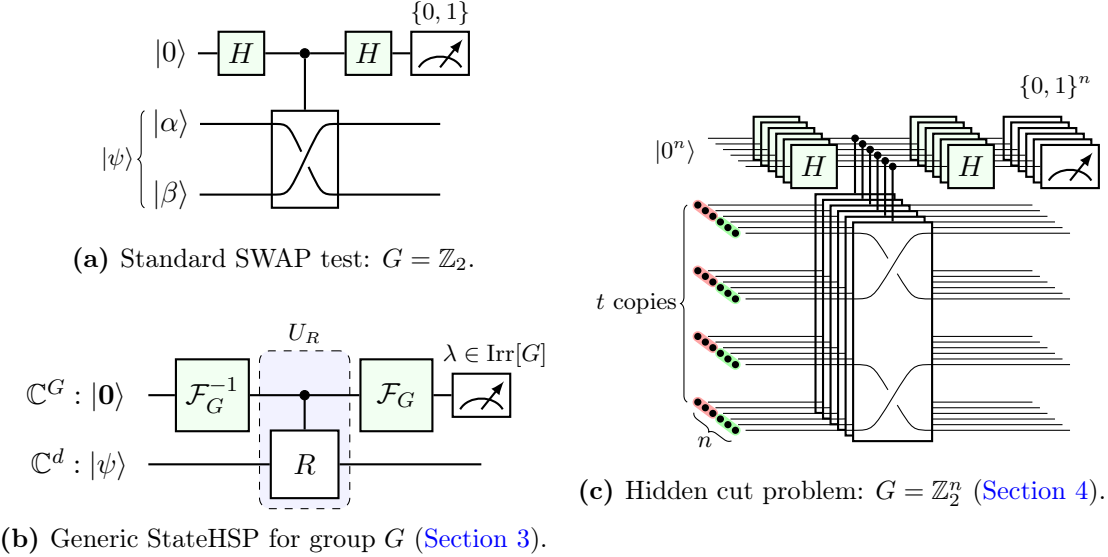

    \hspace*{0.0in}
    \centering
    \begin{minipage}[b]{0.45\textwidth}
        \begin{subfigure}[b]{\linewidth}
                \hspace*{0.46in}
                \swaptest
                \caption{Standard SWAP test: $G=\Z_2$.}\label{subfig-1:swaptest}
        \end{subfigure}\\[\baselineskip]
        \begin{subfigure}[b]{\linewidth}
            \hspace*{0.05in}
            \StateHSPFourierSampling
          \caption{Generic StateHSP for group $G$ (\Cref{sec:StateHSP}).}\label{subfig-2:StateHSP}
        \end{subfigure}
      \end{minipage}
      \begin{subfigure}{.45\linewidth}
        \hspace*{0.025in}
        \HiddenCutCircuit
      \caption{Hidden cut problem: $G=\Z_2^n$ (\Cref{sec:HiddenCut}).}\label{subfig-3:HiddenCutCircuit}
      \vspace*{0.25in}
    \end{subfigure}
    \caption{
        Group Fourier sampling circuits for specific cases of the state hidden subgroup problem. In all the examples, the top ancillary register supports a regular representation of the group, and the central operation is a controlled group action on the input state. {\bf (a)} The standard SWAP test. {\bf (b)} The generic StateHSP for arbitrary group $G$ admitting circuit implementations of the quantum group Fourier transform $\calF_G$, and which acts by a unitary representation $R$ on the input state. The top register is initialized in the basis state associated with the trivial representation of $G$. The outcome is a label $\lambda \in \Irr{G}$ of a group irreducible representation. {\bf (c)} Schematic of the circuit which solves the hidden cut problem as part of the non-adaptive \Cref{alg:main}.
   }
   \label{fig:GFouriersampling}
\end{figure}

Different promises on the internal entanglement of the factor states will ultimately impact the number of state copies $t$ required for orthogonality amplification. Following the formulation common to property testing scenarios, the generic version of the hidden cut problem (\Cref{def:HiddenCut}) promises that the factor states are at least a constant trace distance $\epsilon$ away from separable. This will require a number of copies $t=O(1/\epsilon^2)$ per Fourier sample in order to exhaustively suppress the contributions coming from all of the possible false internal cuts. This count can be reduced to a constant of only $t=2$ if the promise is strengthened to Haar-random factor states; furthermore, the corresponding circuits require only a constant depth. This improvement requires several changes to the analysis specific to the special case of Haar-random factor states (\Cref{thm:haar}), which we will detail in \Cref{sec:Haar}. First, we show that in this case the Fourier sampling distribution self-averages in a particularly strong sense. Second, we relax the strong multiplicative error condition mentioned above, and generalize Simon's algorithm to a setting which no longer involves uniform samples from the orthogonal subspace, but is skewed towards lower-weight strings. We will show that the Simon-like ``basis coupon collection'' process via Fourier sampling nonetheless succeeds to find the hidden cut under this modification.

We also note this algorithm directly generalizes to the multicut case --- as this simply corresponds to larger Abelian subgroups of this same group action, where the subgroup is generated by $1^{C_i}0^{[n]\setminus C_i}$ for any sub-partition of the qubits $C_i$. The main challenge again is to carefully keep track of the errors in the Fourier sampling distribution in this more general setting.
One can also observe our algorithm does not require knowing the number of cuts in advance, as this can be efficiently inferred from the linear algebra of the obtained Fourier samples.

We remark that one can interpret our algorithm as a combinatorial generalization of the standard SWAP test, the canonical primitive for state comparison in property testing. The SWAP test is indeed a simple instance of StateHSP for the group $G=\Z_2$ which acts as an exchange operation, with Fourier sampling implementing the projective measurement against the symmetric and antisymmetric subspaces. Our own Fourier sampling circuit for the hidden cut problem consists of $n$ parallel amplified SWAP tests which are only entangled through the internal structure of the input state (see \Cref{fig:GFouriersampling}).

\subsection{Outline of the paper}

\Cref{sec:prelims} contains basic preliminaries, as well as the exponential-time algorithm which solves the hidden cut problem using $O(n/\epsilon^2)$ copies.
In \Cref{sec:StateHSP}, we define the state version of the hidden subgroup problem. We adapt the Fourier sampling algorithm to the state problem, and describe a setting in which the state version and the standard version of the hidden subgroup problem produce similar outcomes. In \Cref{sec:HiddenCut}, we describe the efficient algorithm for the hidden cut problem and prove \Cref{thm:main}. By taking advantage of the permutational symmetries of the global state, we design an Abelian group action which fits into the StateHSP framework of \Cref{sec:StateHSP}, and show how the Fourier sampling outcomes concentrate towards a version of Simon's algorithm. We will first describe a non-adaptive, Simon-like Fourier sampling algorithm (\Cref{alg:main}) which finds the cut given $O(n^2/\epsilon^2)$ state copies. Subsequently, we will introduce an adaptive modification of the algorithm (\Cref{alg:adaptive}) and show how this decreases the state copy requirement to the optimal value of $O(n/\epsilon^2)$. \Cref{sec:Haar} is dedicated to the special case of Haar-random states, which will require a more in-depth technical analysis. Specifically, to find the cut with constant-depth circuits in the Haar-random case (\Cref{thm:haar}), we will show a self-averaging property of the Fourier sampling distribution, by approximately diagonalizing the covariance matrix of internal purities of Haar-random states; additionally, this special case requires a modification of Simon's algorithm to allow non-uniform samples. In \Cref{sec:HiddenManyCut} we generalize our results to the ``hidden many-cut problem'', showing how the same algorithm can solve not just for a single bipartition, but can similarly identify arbitrary product state structures. Finally, in \Cref{sec:Discussion} we discuss open questions and possible applications of the hidden cut problem and StateHSP to cryptography and pseudorandomness.

\section{Preliminaries} \label{sec:prelims}

We start by collecting a few basic notions about the geometry of quantum states. For more background, we refer readers to a standard reference such as \cite{nielsen2010quantum}.

\subsection{Distances}

\begin{definition}[Overlap]
    The \emph{overlap} of two pure states $\ket{\psi},\ket{\phi} \in \C^d$ is given by $\abs{\ip{\psi}{\phi}}^2$. 
\end{definition}

\begin{definition}[Trace distance]
    The \emph{trace distance} between two mixed states $\rho, \sigma \in \C^{d \times d}$ is given by
    \begin{equation*}
    \dtr(\rho, \sigma) \coloneqq \frac{1}{2}  \norm{\rho - \sigma}_1,
    \end{equation*}
    where $\norm{\cdot}_1$ denotes the trace norm, also known as the Schatten 1-norm. If $\rho = \ketbra{\psi}$ and $\sigma = \ketbra{\phi}$ are both pure states, then
    \begin{equation}\label{eq:trace-for-product}
        \dtr(\rho, \sigma) =\sqrt{1 - \abs{\ip{\psi}{\phi}}^2}.
    \end{equation}
\end{definition}

\begin{definition}[Distance from a subset]
    Let $\calH$ be a Hilbert space,
    and let $\calP$ be a subset of the pure states in $\calH$.
    Then $\ket{\psi}$ is \emph{$\epsilon$-far from $\calP$} if
    \begin{equation*}
        \dtr(\ketbra{\psi}, \ketbra{\phi}) \geq \epsilon
    \end{equation*}
    for all states $\ket{\phi}$ in $\calP$.
    Via \eqref{eq:trace-for-product}, this is equivalent to
    \begin{equation*}
        \abs{\ip{\psi}{\phi}}^2 \leq 1 - \epsilon^2
    \end{equation*}
    for all states $\ket{\phi}$ in $\calP$.
\end{definition}

\subsection{Product states}

\begin{definition}[Product states]
Let $\calH_A$ and $\calH_B$ be Hilbert spaces.
Then a \emph{product state on $\calH_A \otimes \calH_B$} is a state of the form $\ket{a}_A \otimes \ket{b}_B$.
If the bipartition of the overall Hilbert space $\calH = \calH_A \otimes \calH_B$ is clear from context, we will usually simply refer to $\ket{\psi}$ as a product state.
\end{definition}

Although not every state is a product state,
every state can be written as a superposition of product states
which are orthogonal on both their $A$ and $B$ registers.
This is given by the Schmidt decomposition.

\begin{definition}[Schmidt decomposition]
  Let $\ket{\psi}_{AB} \in \calH_A \otimes \calH_B$ be a bipartite quantitum state.
  Suppose $\calH_A$ and $\calH_B$ have dimensions $d_A$ and $d_B$, respectively, and write $r \coloneqq \min\{d_A, d_B\}$.
  The \emph{Schmidt decomposition of $\ket{\psi}$} is given by
  \begin{equation*}
      \ket{\psi}_{AB} = \sum_{i=1}^r \sqrt{\lambda_i} \cdot \ket{u_i}_A \otimes \ket{v_i}_B,
  \end{equation*}
  where (i) $\lambda_1, \ldots, \lambda_r$ are nonnegative real numbers such that $\lambda_1 + \cdots + \lambda_r = 1$,
  (ii) $\ket{u_1}, \ldots, \ket{u_r}$ are orthonormal vectors in $\calH_A$,
  and (iii) $\ket{v_1}, \ldots, \ket{v_r}$ are orthonormal vectors in $\calH_B$.
  The numbers $\lambda_1, \ldots, \lambda_r$ are known as $\ket{\psi}$'s \emph{Schmidt coefficients}.
\end{definition}

Thus, $\ket{\psi}$ is a product state if and only if its largest Schmidt coefficient is equal to $1$ and all other Schmidt coefficients are equal to $0$.
The next lemma shows a robust version of this statement,
namely that $\ket{\psi}$'s maximum Schmidt coefficient is exactly its largest overlap with any product state.

\begin{proposition}\label{prop:schmidt-overlap}
    Suppose $\ket{\psi}_{AB} \in \calH_A \otimes \calH_B$
    has Schmidt coefficients $\lambda_1 \geq \cdots \geq \lambda_r$. Then $\ket{\psi}$'s maximum squared overlap with any product state is equal to $\lambda_1$.
\end{proposition}
\begin{proof}
    First, we show that $\ket{\psi}$'s maximum overlap with any product state is at least $\lambda_1$.
    Consider the product state $\ket{u_1}_A \otimes \ket{v_1}_B$. Then
    \begin{equation*}
        \abs{\bra{u_1}_A\otimes \bra{v_1}_B \cdot \ket{\psi}}^2
        = \abs{\bra{u_1}_A\otimes \bra{v_1}_B \cdot \Big(\sum_{i=1}^r \sqrt{\lambda_i} \cdot \ket{u_i}_A \otimes \ket{v_i}_B\Big)}^2
        = \abs{\sqrt{\lambda_1}}^2
        = \lambda_1.
    \end{equation*}
    Next, we show that $\ket{\psi}$'s maximum overlap with any product state is at most $\lambda_1$.
    Let $\ket{a}_A \otimes \ket{b}_B$ be a product state.
    Then
    \begin{align}
        \abs{\bra{a}_A\otimes \bra{b}_B \cdot \ket{\psi}}^2
        &= \abs{\bra{a}_A\otimes \bra{b}_B \cdot \Big(\sum_{i=1}^r \sqrt{\lambda_i} \cdot \ket{u_i}_A \otimes \ket{v_i}_B\Big)}^2\nonumber\\
        &= \abs{\sum_{i=1}^r \sqrt{\lambda_i} \cdot \ip{a}{u_i} \cdot \ip{b}{v_i}}^2\nonumber\\
        &\leq \Big(\sum_{i=1}^r \lambda_i \cdot \abs{\ip{a}{u_i}}^2\Big) \cdot \Big(\sum_{i=1}^r \abs{\ip{b}{v_i}}^2\Big), \label{eq:just-used-cauchy-schwarz}
    \end{align}
    where the last step used the Cauchy-Schwarz inequality.
    Because $\ket{u_1}, \ldots, \ket{u_r}$ are orthonormal
    and $\ket{v_1}, \ldots, \ket{v_r}$ are orthonormal,
    we have that
    \begin{equation*}
        \abs{\ip{a}{u_1}}^2 + \cdots + \abs{\ip{a}{u_r}}^2 \leq 1
        \qquad \text{and} \qquad
        \abs{\ip{b}{v_1}}^2 + \cdots + \abs{\ip{b}{v_r}}^2 \leq 1.
    \end{equation*}
    Plugging this into \eqref{eq:just-used-cauchy-schwarz}, we get that
    \begin{equation*}
        \eqref{eq:just-used-cauchy-schwarz}
        \leq \sum_{i=1}^r \lambda_i \cdot \abs{\ip{a}{u_i}}^2
        \leq \sum_{i=1}^r \lambda_1 \cdot \abs{\ip{a}{u_i}}^2
        \leq \lambda_1.
    \end{equation*}
    This completes the proof.
\end{proof}

Combining \Cref{prop:schmidt-overlap} with \eqref{eq:trace-for-product} gives us the following immediate corollary.

\begin{corollary}\label{cor:not-product-small-schmidt}
Suppose $\ket{\psi}_{AB} \in \calH_A \otimes \calH_B$ is $\epsilon$-far from any product state.
Then its maximum Schmidt coefficient is at most $1 - \epsilon^2$.
\end{corollary}

An equivalent characterization of product states
is that $\ket{\psi}_{AB}$ is product if and only  if the reduced density matrix $\psi_A$ is a pure state.
The purity is an analytic measure for how pure a density matrix is.

\begin{definition}[Purity]
    Given a mixed state $\rho \in \C^{d \times d}$, its \emph{purity} is the quantity $\Tr(\rho^2)$.
\end{definition}

The following proposition shows that the purity of $\psi_A$ can be used as a measure for how close $\ket{\psi}_{AB}$ is to being a product state.

\begin{proposition}\label{prop:robust-purity}
    Suppose $\ket{\psi}_{AB} \in \calH_A \otimes \calH_B$ is $\epsilon$-far from any product state.
    Then the purity of $\psi_A$ is at most
    \begin{equation*}
        \Tr(\psi_A^2) \leq 1 - \epsilon^2.
    \end{equation*}
\end{proposition}
\begin{proof}
    Write $\lambda_1 \geq \cdots \geq \lambda_r$ for the Schmidt coefficients of $\ket{\psi}_{AB}$.
    Then \Cref{cor:not-product-small-schmidt} implies that
    $\lambda_1 \leq 1 - \epsilon^2$.
    Thus, we can bound the purity of $\psi_A$ by
    \begin{equation*}
        \Tr(\psi_A^2)
        = \sum_{i=1}^r \lambda_i^2
        \leq \sum_{i=1}^r \lambda_1 \cdot\lambda_i
        = \lambda_1 \cdot \Big(\sum_{i=1}^r \lambda_i\Big)
        = \lambda_1
        \leq 1 - \epsilon^2. \qedhere
    \end{equation*}
    This completes the proof.
\end{proof}

\subsection{Purity testing}

Given a mixed state
\begin{equation*}
    \rho = \sum_{i=1}^d \alpha_i \cdot \ketbra{v_i},
\end{equation*}
testing if it is actually a pure state
is impossible with only one copy of $\rho$
because no matter how far from pure $\rho$ is,
a single copy of it can be always viewed as a mixture over pure states.
It turns out, however, that \emph{two} copies of $\rho$, i.e.
\begin{equation}\label{eq:two-copies}
    \rho^{\otimes 2} = \sum_{i=1}^d \sum_{j=1}^d \alpha_i\alpha_j \cdot \ketbra{v_i} \otimes \ketbra{v_j},
\end{equation}
suffice for this task,
because if $\rho$ is not a pure state,
this mixture will contain nonzero weight on terms $\ketbra{v_i} \otimes \ketbra{v_j}$
for which $i \neq j$,
consisting of two orthogonal pure states.
We need only be able to detect when two pure states are orthogonal rather than equal,
and this can be done via the well-known \emph{SWAP test} procedure \cite{gottesman2001quantum}.

The most basic component of the SWAP test is the SWAP gate:

\begin{definition}[The SWAP gate]
Let $d$ be an integer.
The \emph{SWAP gate} is the unitary operator $\swapjohn$ acting on $\C^d \otimes \C^d$ such that
\begin{equation*}
    \swapjohn \cdot \ket{i} \otimes \ket{j} = \ket{j} \otimes \ket{i},
\end{equation*}
for all $1\leq i, j \leq d$.
\end{definition}

Then the SWAP test acts as follows.

\begin{definition}[The SWAP test]
    Let $\rho_{A B}$ be a mixed state in $\C^d \otimes \C^d$
    (which will typically be a tensor product state $\rho_A \otimes \sigma_{B}$).
    The \emph{SWAP test} is the quantum algorithm which acts as follows.
    Beginning with the input state $\rho_{A B}$, (i) append an ancilla qubit in the $\ket{+}_{\anc}$ state.
        Next, (ii) apply $\swapjohn_{AB}$ conditioned on the ancilla qubit and then (iii) Hadamard the ancilla qubit.
        Finally, (iv) measure the ancilla qubit and accept if the outcome is ``0''.
        An illustration of the SWAP test applied to a product state $\ket{\alpha} \otimes \ket{\beta}$ can be found in \Cref{subfig-1:swaptest}.
\end{definition}

The SWAP test can equivalently be stated in terms of a two-outcome projective measurement.

\begin{proposition}[SWAP test, projector version]\label{prop:swap-is-projector}
    Write $\Pi_{\swapjohn}$ for the projector $\tfrac{1}{2}(I_{AB} + \swapjohn_{AB})$. Then the SWAP test implements the projective measurement $\{\Pi_{\swapjohn}, I - \Pi_{\swapjohn})$.
\end{proposition}
\begin{proof}
    Given an input state $\ket{\psi}_{AB}$, the SWAP test acts as follows.
    \begin{align*}
        \ket{\psi}_{AB}
        & \longrightarrow \ket{+}_{\anc} \otimes \ket{\psi}_{AB} \tag{append the ancilla}\\
        & \longrightarrow \frac{1}{\sqrt{2}} \cdot\ket{0}_{\anc} \otimes \ket{\psi}_{AB} + \frac{1}{\sqrt{2}}\cdot \ket{1}_{\anc} \otimes (\swapjohn_{AB} \cdot \ket{\psi}_{AB}) \tag{apply the controlled $\swapjohn$}\\
        & \longrightarrow \frac{1}{\sqrt{2}} \cdot\ket{+}_{\anc} \otimes \ket{\psi}_{AB} + \frac{1}{\sqrt{2}}\cdot \ket{-}_{\anc} \otimes (\swapjohn_{AB} \cdot \ket{\psi}_{AB}). \tag{Hadamard the ancilla}
    \end{align*}
    This state is equal to
    \begin{align*}
    &\frac{1}{2} \cdot \ket{0}_{\anc} \otimes (\ket{\psi}_{AB} + \swapjohn_{AB} \cdot \ket{\psi}_{AB})
    + 
    \frac{1}{2} \cdot \ket{1}_{\anc} \otimes (\ket{\psi}_{AB} - \swapjohn_{AB} \cdot \ket{\psi}_{AB})\\
    ={}& \ket{0}_{\anc} \otimes (\Pi_{\swapjohn} \cdot \ket{\psi}_{AB}) + \ket{0}_{\anc} \otimes ((I-\Pi_{\swapjohn}) \cdot \ket{\psi}_{AB}),
    \end{align*}
    where here we used the fact that $I - \Pi_{\swapjohn} = \tfrac{1}{2}(I_{AB} - \swapjohn_{AB})$.
    The SWAP test concludes by measuring the ancilla in the standard basis, which concludes the proof.
\end{proof}

Hence, the probability that the SWAP test accepts on $\rho^{\otimes 2}$ is $\Tr(\Pi_{\swapjohn} \cdot \rho^{\otimes 2}) = \frac{1}{2} + \frac{1}{2} \cdot \Tr(\swapjohn \cdot \rho^{\otimes 2})$.
The next proposition computes the second term.

\begin{proposition}[Purity formula]
\begin{equation*}
    \Tr(\swapjohn \cdot \rho^{\otimes 2})
    = \Tr(\rho^2).
\end{equation*}
\end{proposition}
\begin{proof}
    Write $\rho^{\otimes 2}$ as in \eqref{eq:two-copies}.
    Then we have
    \begin{equation*}
        \Tr(\swapjohn \cdot \ketbra{v_i} \otimes \ketbra{v_j})
        = \Tr(\ketbra{v_j}{v_i} \otimes \ketbra{v_i}{v_j})
        = \abs{\ip{v_i}{v_j}}^2
        = \left\{\begin{array}{rl}
        1 & \text{if $i = j$},\\
        0 & \text{otherwise.}
        \end{array}\right.
    \end{equation*} 
    Extending via linearity,
    \begin{equation*}
        \Tr(\swapjohn \cdot \rho \otimes \rho)
        = \sum_{i=1}^d \sum_{j=1}^d \alpha_i \alpha_j \cdot 1[i= j]
        = \sum_{i=1}^d \alpha_i^2 = \Tr(\rho^2).\qedhere
    \end{equation*}
\end{proof}

Putting everything together gives the following formula for the probability the SWAP test accepts.

\begin{corollary}[SWAP test acceptance probability]\label{cor:swap-test-prob}
    The probability the SWAP test accepts on $\rho^{\otimes 2}$ is
    $\frac{1}{2} + \frac{1}{2} \cdot \Tr(\rho^2)$.
\end{corollary}
\noindent
Thus, if $\rho$ is pure,
i.e.\ its purity is 1,
then the SWAP test will always accept,
but if $\rho$ is very mixed,
i.e.\ its purity is close to 0,
then the SWAP test will accept with probability roughly $\frac{1}{2}$.

This also gives an algorithm for testing if a bipartite pure state $\ket{\psi}_{AB}$ is entangled given just two copies $\ket{\psi}_{AB} \otimes \ket{\psi}_{A'B'}$: simply run the SWAP test on the $A$ and $A'$ registers of this two-copy state,
which is equivalent to running the SWAP test on $\psi_A^{\otimes 2}$.
Doing so will accept with probability $\frac{1}{2} + \frac{1}{2} \cdot \Tr(\psi_A^2)$,
which is equal to 1 if $\ket{\psi}_{AB}$ is a product state but is at most $1 - \epsilon^2/2$ if $\ket{\psi}_{AB}$ is $\epsilon$-far from product (via \Cref{prop:robust-purity}).
This algorithm can be thought of as exploiting the fact that $\ket{\psi}_{AB} \otimes \ket{\psi}_{A'B'}$
is unchanged by applying $\swapjohn_{A A'}$ if and only if $\ket{\psi}$ is a product state.

\subsection{Multipartite product states}

\begin{notation}[$n$-qubit systems]
    Much of this paper is about $n$-qubit systems.
    Given a subset $S \subseteq [n]$ of the qubits, we will write $\overline{S} \equiv [n] \setminus S$ for the qubits outside of $S$.
    We will write $\calH_S$ for the Hilbert space consisting of the qubits in $S$,
    and so we will often write a state in $\calH_S$ as $\ket{\psi}_S$, i.e.\ with the ``$S$'' subscript.
\end{notation}

\begin{definition}[Multipartite product states]
    An $n$-qubit state $\ket{\psi}$ is a \emph{multipartite product state}
    if there exists a subset of the qubits $C \subseteq [n]$ such that $\ket{\psi}$
    can be written as $\ket{\psi} = \ket{a}_C \otimes \ket{b}_{\overline{C}}$, for some states
    $\ket{a}_C$ and $\ket{b}_{\overline{C}}$ supported on the qubits in $C$ and $\overline{C}$, respectively.
\end{definition}

We will often consider the case when $\ket{\psi} = \ket{a}_C \otimes \ket{b}_{\overline{C}}$ in which $\ket{a}_C$ and $\ket{b}_{\overline{C}}$ are both $\epsilon$-far from any multipartite product state,
and our goal is to determine $C$.
It is natural to pick a subset $S$ and test if $C = S$
by running the product test on the qubits within $S$.
To analyze this, we first show the following proposition.

\begin{proposition}[Purity across different cuts]\label{fact:epsilonpurity}
    Let $\ket{\psi} = \ket{a}_C \otimes \ket{b}_{\overline{C}}$ in which $\ket{a}_C$ and $\ket{b}_{\overline{C}}$ are both $\epsilon$-far from any multipartite product state.
    Let $S \subseteq [n]$ be a subset of the qubits. Then the purity of $\psi_S$ is $\Tr(\psi_S^2) = 1$ if $S = C$ or $\overline{C}$ and otherwise $\Tr(\psi_S^2) \leq 1 - \epsilon^2$.
\end{proposition}
\begin{proof}
    If $S = C$ then $\psi_S = \ketbra{a}$, which is a pure state, and so its purity is 1;
    a similarly argument applies when $S = \overline{C}$.
    On the other hand, when $S \neq C, \overline{C}$, we have that
    \begin{equation*}
        \psi_S = a_{C \cap S} \otimes b_{\overline{C} \cap S}.
    \end{equation*}
    Because $S \neq C, \overline{C}$,
    it must be the case that either
    $\emptyset \subsetneq C \cap S \subsetneq C$
    or $\emptyset \subsetneq \overline{C} \cap S \subsetneq \overline{C}$; let us assume without loss of generality that the former is true.
    Then because $\ket{a}_C$ is $\epsilon$-far from multiproduct,
    it is $\epsilon$-far from being a product state on the bipartition $(C \cap S, C \cap \overline{S})$. Hence, by \Cref{prop:robust-purity}, we can bound its purity by
    \begin{equation*}
        \Tr(a_{C \cap S}^2) \leq 1 - \epsilon^2.
    \end{equation*}
    As a result, we can bound the purity of the overall state by
    \begin{equation*}
        \Tr(\psi_S^2) = \Tr(a_{C \cap S}^2) \cdot \Tr(b_{\overline{C} \cap S}^2)
        \leq \Tr(a_{C \cap S}^2)
        \leq 1 - \epsilon^2.\qedhere
    \end{equation*}
\end{proof}

Combining this with \Cref{cor:swap-test-prob},
we get the following bound on the probability that the SWAP test on the qubits in $S$ accepts.

\begin{corollary}\label{cor:swap-on-s}
    Let $\ket{\psi} = \ket{a}_C \otimes \ket{b}_{\overline{C}}$ in which $\ket{a}_C$ and $\ket{b}_{\overline{C}}$ are both $\epsilon$-far from any multipartite product state.
    Suppose we are given two copies of $\ket{\psi}$
    and we run the $\swapjohn$ test on some subset $S\subseteq [n]$ of the qubits. Then if $S = C$ or $S = \overline{C}$, the $\swapjohn$ test always accepts. Otherwise, if $S \neq C, \overline{C}$,
    \begin{equation*}
        \mathbb{P}\left[\mathrm{SWAP\;test\; accepts}\right] \leq 1 - \epsilon^2/2.
    \end{equation*}
\end{corollary}

We can also amplify the probability of detecting that $\ket{\psi}$ in the $S \neq C, \overline{C}$ case of \Cref{cor:swap-on-s} by taking additional copies of $\ket{\psi}$.
In particular, suppose we have $2m$ copies of $\ket{\psi}$ and we
group them up into $m$ pairs.
If we run the SWAP test on the qubits in $S$ for each pair and accept only if all $n$ SWAP tests accept,
then the probability we accept is at most
\begin{equation*}
(1-\epsilon^2/2)^m \leq e^{- \frac{1}{2} \epsilon^2 m}.
\end{equation*}
This gives us the following proposition.

\begin{proposition}\label{prop:amplified}
    Given an integer $k$,
    there is a projective measurement $\{\Pi_S, \overline{\Pi}_S\}$ which acts as follows.
    Let $\ket{\psi} = \ket{a}_C \otimes \ket{b}_{\overline{C}}$ in which $\ket{a}_C$ and $\ket{b}_{\overline{C}}$ are both $\epsilon$-far from any multipartite product state.
    Suppose we measure $\ket{\psi}^{\otimes 2m}$
    with $\{\Pi_S, \overline{\Pi_S}\}$.
    If $S = C$ or $S = \overline{C}$, this measurement always accepts. Otherwise, if $S \neq C, \overline{C}$, the probability it accepts is at most $\mathrm{exp}(-\epsilon^2 m/2)$.
\end{proposition}

That the measurement is projective follows from the fact that it is performing $m$ SWAP tests, and each SWAP test is a projective measurement due to \Cref{prop:swap-is-projector}.

\subsection{An information theoretic algorithm for the hidden cut problem}
\label{sec:infotheoretic}

Harrow, Lin, and Montanaro~\cite{harrow2017sequential}
studied the problem of testing whether a given $n$-qubit state $\ket{\psi}$ is a multipartite product state or is $\epsilon$-far from all multipartite product states.
The key intuition is that if $\ket{\psi}$ is a multipartite product state,
then there exists an $S$ such that the measurement $\{\Pi_S, \overline{\Pi}_S\}$ accepts with probability 1.
On the other hand, if $\ket{\psi}$ is $\epsilon$-far from all multipartite product states,
then the measurement will accept with probability at most $\mathrm{exp}(-\epsilon^2 m/2)$,
which can be made smaller than, say, $100^{-n}$ by taking $m = O(n/\epsilon^2)$.
Since this is so small, we can apply their quantum OR bound to combine all $(2^n-2)$ different $\{\Pi_S, \overline{\Pi}_S\}$ measurements into a single measurement $\{Q, \overline{Q}\}$ which always accepts on multipartite product states
and accepts with exponentially small probability on states which are $\epsilon$-far from multipartite product.

We now observe that if instead of combining these measurements with the quantum OR bound,
we combine them with Gao's quantum union bound,
we get a sample-efficient algorithm for finding the cut $C$. 

\begin{fact}[Gao's quantum union bound \cite{Gao15}]
    Let $\rho$ be a density matrix.
    For each $1 \leq i \leq k$, let $\{\Pi_i, \overline{\Pi}_i\}$ be a two outcome projective measurement, and write $\mathrm{err}_i \coloneqq \Tr(\Pi_i \cdot \rho)$.
    If we measure $\rho$ with each $\{\Pi_i, \overline{\Pi}_i\}$ measurement in sequence from $i = 1$ to $k$, then the probability that we only observe the $\overline{\Pi}_i$ outcomes is at least $1 - 4 \cdot (\mathrm{err}_1 + \cdots + \mathrm{err}_k)$.
\end{fact}

\begin{theorem}[Information theoretic algorithm for the hidden cut problem]
    Let $\ket{\psi} = \ket{a}_C \otimes \ket{b}_{\overline{C}}$ in which $\ket{a}_C$ and $\ket{b}_{\overline{C}}$ are both $\epsilon$-far from any multipartite product state.
    There is an algorithm which can identify $C$ with probability at least $99\%$ given $m = O(n/\epsilon^2)$ copies of $\ket{\psi}$.
\end{theorem}
\begin{proof}
    Set $m = O(n/\epsilon^2)$ so that $\mathrm{exp}(-\epsilon^2 m / 2) \leq 100^{-n}$.
    Then for each subset $S$, \Cref{prop:amplified} gives us a projective measurement $\{\Pi_S, \overline{\Pi}_S\}$ which accepts with probability $1$ if $S = C, \overline{C}$ and with probability at most $100^{-n}$ if $S \neq C, \overline{C}$. Set $N = 2^n - 2$ and pick an arbitrary order $S_1, \ldots, S_N$ on the nontrivial subsets of~$[n]$. The algorithm is as follows:
    given $m$ copies of $\ket{\psi}$,
    perform the $N$ measurements $\{\Pi_{S_1}, \overline{\Pi}_{S_1}\}$ through $\{\Pi_{S_N}, \overline{\Pi}_{S_N}\}$ in order until the first time observing a $\Pi_{S_i}$ outcome; when this happens, output ``$S_i$'' and terminate.

    Suppose without loss of generality that $S_i$ is equal to the true cut $C$, and none of the previous $S_j$'s is equal to $\overline{C}$.
    The algorithm succeeds if the first $i-1$ measurements reject and the $i$-th measurement accepts. To compute the probability that this does \emph{not} happen, we can apply Gao's quantum union bound with $\Pi_j = \Pi_{S_j}$ for each $1 \leq j \leq i - 1$ and $\Pi_i = \overline{\Pi}_{S_i}$.
    Then $\mathrm{err}_j = \Tr(\Pi_{S_j} \cdot \rho) \leq 100^{-n}$ for each $1 \leq j \leq i-1$
    and $\mathrm{err}_i = \Tr(\overline{\Pi}_{S_i} \cdot \rho) = 0$.
    Then the union bound says that the probability the algorithm does not succeed is at most
    \begin{equation*}
        4 \cdot (\mathrm{err}_1 + \cdots + \mathrm{err}_{i-1} + \mathrm{err}_i)
        \leq 4 \cdot (i-1) \cdot \frac{1}{100^n}
        \leq \frac{4 N}{100^n} \leq 0.01.
    \end{equation*}
    This completes the proof.
\end{proof}

Note that the multipartite product state detection algorithm of Harrow, Lin, and Montanaro runs in exponential time, because it involves computing the quantum OR of an exponentially large number of measurements and then implementing that (likely computationally infeasible) measurement on $\ket{\psi}^{\otimes m}$.
Similarly, this algorithm for the hidden cut problem also requires exponential time as it performs exponentially many measurements in sequence.

\section{The State Hidden Subgroup Problem}\label{sec:StateHSP}

To produce our algorithm we introduce a \emph{quantum state} version of the hidden subgroup problem which may be of independent interest. 
Our definition is motivated by the observation that in the hidden cut problem, the input states have certain symmetries determined by the secret cut. For example, the Haar measure is invariant under arbitrary unitaries. If we instantiate the hidden cut problem with two $n/2$-qubit Haar random states separated by a random cut, this means that the input states to the hidden cut problem, when viewed as a density matrix, are invariant under the action of a group isomorphic to $\mathrm{U}(2^{n/2}) \times \mathrm{U}(2^{n/2})$ --- the issue is that we don't know the qubit bipartition which defines the specific symmetry group.

This sounds related to the well-studied Hidden Subgroup Problem (HSP) \cite[Section 5.4.3]{nielsen2010quantum}, in which one is given a function $f:G\rightarrow \{0,1\}^n$ invariant on left cosets of a subgroup $H<G$, with the goal of learning $H$:
\begin{definition}[Hidden subgroup problem (HSP)]
     A function $f: G \rightarrow L$ from a finite group $G$ to a set of discrete labels $L$ is said to {\em hide} a subgroup $H \leq G$ if it is constant on the (left) cosets of $H$, and different across the cosets, in other words $f(x)=f(y)$ if and only if $x^{-1}y\in H$.
     The hidden subgroup problem is the task of determining the hidden subgroup $H$ from as few queries to the function $f$ as possible, assuming black-box access to an oracle implementation $O_f:\ket{g}\ket{0} \mapsto \ket{g}\ket{f(g)}$.
\end{definition}
One crucial difference, however, is that the HSP takes as input a \emph{function} respecting certain subgroup symmetries, while the hidden cut problem takes as input \emph{quantum states} invariant under a certain subgroup action. Motivated by this observation, we define a quantum state version of the HSP:
\begin{customdefinition}{\ref{def:StateHSP}}[The state hidden subgroup problem (StateHSP) --- restated]
    Let $G$ be a finite group with a unitary representation $R:G\rightarrow \mathrm{U}(d)$. Let $\C^G$ denote a Hilbert space in the regular representation of $G$, meaning $\C^G=\mathrm{span}\{\ket{g}\}_{g\in G}$ where $\ip{g}{h}=\delta_{g,h}$. Assume efficient implementation of the controlled group action $U_R=\sum_{g\in G}\dyad{g}\otimes R(g)$ acting on $\C^G\otimes \C^d$. Assume access to (copies of) a quantum state $\ket{\psi} \in \C^d$ with the following properties:
\begin{itemize}[leftmargin=*]
    \item $\ket{\psi}$ is invariant under the action of a subgroup $H$, i.e.\  for all $h\in H$, $R(h)\ket{\psi} = \ket{\psi}$.
    \item $\ket{\psi}$ is acted on nontrivially by elements outside the subgroup: for any $g\notin H$, $\abs{\mel{\psi}{R(g)}{\psi}} \leq 1-\epsilon$,
\end{itemize}
where the parameter $\epsilon$, which we call the ``orthogonality allowance'', can depend on the dimension $d$, the group $G$, and the representation $R$. The goal is to identify the hidden subgroup $H$.
\end{customdefinition}
Here we are assuming one has efficient access to the group representation --- i.e.\  given $g\in G$, one can apply $R(g)$ efficiently. In this sense the problem is similar to the black box group model of computing (e.g.\ employed in \cite{watrous2000succinct}), but defined with respect to a representation of the group other than the left regular representation. Our definition can also be viewed as inspired by recent works studying quantum state/unitary variants of complexity classes and cryptography such as \cite{bostanci2023unitary,rosenthal2021interactive,lombardi2024one,zhandry2023quantum}. We emphasize that the value of $\epsilon$ could vary substantially between different representations, so the scaling behavior of $\epsilon$ might significantly affect the difficulty of this problem. Additionally, this framework can accommodate problems in which either the parent group $G$, or the representation $R$ and its dimension $d$, or possibly both, can depend on the specific underlying parameter of the problem. Multiple variations can be imagined, such as an additional promise that the hidden subgroups are mutually conjugate (a common HSP flavor), or introducing unknowns about the specific group representation $R$.

We note that a recent set of works has studied the problem of determining whether a given quantum state is preserved by a known, specific group action \cite{laborde2023testing, rethinasamy2023quantum}. This particular property testing task can be seen as a special case of decisional StateHSP, in which the role of the hidden subgroup $H$ is played by the parent group $G$ itself. Finally, we notice a connection to the problem of {\em state isomorphism} \cite{lockhart2017quantum}, which asks whether two input states can be made equal under a permutation of the qubits.

\subsection{Coset states and the standard HSP approach}

Despite the syntactic differences between the HSP and the StateHSP, there is a sense in which the HSP can be viewed as a special case of the StateHSP.
The ``standard method'' for the HSP~\cite{GSVV04} involves preparing a uniform superposition over the elements of $G$ and feeding it into the $f$ oracle, resulting in the state
\begin{equation}\label{eq:already-state-HSP}
\ket{\psi} = \frac{1}{\sqrt{|G|}} \cdot \sum_{x \in G} |x\rangle_G \otimes \ket{f(x)}_L.
\end{equation}
Next, one discards the label register, resulting in a uniform mixture of coset states
\begin{equation*}
\frac{1}{|G|} \cdot \sum_{g \in G} \ketbra{gH}, \qquad \text{where } \ket{gH}\coloneqq \frac{1}{\sqrt{|H|}} \cdot \sum_{h \in H} \ket{gh}
\text{ is a coset state.}
\end{equation*}
One can then repeatedly run this procedure to generate multiple coset states, and the task is to use them to learn $H$.

However, even before discarding the label register,
the state in Equation~\eqref{eq:already-state-HSP} is already an instance of the StateHSP.
In particular, let $R : G \rightarrow \mathrm{U}(d)$ be the \emph{right regular representation} of $G$,
meaning that it acts on $\C^G$ via $R(g) \cdot \ket{x} = \ket{x g^{-1}}$.
Suppose $f$ hides the subgroup $H$.
Then $\ket{\psi}$ is invariant under the action of $H$, because for any $h \in H$,
\begin{align*}
R(h)_G \cdot \ket{\psi}
&= \frac{1}{\sqrt{|G|}}\cdot \sum_{x \in G} |x h^{-1}\rangle_G \otimes \ket{f(x)}_L\\
&= \frac{1}{\sqrt{|G|}}\cdot \sum_{y \in G} |y\rangle_G \otimes \ket{f(yh)}_L
= \frac{1}{\sqrt{|G|}}\cdot \sum_{y \in G} |y\rangle_G \otimes \ket{f(y)}_L
= \ket{\psi},
\end{align*}
where we used the fact that $f(y) = f(yh)$ because $y^{-1} y h = h \in H$.
On the other hand, for any $g \notin H$,
\begin{equation*}
R(g)_G \cdot \ket{\psi}
= \frac{1}{\sqrt{|G|}}\cdot \sum_{x \in G} |x g^{-1}\rangle_G \otimes \ket{f(x)}_L
= \frac{1}{\sqrt{|G|}}\cdot \sum_{y \in G} |y\rangle_G \otimes \ket{f(y g)}_L.
\end{equation*}
But $f(y) \neq f(yg)$ because $y^{-1} y g = g \notin H$. 
This means that for all $g \notin H$,
\begin{equation*}
\bra{\psi} R(g)_G  \ket{\psi} = 0,
\end{equation*}
and so this state satisfies the definition of the StateHSP with an orthogonality allowance of $\epsilon = 1$.
In general, we will see that instances of the StateHSP where $\epsilon$ is close to 1 act like instances of the traditional HSP, which we can sometimes solve efficiently.

\subsection{Fourier sampling in HSP vs.\  StateHSP}

We begin by reviewing the Fourier sampling approach to solving the standard hidden subgroup problem, which we will proceed to generalize to the StateHSP setting. The construction at its core is the group Fourier transform:

\begin{definition}[Group Fourier transform \cite{diaconis1988group}]
Let $G$ be a finite group (not necessarily Abelian), and let $\{\rho_\lambda\}_{\lambda \in \Irr{G}}$ be a full set of irreducible unitary $G$-representations (irreps), such that each $\rho_\lambda:G\to\U(d_\lambda)$ is a unitary irrep of $G$ of dimension $d_\lambda$. Then, the group Fourier transform is the $\abs{G}\times\abs{G}$ unitary $\calF_G$ which transforms from the regular representation basis $\{\ket{g}\}_{g\in G}$ to a Schur basis $\{\ket{\lambda, i, j}\}_{\substack{\lambda \in \Irr{G}\\i,j \in [d_\lambda]}}$. Explicitly:
\begin{align}\label{eq:GFourierTransform}
    \calF_G \ket{g} = \sum_{\substack{\lambda \in \Irr{G}\\i,j \in [d_\lambda]}} \sqrt\frac{d_\lambda}{\abs{G}} \rho_\lambda(g)_{i,j}\ket{\lambda,i,j}, && \calF^{-1}_G\ket{\lambda,i,j} = \sqrt\frac{d_\lambda}{\abs{G}}\sum_{g\in G}\rho_\lambda(g^{-1})_{j,i}\ket{g}\,.
\end{align}
\end{definition}

The Fourier sampling approach to the generic HSP becomes possible when there is an efficient circuit for the group Fourier transform. This is usually a safe assumption if the group is Abelian; efficient circuits for the non-Abelian quantum Fourier transform are known for several important groups including the dihedral and symmetric groups, but other groups are conjectured not to admit efficient QFT circuits \cite{moore2006generic}. Given an HSP with hidden subgroup $H<G$, the measurement outcome of a so-called ``weak'' Fourier sampling circuit are samples from a distribution over the irrep labels  $\lambda\in\Irr{G}$:

\begin{fact}[Weak Fourier sampling in HSP \cite{hallgren2003hidden}]\label{fact:HSPFourierSampling} In an HSP over the parent group $G$ with a hidden subgroup $H<G$, the weak Fourier sampling outcome distribution over the labels $\lambda \in \Irr{G}$ is given by:
\begin{equation}\label{eq:HSPFourierDistr}
    \mathrm{P}_{\mathrm{HSP}}[\lambda] = \frac{d_\lambda}{\abs{G}} \sum_{h\in H} \chi_\lambda(h)\,,
\end{equation}
where $\chi_\lambda$ denotes the corresponding irreducible character of $G$, i.e.\  $\chi_\lambda(g)=\Tr\rho_\lambda(g)$. 
\end{fact}

\begin{proof}
    In the HSP weak Fourier sampling setting, one starts with an arbitrary {\em coset state} in the regular representation (which can be efficiently prepared from oracle access to the input function):
    \begin{equation}
        \ket{gH} = \frac{1}{\sqrt{\abs{H}}}\sum_{h\in H}\ket{gh},
    \end{equation}
    to which the $G$-Fourier transform $\calF_G$ \eqref{eq:GFourierTransform} is applied, leading to the state:
    \begin{equation}
        \calF_G\ket{gH} = \frac{1}{\sqrt{\abs{H}}}\sum_{h\in H}\sum_{\substack{\lambda \in \Irr{G}\\i,j \in [d_\lambda]}}\sqrt\frac{d_\lambda}{\abs{G}}\rho_\lambda(gh)_{i,j}\ket{\lambda,i,j}\,.
    \end{equation}
    Finally, only the irrep label register is measured, leading to an output probability:
    \begin{align}
        \mathrm{P}_{\mathrm{HSP}}[\lambda] &= \frac{d_\lambda}{\abs{G}\abs{H}}\sum_{\substack{h,h'\in H\\i,j\in[d_\lambda]}}\rho_\lambda(gh)_{i,j}\overline{\rho_\lambda(gh')_{i,j}} & \text{(via Born rule)}\\
        &=\frac{d_\lambda}{\abs{G}\abs{H}}\sum_{\substack{h,h'\in H\\i,j\in[d_\lambda]}}\rho_\lambda(gh)_{i,j}\rho_\lambda(h'^{-1}g^{-1})_{j,i} & \text{(since $\rho_\lambda(gh')$ is unitary)}\\
        &= \frac{d_\lambda}{\abs{G}\abs{H}}\sum_{h,h'\in H}\Tr\rho_\lambda(hh'^{-1}) & \text{(by cyclic property of trace)}\\
        &= \frac{d_\lambda}{\abs{G}}\sum_{h\in H}\chi_\lambda(h) & \text{(by double-summing over the subgroup)}.
    \end{align}
    See \cite{hallgren2003hidden} for more detail.
\end{proof}

A natural question is understanding how this distribution changes when we generalize Fourier sampling to a StateHSP problem --- specifically, by applying the circuit in \Cref{subfig-2:StateHSP}:
\begin{fact}[Weak Fourier sampling in StateHSP]\label{fact:StateHSPFourierSampling}
    The output of the weak Fourier sampling circuit in \Cref{subfig-2:StateHSP} for a StateHSP problem in which the input state $\ket{\psi}$ is invariant under the $R$-action of the hidden subgroup $H<G$ is:
    \begin{equation}\label{eq:StateHSPFourierDistr}
        \mathrm{P}_{\mathrm{StateHSP}_{\psi}}[\lambda] = \frac{d_\lambda}{\abs{G}}\sum_{c \in G/H} \sum_{h\in H} \chi_\lambda(ch) \,\mel{\psi}{R(c)}{\psi}\,,
    \end{equation}
    where $\lambda\in\Irr{G}$ is a $G$-irrep label, and $G/H$ denotes a set of left coset representatives.
\end{fact}

\begin{proof} The result follows from direct calculation along similar lines as \Cref{fact:HSPFourierSampling}. In the StateHSP setting, there is an ancillary register admitting a regular representation of the group $G$ (initialized in the trivial representation), together with the input state $\ket{\psi}$. The first application of the inverse $G$-Fourier transform prepares a uniform superposition over group elements in the first register:
    \begin{equation}
        \left(\calF_G^{-1}\otimes I\right)\ket{\mathbf{0}}\ket{\psi} = \frac{1}{\sqrt{\abs{G}}}\sum_{g\in G} \ket{g}\ket{\psi}.
    \end{equation}
    Applying the controlled group action $U_G = \sum_{g\in G}\dyad{g}\otimes R(g)$ leads to the state:
    \begin{equation}
        U_G\left(\calF_G^{-1}\otimes I\right)\ket{\mathbf{0}}\ket{\psi} = \frac{1}{\sqrt{\abs{G}}}\sum_{g\in G} \ket{g}\otimes R(g)\ket{\psi}.
    \end{equation}
    Finally, the last application of the $G$-Fourier transform results in the state:
    \begin{equation}
        \left(\calF_G\otimes I\right)U_G\left(\calF_G^{-1}\otimes I\right)\ket{\mathbf{0}}\ket{\psi} = \frac{1}{\abs{G}}\sum_{\substack{\lambda \in \Irr{G}\\i,j\in[d_\lambda]}}\sqrt{d_\lambda}  \ket{\lambda,i,j}\otimes \sum_{g\in G} \rho_{\lambda}(g)_{i,j} R(g)\ket{\psi}.
    \end{equation}
    On this state, we measure the ancillary register corresponding to the irrep label $\lambda\in\Irr{G}$, leading to the output probability:
    \begin{equation*}
    \def\arraystretch{2.2}
    \begin{array}{rlr}
        \mathrm{P}_{\mathrm{StateHSP}_\psi}[\lambda] &= \frac{d_\lambda}{\abs{G}^2}\sum\limits_{\substack{i,j\in[d_\lambda]\\g,g'\in G}}\rho_\lambda(g)_{i,j}\overline{\rho_\lambda(g')_{i,j}}\mel{\psi}{R(g')^{\dagger}R(g)}{\psi} & \text{(via Born rule)}\\
        &= \frac{d_\lambda}{\abs{G}^2}\sum\limits_{g,g'\in G} \Tr\rho_\lambda(g'^{-1}g)\mel{\psi}{R(g'^{-1}g)}{\psi} & \text{(by unitarity of $\rho_\lambda$ and $R$)}\\
        &= \frac{d_\lambda}{\abs{G}}\sum\limits_{g\in G}\chi_\lambda(g)\mel{\psi}{R(g)}{\psi} & \text{(after double-summing over the group $G$).}
    \end{array}
    \end{equation*}
    Splitting the group elements over the right-$H$-cosets as $g=ch$, where $h\in H$ are subgroup elements and $c\in G/H$ are coset representatives, we obtain:
    \begin{equation*}
    \def\arraystretch{2.2}
    \begin{array}{rlr}
        \mathrm{P}_{\mathrm{StateHSP}_\psi}[\lambda] &= \frac{d_\lambda}{\abs{G}}\sum\limits_{\substack{c\in G/H \\ h \in H}}\chi_\lambda(ch)\mel{\psi}{R(ch)}{\psi} & \\
        &= \frac{d_\lambda}{\abs{G}}\sum\limits_{\substack{c\in G/H \\ h \in H}}\chi_\lambda(ch)\mel{\psi}{R(c)R(h)}{\psi} & \text{(since $R$ is a $G$-representation)}\\
        &= \frac{d_\lambda}{\abs{G}}\sum\limits_{\substack{c\in G/H \\ h \in H}}\chi_\lambda(ch)\mel{\psi}{R(c)}{\psi} & \text{(by the $H$-invariance of $\ket\psi$)},
    \end{array}
    \end{equation*}
    where in the last line we used the subgroup invariance assumption about the input state:  $R(h)\ket{\psi}=\ket{\psi}$ for all $h \in H$.
\end{proof}

Notice that the StateHSP outcome distribution \eqref{eq:StateHSPFourierDistr} and the equivalent HSP outcome distribution \eqref{eq:HSPFourierDistr} are identical when $\mel{\psi}{R(c)}{\psi}=0$ for all nontrivial coset representatives $c\neq \mathrm{id}$. The analogy is explained by the fact that in StateHSP, the states $\{R(c)\ket{\psi}\}_{c\in G/H}$ play a similar role to the {\em coset states} $\ket{cH}=\frac{1}{\sqrt{\abs{H}}}\sum_{h\in H}\ket{ch}$ in HSP. The coset states $\{\ket{cH}\}_{c\in G/H}$ are manifestly $H$-invariant under a right-regular group action, and also mutually orthogonal. Therefore, the usual approach to HSP involving the construction of coset states via the function oracle is a specific instance of StateHSP. However, the more generic StateHSP problem differs in that it allows non-orthogonal coset states $\{R(c)\ket{\psi}\}_{c\in G/H}$. On the other hand, if the coset states are too close to each other, the problem risks becoming information-theoretically intractable: the state would be too close to the symmetric subspace invariant under all group operations, and distinguishing between any nontrivial subgroup symmetry and the full group symmetry would require an inefficient number of measurements. For this reason, it is crucial to introduce an orthogonality allowance $\epsilon$ in \Cref{def:StateHSP}.

A simple but powerful observation is that a large enough orthogonality can be further amplified when one is allowed multiple copies of the input state:

\begin{fact}\label{fact:StateHSPFourierSamplingAmplified} There is a Fourier sampling circuit for an StateHSP problem with orthogonality allowance $\epsilon$ which uses $t$ copies of the input state, yielding outcome distribution:
    \begin{equation}
        \mathrm{P}_{\mathrm{StateHSP}_{\psi,t}}[\lambda] = \mathrm{P}_{\mathrm{HSP}}[\lambda] + O\left(d_\lambda^2\,\mathsf{exp}(-\epsilon t)\right)\,.
    \end{equation}
The depth of the circuit can be as shallow as $O(\log t)$ by using $O(t\log\abs{G})$ ancillary qubits.
\end{fact}
\begin{proof}
    The construction is natural and it involves promoting the group action $\sum_{g\in G}\dyad{g}\otimes R(g)$ on $\C^G\otimes \C^d$ to the $t$-fold version $\sum_{g\in G}\dyad{g}\otimes R(g)^{\otimes t}$ on $\C^G\otimes \left(\C^d\right)^{\otimes t}$, and plugging this action into the StateHSP Fourier sampling circuit of \Cref{fact:StateHSPFourierSampling} with input state $\ket{\psi}^{\otimes t}$. A simple application of triangle inequality gives:
    \begin{align}
        \abs{\mathrm{P}_{\mathrm{StateHSP}_{\psi,t}}[\lambda] - \mathrm{P}_{\mathrm{HSP}}[\lambda]} &= \abs{\frac{d_\lambda}{\abs{G}}\sum_{\substack{c \in G/H \neq \mathrm{id}\\h\in H}} \chi_\lambda(ch) \,\mel{\psi^{\otimes t}}{R(c)^{\otimes t}}{\psi^{\otimes t}}}\\
        &\leq \frac{d_\lambda}{\abs{G}}\sum_{\substack{c \in G/H \neq \mathrm{id}\\h\in H}} \abs{\chi_\lambda(ch)} \abs{\mel{\psi}{R(c)}{\psi}}^t\\
        &\leq d_\lambda^2\left(1-\frac{\abs{H}}{\abs{G}}\right)(1-\epsilon)^t\leq O\left(d_\lambda^2\,\mathsf{exp}(-\epsilon t)\right),
    \end{align}
    where in the last line we used the simple character bound $\abs{\chi_\lambda(g)}\leq d_\lambda$ for all $g\in G$. The circuit can be implemented by $t$ successive applications of the single group action, once per each copy of the input Hilbert space. To apply this action in depth $O(\log t)$, we make use of $t$ ancillary registers which host regular representations of $G$. One can copy the group information from the first regular representation register onto all of these additional registers by a $O(\log t)$-depth binary tree of controlled two-register unitaries; then, each of these $t$ ancillary registers can control the group action on the $t$ input state copies in parallel. Finally, one uncomputes the binary tree of two-register unitaries in depth $O(\log t)$.
\end{proof}

We note that this simple bound uses no information about the specific group $G$; it is often the case that the characters decay rapidly in magnitude from the maximum value $\chi_\lambda(\mathrm{id})=d_\lambda$ across the group, so even tighter bounds might be possible. Similarly, improvements can be obtained if the specific StateHSP problem presents additional information about the inner products of the coset states\footnote{This comment captures the different behavior of our algorithm for the hidden cut problem with a constant entanglement promise versus a Haar-random promise, as will be described in later sections.}.

The key takeaway is that by increasing the number of copies $t$, we can naturally enhance the orthogonality of coset states and make the StateHSP instance behave like the equivalent HSP problem from the point of view of Fourier sampling. The number of copies required to make the non-orthogonality correction negligible in this way will depend on the specifics of the problem: for a problem parameter $n$, as long as $\epsilon\geq\Omega(1/\poly(n))$ and\footnote{This is because irrep dimensions cannot be larger than $d_\lambda\leq \sqrt{\abs{G}}$.} $\abs{G}\leq O(\mathsf{exp}(\poly(n)))$, then there is a choice of $t=\poly(n)$ which would ensure the corrections are negligibly small in $n$.

As we will show in the next section, this framework applies to the hidden cut problem, and we will be able to use an efficient number of copies to enhance orthogonality such that a standard HSP can be applied to the hidden cut problem via Fourier sampling. We leave it as an open question for future work to find meaningful instances of StateHSP which cannot be efficiently amplified by a polynomial number of copies for purposes of Fourier sampling.

We end this section by mentioning an immediate corollary which follow from the above connection between HSP and StateHSP when multiple state copies are available.
Namely, the general non-Abelian StateHSP is information-theoretically solvable using few copies of the state, just as the standard HSP is information-theoretically solvable with few queries to the function \cite{ettinger2004quantum}.
\begin{corollary}\label{cor:StateHSPInformationTheoretical}
The general non-Abelian HSP over a group $G$ with can be information-theoretically determined with a number of copies $O(\poly(\log\abs{G}, \,\epsilon^{-1}))$ of the input state $\ket{\psi}$.
 \end{corollary}
Of course, whether computationally efficient algorithms exist for non-Abelian groups is an open problem, and the non-existence of such algorithms underlies hardness of post-quantum cryptographic schemes \cite{regev2004quantum}.
This corollary follows immediately from the information-theoretic feasibility of HSP \cite{ettinger2004quantum}, combined with \Cref{fact:StateHSPFourierSamplingAmplified} outlined above.
Therefore, we know a StateHSP is information-theoretically solvable with enough copies and a large enough orthogonality allowance; the question is when it is computationally efficiently solvable.

\section{An algorithm for the hidden cut problem}\label{sec:HiddenCut}

In the hidden cut problem, the global state is a $tn$-qubit state of the form $\ket{\psi}^{\otimes t}$, where $\ket{\psi}\in (\C^{2})^n$ is separable across an unknown cut $C\subset[n]$, denoted $\ket{\psi}=\ket{\phi_1}_C\otimes\ket{\phi_2}_{\overline{C}}$, where $\overline{C}\equiv [n]\setminus C$. Aligning the global state in an imaginary qubit grid with $t$ rows and $n$ columns (see \Cref{fig:example}), let us define the relevant column permutation operations:
\begin{definition}[Permutation operations]
    The state $\ket{\psi}^{\otimes t} \in \left(\C^{2^{n}}\right)^{\otimes t}=\C^{2^{tn}}$ is a state on $t \times n$ qubits. Let the standard basis of this space be of the form $\bigotimes_{i\in[t],j\in[n]}\ket{e_{i,j}}$. Let us define the action of $\bbS_t$ on the $k$-th column of $t$ qubits as the operators $R_k(\pi)$:
    \begin{equation}
        \text{for }\pi\in \bbS_t, \,k\in[n]:\quad R_k(\pi)\bigotimes_{i\in[t],j\in[n]}\ket{e_{i,j}} \equiv \bigotimes_{i\in[t],j\in[n]\setminus\{k\}}\ket{e_{i,j}}\;\bigotimes_{i\in[t]}\ket{e_{\pi^{-1}(i),k}}\,.
    \end{equation}
\end{definition}
In what follows, we will choose appropriate group actions on the global state $\ket{\psi}^{\otimes t}$ expressed by groups of `column-wise' permutation operators $\{R_k(\pi)\}_{k\in[n],\pi\in\bbS_t}$, i.e. the permutation of the $k$'th qubit across the $t$ copies. 

\subsection{The full permutational symmetries defy Fourier sampling}\label{sec:symmetricgroupHSP}

In order to apply the StateHSP framework from \Cref{sec:StateHSP} to the hidden cut problem, the key first step is to choose an appropriate group action on the state. A first natural choice is to take advantage of all the manifest permutational symmetries of the global state $\ket{\psi}^{\otimes t}$ containing the $t$ copies of the input state. Such a $t$-fold state is trivially symmetric under the $\bbS_t$ group which permutes the copies. When additionally the input state is separable across a cut $C\subset[n]$, then the $t$-fold global state $\ket{\psi}^{\otimes t}=\ket{\phi_1}_C^{\otimes t}\otimes\ket{\phi_2}_{\overline{C}}^{\otimes t}$ has a larger $\bbS_t\times\bbS_t$ permutational symmetry group which acts by permuting the individual factor state copies (see \Cref{fig:example}). The natural parent group which accommodates all of these cut-specific symmetry subgroups is $G=\bbS_t^{\times n}$ acting as single-column permutations. In the language of StateHSP, the corresponding group action is $R:(\sigma_1,\dots,\sigma_n)\mapsto R_1(\sigma_1)\otimes\dots\otimes R_n(\sigma_n)$, and the hidden cut subgroup $H_C<G$ preserving the state under this action is\footnote{As elsewhere in this paper, the notation $\sigma^{C}\mu^{\overline{C}}$ means the vector in $\bbS_t^{\times n}$ with $\sigma$ in the $C$ positions and $\mu$ elsewhere.} $H_C=\{\sigma^{C}\mu^{\overline{C}}\}_{\sigma,\mu \in \bbS_t}$.

In \Cref{sec:StateHSP} we imported the Fourier sampling approach from HSP as a possible algorithm to solve StateHSP. We now briefly outline an obstacle to applying Fourier sampling to find the hidden cut with the permutation group action introduced above. While the Fourier sampling circuit can be implemented efficiently,\footnote{This is because of the known efficient quantum Fourier transform constructions for $\bbS_t$ \cite{beals1997quantum}.} the difficulty comes from the information-theoretic properties of the equivalent standard HSP:

\begin{fact}
    Assume $t\geq\Omega(\poly(n))$. Given a cut $C\subset[n]$ which determines a symmetry subgroup $H_C$ isomorphic to $\bbS_t\times \bbS_t$ inside $\bbS_{t}^{\times n}$ as defined above, then performing non-Abelian weak Fourier sampling over $\bbS_{t}^{\times n}$ gives the output probabilities:
    \begin{align}
        \label{eq:almostplancherel}
        \mathrm{P}_{\mathrm{HSP}}[\lambda_1,\dots,\lambda_n] &= \frac{d_{\lambda_1}^2\dots d_{\lambda_n}^2}{t!^n}\left(1 + O(t^2\,b^n)\right) & \text{for }\frac{\lambda_{i,1}}{t},\frac{\lambda_{i,1}'}{t}<o(1),\;i\in[n],
    \end{align}
    for some constant $0<b<1$, with a negligible probability mass outside of this regime. Here, $\lambda_i$ are irreducible representations of $\bbS_{t}$ of dimensions $d_{\lambda_i}$, and $\lambda_{i,1}$, $\lambda_{i,1}'$ are the lengths of the first row and first column of the Young diagram $\lambda_i$. As a consequence, the probability of observing irreps $\lambda_i$ outside of the range $\frac{\lambda_{i,1}}{t},\frac{\lambda_{i,1}'}{t}<o(1)$ is negligibly small. Inside the typical observable range, the result means that all cuts result in the same ``Plancherel distribution'' of outcomes to within negligible relative corrections $O(t^2b^n)$.
\end{fact}

{\bf\em Proof sketch.}
    We merely outline the argument, which relies on technical aspects of the representation theory of the symmetric group; a similar proof is detailed in \cite{moore2008symmetric} to show that Fourier sampling cannot solve the generic HSP for the symmetric group. The goal is to estimate the benchmark HSP outcome distribution \eqref{eq:HSPFourierDistr}. Let $\cyc(\sigma)$ denote the number of cycles in the permutation $\sigma$. The starting observation is that a subgroup element $\sigma_1^C\sigma_2^{\overline{C}}$ has a number of cycles equal to $n\cyc(\sigma_1)/2 + n\cyc(\sigma_2)/2$ as a member of $\bbS_{t}^{\times n}$. Using Roichman's bounds on the characters of the symmetric group \cite{roichman1996upper}, this gives us that there exists some $0<b<1$ such that:
    \begin{equation}
        \abs{\chi_{\lambda_{1},\dots,\lambda_n}(\sigma_1^C\sigma_2^{\overline{C}})} \leq d_{\lambda_1}\dots d_{\lambda_n}\,b^{n(t - \cyc(\sigma_1)-\cyc(\sigma_2))/2}\,.
    \end{equation}
    This holds true as long as the Young diagrams $\lambda_1,\dots,\lambda_n \vdash t$ have a first row or column shorter than a $o(1)$ fraction of $t$. This condition is satisfied with high probability due to arguments such as \cite{baik1999distribution} about the typical Young diagrams concentrating towards $\Theta(\sqrt{t})$ rows and columns. Therefore, the chance of ever seeing any diagram outside the scope of this typical regime is negligibly small when $t\geq\Omega(\poly(n))$. Within the typical regime, the bound above is enough to control the non-identity terms in the sum over the subgroup $H$ in \eqref{eq:HSPFourierDistr}, which results in the claim. \hfill \qed

The message of the above claim is that the corresponding HSP problem becomes information-theoretically harder with increasing number of copies $t$, which is counter-productive if we expect to use the number of copies $t$ as a method of signal amplification. Additionally, it is unclear whether in the low-$t$ regime Fourier samples can be efficiently analyzed to detect the cut. This choice of group action has the disadvantage that the same parameter $t$ defines both the accuracy of the HSP approximation via orthogonality amplification, as well as the complexity of the HSP problem. In the next section, we find a much simpler permutation action such that the parent group depends on the state size $n$ but not on $t$, with the added benefit of the group being Abelian.

\subsection{A first Abelian HSP algorithm for the hidden cut problem}\label{sec:mainproof}

It turns out we can restrict the full permutation group $\bbS_t^{\times n}$ to a smaller group which still supports the mapping of cuts to subgroups, but whose size grows only with the number of qubits $n$, but not with the number of copies $t$. Crucially, the group is the simple Abelian group $G=\Z_2^n$; since HSP is known to allow efficient Fourier sampling algorithms in the Abelian case, this opens up the possibility of an efficient algorithm for the hidden cut problem by the technique in \Cref{fact:StateHSPFourierSamplingAmplified}.

Assume $t$ is even and define the following permutation in $\bbS_t$:    
    \begin{equation}\label{eq:xipermutation}
        \xi \equiv (1\;\;2)\,(3\;\;4)\,\dots\,(t-1\;\;t)\,.
    \end{equation}
Let us allow $G=\Z_2^{n}$ to act on the space of $t \times n$ qubits as:
    \begin{equation}
        (x_1,\dots,x_{n}) \circ \ket{\psi}^{\otimes t} \equiv R_1(\xi)^{x_1}R_2(\xi)^{x_2}\dots R_{n}(\xi)^{x_{n}}\ket{\psi}^{\otimes t}.
    \end{equation}
This is a well-defined $\Z_2^{n}$ group action since the $R_i(\xi)$ operators mutually commute (Abelian) and $R_i(\xi)^2=I$ (order two since $\xi^2=\text{id}$). Unless specified otherwise, we can simply denote $R_k\equiv R_k(\xi)$ and $R(\mathbf{x})\equiv R_1^{x_1}\dots R_n^{x_n}$ to simplify notation from here onwards. See \Cref{fig:AbelianGroupAction} for an illustration.

A first algorithm to find the hidden cut $C\subset[n]$ can be laid out as follows:

\begin{algorithm}[H]\caption{Non-adaptive hidden cut algorithm}
\label{alg:main}
\SetStartEndCondition{ }{}{}%
\SetKwProg{Fn}{def}{\string:}{}
\SetKwInput{kwReqs}{Requirements}
\SetKwFunction{Range}{range}
\SetKw{KwTo}{in}\SetKwFor{For}{for}{\string:}{}%
\SetKwIF{If}{ElseIf}{Else}{if}{:}{elif}{else:}{}%
\SetKwFor{While}{while}{:}{fintq}%
\newcommand{\forcond}{$i$ \KwTo\Range{$n$}}
\AlgoDontDisplayBlockMarkers\SetAlgoNoEnd\SetAlgoNoLine%
\kwReqs{$n$ additional qubits, implementation of the $\Z_2^n$ group action $U_{\Z_2^n}$.}
\uIf{the factor states are promised to be $\epsilon$-far from separable (\Cref{thm:main})}{
        Let $t=O(n/\epsilon^2)$.
}
\uIf{the factor states are promised to be Haar-random (\Cref{thm:haar})}{
        Let $t=2$.
}
\For{sample count $k\in\{1,\dots,p\}$, where $p=O(n)$}{
    Prepare $t$ copies of the state $\ket{\psi}$.\\
    Run the Fourier sampling circuit (\Cref{subfig-3:HiddenCutCircuit}): $\left(H^{\otimes n}\otimes I\right)U_{\Z_2^n}\left(H^{\otimes n}\otimes I\right)\;\ket{0^{n}}\otimes \ket{\psi}^{\otimes t}$.\\
    Measure the group register to get sample $\mathbf{y}^{(k)}\in\Z_2^n$.
}
Classically solve for the nullspace of $Y=\left(\mathbf{y}^{(1)},\dots,\mathbf{y}^{(p)}\right)^T \in \Z_2^{p\times n}$ which is $\mathrm{span}\{1^C0^{\overline{C}}, 0^C1^{\overline{C}}\}$.
\end{algorithm}

The above algorithm succeeds in efficiently finding the hidden cut:

\begin{theorem}[Non-adaptive hidden cut algorithm]\label{thm:nonadaptive}
    For $\epsilon>0$ and an $n$-qubit input state $\ket{\psi}=\ket{\phi_1}_C\otimes \ket{\phi_2}_{\overline{C}}$ separable across a cut $C\in\binom{[n]}{n/2}$, assume that the factor states $\ket{\phi_{1,2}}$ are at least $\epsilon$-far from all separable $(n/2)$-qubit states. Then, \Cref{alg:main} succeeds in finding the hidden cut $C$ with high probability using $O(n^2/\epsilon^2)$ copies of the input state $\ket{\psi}$. The algorithm requires coherent access to $O(n/\epsilon^2)$ copies at a time, on which it acts with circuits of depth $O(\log(n/\epsilon^2))$ given $O(n^2/\epsilon^2)$ ancillary qubits.
\end{theorem}

As outlined above, \Cref{alg:main} addresses both the generic hidden cut problem, which promises that the factor states are $\epsilon$-far from separable, as well as the case of Haar-random factor states. The difference between the two is in the choice of the number of state copies $t$ used to produce each Fourier sample. This section focuses on the first, generic case. In the later \Cref{sec:Haar}, we will prove why the stronger promise of Haar-random factor states improves the requirements of \Cref{alg:main}, such that useful Fourier samples can be produced by constant-depth circuits acting on only two state copies at a time (see \Cref{thm:haar}).

As stated in \Cref{thm:nonadaptive}, the above \Cref{alg:main} succeeds in finding the cut with $O(n^2/\epsilon^2)$ total state copies. Crucially, at the end of this section we will introduce an adaptive modification of \Cref{alg:main}, which will allow us to find the hidden cut with only $O(n/\epsilon^2)$ state copies, thus further reducing the requirement by a factor of $n$. This second, adaptive hidden cut algorithm will be described in Section \ref{sec:AdaptiveSubspaceAlgorithm} and will build upon the analysis of \Cref{alg:main}, which follows below.

\subsection{Analyzing the hidden cut algorithm: proof of \Cref{thm:nonadaptive} 
}

To show that the algorithm succeeds in finding the cut with high probability, we organize the analysis across the following facts. The first two facts are straightforward:
\begin{fact}[Circuit size]\label{fact:CircuitSizeFanOut}
    With $n$ ancillary qubits to represent the group $\Z_2^n$, the Fourier sampling circuit for the $\Z_2^n$ action defined above on the $t$-fold input state $\ket{\psi}^{\otimes t} \in (\C^{2^n})^{\otimes t}$ can be implemented efficiently with a circuit of depth $O(\log t)$ involving $nt/2$ ancillary qubits.
\end{fact}
\begin{proof}
    The circuit is sketched in \Cref{subfig-3:HiddenCutCircuit}. The group quantum Fourier transform over $\Z_2^n$ is simply $n$ parallel Hadamard gates, while the controlled group action $U_{\Z_2^n}=\sum_{\mathbf{x}\in \Z_2^n}\dyad{\mathbf{x}}\otimes R(\mathbf{x})$ is efficient to implement as $n$ parallel sequences of $t/2$ controlled-SWAPs. With a single ancllary register of $n$ qubits, this can be implemented in depth $O(t)$, by having each of the $n$ ancillary qubits control a sequence of $t/2$ SWAPs in parallel. The depth can be lowered to $O(\log t)$ by a fan-out construction at the expense of using $t/2$ ancillary $n$-qubit registers; one copies the information from the first ancillary register onto all $t$ registers by a $O(\log t)$-depth binary tree of CNOT gates, after which each of the $nt/2$ ancillary qubits controls a SWAP in parallel, followed by uncomputing the CNOT fan-out in depth $O(\log t)$.
\end{proof}

\begin{fact}[Hidden cut as Abelian StateHSP]\label{fact:hiddensubgroup}
    To each cut $C\subset[n]$ corresponds a hidden subgroup $H_C<\Z_2^n$ isomorphic to $\Z_2^2$ which preserves the state under the group action, given by:
    \begin{equation}\label{eq:hiddensubgroupHC}
        H_C=\{0^n,\;1^C0^{\overline{C}},\;0^C1^{\overline{C}},\;1^n\}\,,
    \end{equation}
    where by $1^C0^{\overline{C}}$ we mean the $n$-bit string with $1$'s on the positions in $C$ and $0$'s everywhere else. Similarly, denote by $\mathbf{y}_1^{C}\mathbf{y}_2^{\overline{C}}$ the $n$-bit string whose restriction to the positions in $C$ is the sub-string $\mathbf{y}_1\in\Z_2^{\abs{C}}$, and whose restriction to the positions in $\overline{C}=[n]\setminus C$ is the sub-string $\mathbf{y}_2\in\Z_2^{n-\abs{C}}$. Then, performing Fourier sampling on the standard HSP over $\Z_2^n$ with the hidden subgroup $H_C$ produces $n$-bit string samples from the probability distribution:
    
    \begin{equation}\label{eq:benchmarHSPdistr}
        \mathrm{P}_{\mathrm{HSP}}[\mathbf{y}_1^{C}\mathbf{y}_2^{\overline{C}}]=\frac{\delta_{\abs{\mathbf{y}_1}\text{ even}}\,\delta_{\abs{\mathbf{y}_2}\text{ even}}}{2^{n-2}}=\left\{
        \begin{array}{lr}
            2^{-n+2} & \text{if }\mathbf{y}_1\cdot1^{n/2}=\mathbf{y}_2\cdot1^{n/2}=0\;\mathrm{mod}\;2 \\
            0 & \text{otherwise.}
        \end{array}
        \right.
    \end{equation}
    
    This distribution is uniform over the $(n-2)$-dimensional Boolean subspace in $\Z_2^n$ defined by the cut $C$:
    \begin{equation}\label{eq:Sigma_C_subspace_def}
        H_C^\perp\equiv\left\{\mathbf{y}\in\Z_2^n\,:\,\mathbf{y}\cdot1^C0^{\overline{C}}=\mathbf{y}\cdot0^C1^{\overline{C}}=0\;\mathrm{mod}\;2\right\}.
    \end{equation}
\end{fact}
This formalizes the hidden cut problem as an instance of StateHSP, and provides the benchmark Fourier sampling probability of the equivalent standard HSP. The proof of the above fact is straightforward, given the manifest permutational symmetries of the $t$-fold input state (see \Cref{fig:AbelianGroupAction}). We remark that the resulting HSP is a minor variation of the well-known Simon's algorithm \cite{simon1997power}, which provides uniform samples from the $(n-1)$-dimensional subspace of $\Z_2^n$ orthogonal to the secret string $\mathbf{s}$: $H_{\mathbf{s}}^{\perp}\equiv\{\mathbf{y}\in\Z_2^n\,:\,\mathbf{y}\cdot\mathbf{s}=0\}$. Just like in Simon's algorithm, learning the orthogonal subspace based on samples from this distribution is the same as learning the secret, which in this case means learning the $n$-bit string $1^C0^{\overline{C}}$ which describes the hidden cut $C$:
\begin{fact}[Linear system]\label{fact:linearsystems}
    A complete spanning set $(\mathbf{y}^{(1)},\dots,\mathbf{y}^{(p)})$ for the cut subspace $H_C^\perp$ can be obtained with $p=O(n)$ independent samples from the HSP distribution \eqref{eq:benchmarHSPdistr} with high probability. Given such a spanning set, the string encoding the cut $1^C0^{\overline{C}}$ or its equivalent mirror opposite $0^C1^{\overline{C}}$ can be determined by solving for the nullspace of the matrix $Y=(\mathbf{y}^{(1)},\dots,\mathbf{y}^{(p)})^T$ with the samples as rows, which can be done in $\poly(n)$ time, e.g.\  by Gaussian elimination.
\end{fact}
\begin{proof}
The proof of this fact is straightforward and follows the same logic as Simon's algorithm. As an alternative to solving for the nullspace of the matrix $Y$, we mention here a slightly slower, but more illustrative equivalent procedure of analyzing the collected samples. The idea is to iteratively learn the members of each side of the cut by solving a number of $n-1$ linear equations involving the matrix $Y$, in the following way: first, ask whether positions 1 and 2 (out of $n$) are on the same side of the cut (i.e.\  whether they are both in $C$ or both in $\overline{C}=[n]\setminus C$). This is answered by solving for $\mathbf{x}$ in the linear system $Y\mathbf{x} = (1,1,0,\dots,0)^T$; there is a solution $\mathbf{x}$ if 1 and 2 are on the same side, otherwise the system is infeasible. Continue in this way for each of the remaining positions $3,\dots,n$ by solving the same pairwise membership check of each position against position 1, thus determining the cut allocation of all coordinates in $\poly(n)$ time.
\end{proof}

To understand when this Simon-like HSP algorithm can be applied to the Fourier samples coming from the associated StateHSP hidden cut problem, we need to apply the framework from \Cref{sec:StateHSP} to bound the difference between the two distributions. This will inform the number of copies necessary for orthogonality amplification such that the two probability distributions become negligibly close at the level of each outcome.
\begin{fact}[Output distribution over $\Z_2^{n}$]\label{fact:HCPOutput}
    Assume $C\in\binom{[n]}{n/2}$ is the true cut and the factorization of the input state is $\ket{\psi}=\ket{\phi_1}^C\otimes \ket{\phi_2}^{\overline{C}}$. Given $t$ state copies for each sample, \Cref{alg:main} returns strings in $\{0,1\}^{n}$ according to the probability distribution $\mathrm{P}_{\mathrm{StateHSP}_{\psi,t}}$, which respects:
    \begin{align}\label{eq:cutdistribution}
        \abs{\mathrm{P}_\mathrm{StateHSP_{\psi,t}}[\mathbf{y_1}^C\mathbf{y_2}^{\overline{C}}] - \mathrm{P}_\mathrm{HSP}[\mathbf{y_1}^C\mathbf{y_2}^{\overline{C}}]} &\leq \frac{1}{2^{n-2}}\left[\prod_{k\in\{1,2\}}\left(1+\Delta_{\phi_k,t}\right)-1\right]\,,
    \end{align}
    where:
    \begin{equation}
        \Delta_{\phi,t}\equiv \sum_{\substack{S\subset[n/2]\\1\notin S\\S\neq\varnothing}}\Tr[\phi_{S}^2]^{t/2}\,.
    \end{equation}
\end{fact}
\begin{proof}
The starting point is the observation that powers of the purity enter naturally as the inner products of the coset states from \Cref{sec:StateHSP}:
\begin{equation}
    \mel{\psi^{\otimes t}}{R(\mathbf{x})}{\psi^{\otimes t}} = \Tr[\psi_{\mathbf x}^2]^{t/2}\,,
\end{equation}
where the $\Tr[\psi_{\mathbf x}^2]$ denotes the purity across the cut $C_{\mathbf{x}}\equiv\{i\in[n]\,:\:x_i=1\}$ represented by the bit-string $\mathbf{x} \in \Z_2^n$. Importing the StateHSP analysis \eqref{eq:StateHSPFourierDistr} from \Cref{sec:StateHSP}, it follows that our Fourier sampling circuit effectively performs a Boolean Fourier transform on this set of amplified purities. Specifically, the output distribution is:
\begin{equation}\label{eq:entanglementfeature}
    \mathrm{P}_\mathrm{StateHSP_{\psi,t}}[\mathbf{y}]=\frac{1}{2^n}\sum_{\mathbf{x}\in\Z_2^n}(-1)^{\mathbf{x}\cdot \mathbf{y}}\;\Tr[\psi_{\mathbf x}^2]^{t/2},\quad\text{ where }\mathbf{y}\in\Z_2^n.
\end{equation}
Operations in the hidden subgroup $H_C$ \eqref{eq:hiddensubgroupHC} preserve the state, so we can split the above sum over the group $\Z_2^n$ into a sum over the subgroup $H_C$ and a sum over the coset representatives $\Z_2^n/H_C$:
\begin{equation}
   \mathrm{P}_\mathrm{StateHSP_{\psi,t}}[\mathbf{y}]=\frac{1}{2^n}\sum_{\mathbf{h}\in H_C}(-1)^{\mathbf{h}\cdot \mathbf{y}}\times\sum_{\mathbf{x}\in\Z_2^n/H_C}(-1)^{\mathbf{x}\cdot \mathbf{y}}\;\Tr[\psi_{\mathbf x}^2]^{t/2}\,.
\end{equation}
The first term is precisely the $\mathrm{P}_\mathrm{HSP}[\mathbf{y}]$ distribution corresponding to the standard HSP problem with the same specifications. Given the internal structure of the state, the purity factors into the two separate contributions from each substate:
\begin{equation}\label{eq:NonAdaptiveOutcomeDistribution}
    \mathrm{P}_\mathrm{StateHSP_{\psi,t}}[\mathbf{y_1}^C\mathbf{y_2}^{\overline{C}}]=\mathrm{P}_\mathrm{HSP}[\mathbf{y}]\times\prod_{k\in\{1,2\}}\left[\sum_{\substack{\mathbf{z}\in\Z_2^{n/2}\\z_1=0}}(-1)^{\mathbf{z}\cdot \mathbf{y}_k}\;\Tr[\phi_{k,\mathbf{z}}^2]^{t/2}\right]\,.
\end{equation}
Here, we have chosen coset representatives $\Z_2^n/H_C=\{\mathbf{z}^C\mathbf{\overline{z}}^{\overline{C}}\;:\;\mathbf{z},\mathbf{\overline{z}}\in\Z_2^{n/2}, \;z_1=\overline{z}_1=0\}$. The conclusion follows from a triangle inequality (i.e.\  ignoring the $\pm 1$ phases) on all the terms on the right hand side except the leading term from $\mathbf{z}=0^{n/2}$, which corresponds to a trivial void cut with purity one.
\end{proof}

A sufficient condition for the algorithm to work is to ensure that the two distributions are negligibly close in a relative sense at the level of each outcome, meaning $\mathrm{P}_\mathrm{StateHSP_{\psi,t}}[\mathbf{y}]=\mathrm{P}_\mathrm{HSP}[\mathbf{y}]\;(1+\negl(n))$. Given the above fact, it is enough to choose the number of copies per sample $t$ such that $\Delta_{\phi_{1,2},t}=\negl(n)$.
\begin{fact}\label{fact:internalentanglementbounds}
    If a state $\ket{\phi}$ on $n/2$ qubits is at least $\epsilon$-far from any separable state along any internal cut, then: $\Delta_{\phi,t}\leq 2^{n/2}(1-\epsilon^2)^{t/2}$. Therefore, a choice of $t=O(n/\epsilon^2)$ copies per sample makes the relative correction in \Cref{fact:HCPOutput} negligible in $n$.
     
\end{fact}
\begin{proof}
    The result follows from a straightforward binomial sum argument. Using \Cref{fact:epsilonpurity}, we have that each nontrivial purity is upper-bounded by $\Tr[\phi_S^2]\leq 1-\epsilon^2$, therefore:
    \begin{align}
        \Delta_{\phi,t} &= \sum_{\substack{S\subset[n/2]\\1\notin S\\S\neq\varnothing}}\Tr[\phi_{S}^2]^{t/2}\\
        &\leq 2^{n/2}(1-\epsilon^2)^{t/2}\\
        &\leq e^{\frac{\ln 2}{2}n - \frac{t\epsilon^2}{2}}\,.
    \end{align}
    Therefore, a choice of $t=O(n/\epsilon^2)$ is enough to make this quantity negligible in $n$.
\end{proof}

This completes the proof of the main theorem.

\subsection{Improving the Abelian HSP algorithm by adaptive subspace preparations}\label{sec:AdaptiveSubspaceAlgorithm}

In this section, we describe an adaptive modification of \Cref{alg:main} which improves the number of state copies required to determine the hidden cut by a factor of $n$, from $O(n^2/\epsilon^2)$ down to $O(n/\epsilon^2)$. This achieves an optimal asymptotic in terms of the number of state copies (up to logarithmic factors) as announced in the introduction, given the related decision lower bound of Jones and Montanaro \cite{jones2024testing}. The adaptive algorithm operates as follows:

\begin{algorithm}[H]
\caption{Hidden cut algorithm with adaptive subspaces}
\label{alg:adaptive}
\SetStartEndCondition{ }{}{}%
\SetKwProg{Fn}{def}{\string:}{}
\SetKwInput{kwReqs}{Requirements}
\SetKwInput{kwParams}{Parameters}
\SetKwFunction{Range}{range}
\SetKw{KwTo}{in}\SetKwFor{For}{for}{\string:}{}%
\SetKwIF{If}{ElseIf}{Else}{if}{:}{elif}{else:}{}%
\SetKwFor{While}{while}{:}{fintq}%
\newcommand{\forcond}{$i$ \KwTo\Range{$n$}}
\AlgoDontDisplayBlockMarkers\SetAlgoNoEnd\SetAlgoNoLine
\kwParams{$n$ qubits, factor states promised to be $\epsilon$ away from product states.}
\kwReqs{$2n$ additional qubits, implementation of the $\Z_2^n$ group action $U_{\Z_2^n}$.}
\For{sample count $k\in\{1,\dots,n-2\}$}{
    Define the Boolean subspaces $\Sigma_k^\perp\equiv \mathrm{span}\{\mathbf{y}^{(1)},\dots,\mathbf{y}^{(k-1)}\}$ and $\Sigma_k\equiv \left(\Sigma_k^\perp\right)^\perp$, defined as $\Sigma_k=\left\{\mathbf{z}\in\Z_2^n\,:\,\mathbf{z}\cdot \mathbf{y}^{(j)}=0\;\mathrm{mod}\;2,\;\forall j\in[k-1]\right\}$. If $k=1$, then set $\Sigma_k=\Z_2^n$.\\
    In the ancillary group register, prepare the superposition $\ket{\Sigma_k}=\frac{1}{\sqrt{2^{n-k+1}}}\sum_{\mathbf{z}\in \Sigma_k} \ket{\mathbf{z}}$.\\
    Prepare $t=O(1/\epsilon^2)$ copies of the state $\ket{\psi}$.\\
    Run the Fourier sampling circuit: $\left(H^{\otimes n}\otimes I\right)U_{\Z_2^n}\;\ket{\Sigma_k}\otimes \ket{\psi}^{\otimes t}$.\\
    Measure the group register to obtain a new sample $\mathbf{y}^{(k)}\in\Z_2^n$.\\
    Keep the sample if it is nonzero and outside $\Sigma_k^\perp$, otherwise repeat.\label{line:RejectionSampling}
}
Classically solve for the nullspace of $Y=\left(\mathbf{y}^{(1)},\dots,\mathbf{y}^{(n-2)}\right)^T \in \Z_2^{p\times n}$ which is $\mathrm{span}\{1^C0^{\overline{C}}, \,0^C1^{\overline{C}}\}$.
\end{algorithm}

Compared to \Cref{alg:main}, the key difference is a different initial state in the ancillary register which hosts the regular representation of the parent group $G=\Z_2^n$. The previous \Cref{alg:main} followed a standard Fourier sampling procedure which initialized the group register in a uniform superposition over all group elements, i.e. over all of $\Z_2^n$. By comparison, the adaptive \Cref{alg:adaptive} introduced here will instead initialize the group register in a uniform superposition over the Boolean subspace which is orthogonal to previously collected samples. We will show that this serves to boost the probability that new samples will be linearly independent, such that a smaller number of copies is needed at every step for amplification purposes. Our main result is the analysis of this algorithm, showing that it succeeds in finding the hidden cut with constant probability:

\begin{customtheorem}{\ref{thm:main}}[Hidden cut algorithm --- restated]
    For $\epsilon>0$ and an $n$-qubit input state $\ket{\psi}=\ket{\phi_1}_C\otimes \ket{\phi_2}_{\overline{C}}$ separable across a cut $C\in\binom{[n]}{n/2}$, assume that the factor states $\ket{\phi_{1,2}}$ are at least $\epsilon$-far from all separable $(n/2)$-qubit states. Then, \Cref{alg:adaptive} succeeds in finding the hidden cut $C$ with constant probability using $O(n/\epsilon^2)$ copies of the input state $\ket{\psi}$. The algorithm requires coherent access to $O(1/\epsilon^2)$ copies at a time, on which it acts with circuits of depth $O(n^2)+O(\log\epsilon^{-1})$, and polynomial-time classical processing.
\end{customtheorem}

\begin{proof}
    Since \Cref{alg:adaptive} is a direct modification of the Fourier sampling approach of \Cref{alg:main}, the proof of this theorem proceeds along similar lines. Three key technical points need to be added to the analysis, which we prove in the rest of this section. First, we show that the `subspace states' $\ket{\Sigma_k}$ can indeed be efficiently prepared on the group register at the beginning of each sampling round (this is shown in \Cref{fact:EfficientSubspaceStateConstruction} below). The efficient circuits involved in preparing these states rely on finding basis vectors for the corresponding subspaces, which can be efficiently obtained classically. Second, we show that all samples lie inside the cut subspace $H_C^\perp$ with probability one, which is a consequence of the hidden cut StateHSP instance admitting the subgroup $H_C$ as the hidden symmetry subgroup; we show this in \Cref{fact:NewSamplesAreInCutSubspace} below. Finally, we show that consuming a number of $t=O(1/\epsilon^2)$ state copies per sample results in a constant probability of the new sample being linearly independent with respect to previous samples (see \Cref{fact:NewSamplesLinearlyIndependent} below). This suffices for an overall constant probability of success of \Cref{alg:adaptive} due to the rejection sampling procedure on \Cref{line:RejectionSampling}, since at every sampling round we reject new outcomes until they are linearly independent.
\end{proof}

\begin{fact}\label{fact:EfficientSubspaceStateConstruction}
    If $\Sigma$ is a $d$-dimensional subspace of $\Z_2^n$, then the $n$-qubit subspace state $\ket{\Sigma}\equiv \frac{1}{2^{d/2}}\sum_{\mathbf{z}\in \Sigma}\ket{\mathbf{z}}$ can be efficiently prepared with circuits of size $O(nd)\leq O(n^2)$.
\end{fact}

\begin{proof}
    Given an $n$-bit string $\mathbf{z}\in\Z_2^n$, one can easily implement the $(1+n)$-qubit controlled addition unitary $U_{\mathbf{z}}: \;\ket{a}\ket{\mathbf{x}}\mapsto \ket{a}\ket{\mathbf{x}\oplus a\mathbf{z}}$ for any $a\in\B,\,\mathbf{x}\in\Z_2^n$. Specifically, this can be implemented with a number $\abs{\mathbf{z}}=O(n)$ of sequential CNOT gates controlled on the $a$ register, which act on the $\mathbf{x}$ registers in the locations on which the string $\mathbf{z}$ has entries equal to one.
    
    Let $\mathbf{z}_1,\dots,\mathbf{z}_d \in \Z_2^n$ be a basis of the subspace $\Sigma$. Then, by a sequence of unitaries $U_{\mathbf{z}_1},\dots,U_{\mathbf{z}_d}$ of the kind described above, one can efficiently implement the $(d+n)$-qubit unitary:
    \begin{align}
        U_\Sigma:\;\ket{\mathbf{a}}\ket{\mathbf{x}} &\mapsto \ket{\mathbf{a}}\ket{\mathbf{x}\oplus a_1\mathbf{z}_1 \oplus\dots \oplus a_d\mathbf{z}_d},
    \end{align}
    with a circuit of total depth $O(nd)$. Similarly, with the same gate count one can implement the `inverse' $(d+n)$-qubit unitary which acts as:
    \begin{align}
        V_\Sigma:\;\ket{\mathbf{b}}\ket{a_1\mathbf{z}_1\oplus\dots\oplus a_d\mathbf{z}_d} &\mapsto \ket{\mathbf{b}\oplus \mathbf{a}}\ket{a_1\mathbf{z}_1\oplus\dots\oplus a_d\mathbf{z}_d},
    \end{align}
    for any $\mathbf{b},\,\mathbf{a}=(a_1,\dots,a_d)\in\Z_2^d$.

    Then, starting from the zero state on $d+n$ qubits, the substate state can be prepared as:
    \begin{align*}
        \ket{\mathbf{0}}\otimes \ket{\mathbf{0}}
        & \longrightarrow  \frac{1}{2^{d/2}}\sum_{\mathbf{a}\in\Z_2^d} \ket{\mathbf{a}}\otimes\ket{\mathbf{0}} \tag{applying $d$ Hadamard gates on the first $d$ qubits}\\
        & \longrightarrow \frac{1}{2^{d/2}}\sum_{\mathbf{a}\in\Z_2^d} \ket{\mathbf{a}}\otimes\ket{a_1\mathbf{z}_1\oplus\dots\oplus a_d\mathbf{z}_d} \tag{applying the $U_S$ circuit defined above}\\
        & \longrightarrow \ket{\mathbf{0}}\otimes \frac{1}{2^{d/2}}\sum_{\mathbf{a}\in\Z_2^d} \ket{a_1\mathbf{z}_1\oplus\dots\oplus a_d\mathbf{z}_d} \tag{applying the $V_S$ circuit defined above}\\
        & = \ket{\mathbf{0}}\otimes\ket{\Sigma}.
    \end{align*}
    This procedure prepares the substate state $\ket{\Sigma}=2^{-d/2}\sum_{\mathbf{a}\in\Z_2^d}\ket{a_1\mathbf{z}_1\oplus\dots\oplus a_d\mathbf{z}_d}$ on the last $n$ qubits, with a circuit of overall size $O(nd)$.
\end{proof}

\begin{fact}\label{fact:NewSamplesAreInCutSubspace}
    Each new sample $\mathbf{y}^{(k)}$ is always in the cut subspace $H_C^\perp$.
\end{fact}

\begin{proof}
    The state prepared at the $k$-th round of \Cref{alg:adaptive} is of the form:
    \begin{equation}
        \frac{1}{2^{n-(k-1)/2}}\sum_{\mathbf{y}\in Z_2^n}\ket{\mathbf{y}}\otimes \sum_{\mathbf{z}\in \Sigma_k}(-1)^{\mathbf{z}\cdot \mathbf{y}}R(\mathbf{z})\ket{\psi}^{\otimes t},
    \end{equation}
    on which measuring the first register returns an outcome $\mathbf{y}\in\Z_2^n$ with probability:
    \begin{align}
        \mathrm{P}[\mathbf{y}] &= \frac{1}{2^{2n-k+1}}\sum_{\mathbf{z},\mathbf{z}'\in \Sigma_k} (-1)^{\mathbf{y}\cdot(\mathbf{z}\oplus\mathbf{z}')}\mel{\psi^{\otimes t}}{R(\mathbf{z}')^\dagger\,R(\mathbf{z})}{\psi^{\otimes t}} \nonumber \\
        &= \frac{1}{2^{2n-k+1}}\sum_{\mathbf{z},\mathbf{z}'\in \Sigma_k} (-1)^{\mathbf{y}\cdot(\mathbf{z}\oplus\mathbf{z}')}\mel{\psi^{\otimes t}}{R(\mathbf{z}'\oplus \mathbf{z})}{\psi^{\otimes t}} \tag{as $R$ is $\Z_2^n$-representation} \nonumber \\
        &= \frac{1}{2^n}\sum_{\mathbf{z}\in \Sigma_k} (-1)^{\mathbf{y}\cdot\mathbf{z}}\mel{\psi^{\otimes t}}{R(\mathbf{z})}{\psi^{\otimes t}} \tag{by summing over $\Sigma_k$}
    \end{align}
    where we used the fact that $\Sigma_k$ is a $(n-k+1)$-dimensional subspace of $\Z_2^n$, so it also operates as a subgroup of $\Z_2^n$ under addition.

    We notice that if a string $\mathbf{z}$ is in the subspace $\Sigma_k$, then we must have that all elements in the associated hidden coset are also in $\Sigma_k$. The argument proceeds by induction. Specifically, $\mathbf{z}\in \Sigma_k$ if it is orthogonal to previous samples: $\mathbf{z}\cdot \mathbf{y}^{(j)}=0\;\mathrm{mod}\,2$, for $j\in [k-1]$. Assume that previous samples are in the cut subspace $H_C^\perp$, meaning that $1^C0^{\overline{C}}\cdot \mathbf{y}^{(j)}=0^C1^{\overline{C}}\cdot \mathbf{y}^{(j)}=0$ for $j\in[k-1]$. Then, we also have that $(\mathbf{z}\oplus 1^C0^{\overline{C}})\cdot\mathbf{y}^{(j)} = (\mathbf{z}\oplus 0^C1^{\overline{C}})\cdot\mathbf{y}^{(j)}=(\mathbf{z}\oplus 1^n)\cdot\mathbf{y}^{(j)}=0\;\mathrm{mod}\,2$ for $j\in[k-1]$ --- in other words, if $\mathbf{z}\in \Sigma_k$, then also the rest of the coset $\mathbf{z}\oplus 1^C0^{\overline{C}},\;\mathbf{z}\oplus 0^C1^{\overline{C}},\;\mathbf{z}\oplus 1^n$ are in $\Sigma_k$. The base case for the induction is true due to the argument of the previous section which underlies \Cref{alg:main}. Another way of stating this fact is that $H_C$ remains a subgroup of all intermediate subspaces $\Sigma_k$, when viewing $\Sigma_k$ as subgroups of $Z_2^n$.

    Finally, we use the fact that $H_C$ is the hidden subgroup defining this StateHSP, which means that the inner product $\mel{\psi^{\otimes t}}{R(\mathbf{z})}{\psi^{\otimes t}}$ remains invariant when taking $\mathbf{z} \longrightarrow \mathbf{z} + \mathbf{h}$, for $\mathbf{h}\in H_C = \{0^n, 1^C0^{\overline{C}}, 0^C1^{\overline{C}},1^n\}$. Therefore we can reformulate the outcome distribution derived above in terms of the cosets of $\Sigma_k$ by the hidden subgroup $H_C$:
    \begin{align*}
        \mathrm{P}[\mathbf{y}] &= \frac{1}{2^n} \sum_{\mathbf{h}\in H_C} (-1)^{\mathbf{y}\cdot\mathbf{h}}\sum_{\mathbf{z}\in \Sigma_k/H_C}(-1)^{\mathbf{y}\cdot\mathbf{z}}\mel{\psi^{\otimes t}}{R(\mathbf{z})}{\psi^{\otimes t}} \\
        &= \frac{1+(-1)^{\mathbf{y}\cdot 1^C0^{\overline{C}}}}{2}\frac{1+(-1)^{\mathbf{y}\cdot 0^C1^{\overline{C}}}}{2}\frac{1}{2^{n-2}}\sum_{\mathbf{z}\in \Sigma_k/H_C}(-1)^{\mathbf{y}\cdot\mathbf{z}}\mel{\psi^{\otimes t}}{R(\mathbf{z})}{\psi^{\otimes t}}\\
        &=\frac{\delta_{\mathbf{y}\in H_C^\perp}}{2^{n-2}}\sum_{\mathbf{z}\in \Sigma_k/H_C}(-1)^{\mathbf{y}\cdot\mathbf{z}}\mel{\psi^{\otimes t}}{R(\mathbf{z})}{\psi^{\otimes t}},
    \end{align*}
    such that all measurement outcomes lie in the cut subspace $H_C^\perp=\{\mathbf{y}\in\Z_2^n\,:\,\mathbf{y}\cdot 1^C0^{\overline{C}}=\mathbf{y}\cdot 0^C1^{\overline{C}}=0\;\mathrm{mod}\;2\}$ by a similar mechanism as in the previous algorithm.
\end{proof}

\begin{fact}\label{fact:NewSamplesLinearlyIndependent}
    When using $t=O(1/\epsilon^2)$ state copies per sample, each new sample $\mathbf{y}^{(k)}$ is outside of the subspace $\Sigma_k^\perp$ with constant probability.
\end{fact}

\begin{proof}
    We can use the derived outcome distribution from the previous \Cref{fact:NewSamplesAreInCutSubspace} to express the probability that a new sample is not in the subspace $\Sigma_k^\perp$ by its complement:
    \begin{align*}
        \mathrm{P}\left[\mathbf{y}\notin \Sigma_k^\perp\right] &= 1 - \sum_{\mathbf{y}\in \Sigma_k^\perp} \mathrm{P}[\mathbf{y}] \\
        &= 1 - \frac{1}{2^n} \sum_{\mathbf{y}\in \Sigma_k^\perp} \sum_{\mathbf{z}\in \Sigma_k} (-1)^{\mathbf{y}\cdot\mathbf{z}}\Tr[\psi_\mathbf{z}]^{t/2} \tag{In terms of purities, as in \Cref{fact:HCPOutput}}\\
        &= 1 - \frac{1}{2^{n-k+1}}\sum_{\mathbf{z}\in \Sigma_k}\Tr[\psi_\mathbf{z}]^{t/2} \tag{Since $\mathbf{y}\cdot\mathbf{z}=0$}\\
        &= 1 - \frac{1}{2^{n-k-1}}\sum_{\mathbf{z}\in \Sigma_k/H_C}\Tr[\psi_\mathbf{z}]^{t/2} \tag{Organizing the sum by cosets},
    \end{align*}
where in the last line we split the sum over the $H_C$ cosets of $\Sigma_k$, using the findings from the proof of \Cref{fact:NewSamplesAreInCutSubspace} outlined above. Next, using \Cref{fact:epsilonpurity} to bound all nontrivial purities leads to: 
    \begin{align*}
       \mathrm{P}\left[\mathbf{y}\notin \Sigma_k^\perp\right] &= 1 - \frac{1}{2^{n-k-1}}\left(1 + \sum_{0^n\neq \mathbf{z}\in \Sigma_k/H_C}\Tr[\psi_\mathbf{z}]^{t/2}\right)\tag{Separating the zero term}\\
        &\geq 1 - \frac{1}{2^{n-k-1}}\left(1 + (2^{n-k-1}-1)(1-\epsilon^2)^{t/2}\right) \tag{Using \Cref{fact:epsilonpurity}}\\
        &\geq \frac{1}{2}\left(1-(1-\epsilon^2)^{t/2}\right) \tag{Since $k\in[n-2]$}.
    \end{align*}
Therefore, a choice of $t=O(1/\epsilon^2)$ makes this probability at least a constant, which suffices for the purpose of \Cref{alg:adaptive}.
\end{proof}

\section{The special case of Haar-random states: proof of \Cref{thm:haar}}\label{sec:Haar}

In this section, we will study the hidden cut problem when the factor states are promised to be sampled independently from the Haar measure. Intuitively, Haar-random states would be at least a constant distance away from product states with high probability, such that \Cref{alg:adaptive} of the previous section can be applied to find the cut given $O(n)$ state copies. While we do not improve on the number of state copies required (and beyond a possible factor of $\log n$, no improvement should be possible at all, given the decision lower bound of \cite{jones2024testing}), in this section we show how a careful analysis can reduce the other algorithmic requirements. Specifically, instead of running our adaptive algorithm (\Cref{alg:adaptive}), we show that our first, conceptually simpler non-adaptive algorithm (\Cref{alg:main}) suffices in this case. By taking advantage of the properties of the Haar measure, we show that in this case \Cref{alg:main} finds the cut with minimal requirements, involving circuits of constant depth (as opposed to depth $O(n^2)+O(\log\epsilon^{-1})$ as required by \Cref{alg:adaptive}) acting on only two state copies at a time:

\begin{customtheorem}{\ref{thm:haar}}[Hidden cut algorithm with Haar-random states --- restated]
    Under the stronger promise of Haar-random factor states, the hidden cut can be found by the version of \Cref{alg:main} using only $O(n)$ copies of the input state, by running circuits of constant depth which coherently access only $t=2$ state copies at a time.
\end{customtheorem}

To prove this result, we will have to further analyze the details of the StateHSP Fourier sampling distribution. In particular, we will relax the strong requirement of negligible relative error between the StateHSP and HSP Fourier sampling distributions \eqref{eq:relativenegligibleerror} used in the previous sections. The Haar measure toolkit will nonetheless provide enough analytic control over the resulting distributions. Many of the technical details will be delegated to Appendices \ref{app:selfaverage} and \ref{app:NonUniformSimons}, but this section will contain the main workflow behind the proof of \Cref{thm:haar}.

To start, define the Fourier purity probabilities generated by an $n$-qubit state $\ket{\phi}$ as:
\begin{equation}\label{eq:FourierPurityProbabilities}
        \mathrm{P}(\mathbf{y};\,\phi) \equiv \frac{1}{2^n}\sum_{\mathbf{x}\in\Z_2^n}(-1)^{\mathbf{y}\cdot \mathbf{x}}\Tr[\phi_\mathbf{x}^2],\quad\text{ where }\mathbf{y}\in\Z_2^n\,,
\end{equation}
where $\Tr[\phi_\mathbf{x}^2]$ is the purity of $\ket{\phi}$ across the cut represented by the binary vector $\mathbf{x}\in\Z_2^n$, i.e.\  when tracing out the qubits in the set $S_\mathbf{x}=\{i\in[n]\,:\,x_i=1\}$. Notice that this is similar to the probabilities studied in the previous section, except that the number of copies is fixed to $t=2$.

The central fact is that these quantities self-average in a strong sense under Haar-random states, as formalized in the following lemma: 
\begin{lemma}[Self-averaging of Fourier sampling distribution]\label{lemma:fourierselfconcentrate}
    With high probability over the choice of a Haar-random state $\ket{\phi}$, the Fourier probabilities self-concentrate:
    \begin{align}
        \mathrm{P}(\mathbf{y};\,\phi) &= \E_{\psi\sim{\Haar[(\C^2)^{\otimes n}]}} \mathrm{P}(\mathbf{y};\,\psi)\left(1+\negl(n)\right)\\
                  &= \delta_{\mathbf{y}\cdot1^n=0\;\mathrm{mod}\;2}\;\frac{2\cdot 3^{n - \abs{\mathbf{y}}}}{2^n(2^n+1)}\;\left(1+\negl(n)\right)\,.
    \end{align}
    Specifically, with probability at least $1-\delta$:
    \begin{equation}
        \forall \mathbf{y}\in\Z_2^n:\quad \abs{\frac{\mathrm{P}(\mathbf{y};\,\phi)}{\E\limits_{\psi\sim{\Haar[(\C^2)^{\otimes n}]}} \mathrm{P}(\mathbf{y};\,\psi)} - 1}\leq \delta^{-1/2}\cdot 3^{-n/2+o(1)}\,.
    \end{equation}
    Therefore, we can choose $\delta=\kappa^{-n}$ for any constant $\kappa \in (1,3)$ to satisfy the conclusion.
\end{lemma}
\begin{proof}
    The proof follows from second-order tail bounds applied to the covariance of the internal purities of a Haar-random state. We delegate the proof details to \Cref{app:selfaverage}.
\end{proof}

Applying this fact to the factor states making up the separable input state immediately leads to the following modification of \Cref{fact:HCPOutput} in the case of Haar-random states, when we restrict the number of copies to $t=2$:
\begin{fact}\label{fact:HaarTwoCopyOutcomeDistr}
    Consider the hidden cut problem with input state $\ket{\psi}=\ket{\phi_1}_C\otimes\ket{\phi_2}_{\overline{C}}$ separable across cut $C\in\binom{[n]}{n/2}$, and assume the factor states $\ket{\phi_{1,2}}$ are independent Haar-random states on $n/2$ qubits. Then, with high probability over the Haar-random samples $\ket{\phi_{1,2}}$, the Fourier sampling probability distribution with $t=2$ copies of the state is:
    \begin{align}\label{eq:TwoCopyHaarOutcomeDistribution}
        \mathrm{P}_{\mathrm{StateHSP}_{\psi,2}}[\mathbf{y}_1^C\mathbf{y}_2^{\overline{C}}] &=  \mathrm{P}(\mathbf{y}_1;\,\phi_1)\;\mathrm{P}(\mathbf{y}_2;\,\phi_2)\\
        &= \frac{\delta_{\abs{\mathbf{y}_1}\text{ even}}\,\delta_{\abs{\mathbf{y}_2}\text{ even}}}{2^{n-2}}\times\frac{3^{n-\abs{\mathbf{y}_1}-\abs{\mathbf{y}_2}}}{\left(2^{n/2} + 1\right)^2}\;(1+\negl(n))\,,
    \end{align}
    where the first factor is the associated HSP Fourier distribution $\mathrm{P}_\mathrm{HSP}[\mathbf{y}]$ defined in \eqref{eq:benchmarHSPdistr}, which is uniform over the $(n-2)$-dimensional Boolean subspace $H_C^\perp$ induced by the cut.

    This distribution is equivalent (up to the negligible relative correction) to producing outcomes $\mathbf{y}$ by the following rejection sampling protocol: produce a sample $\mathbf{y}\in\Z_2^n$ by independently sampling each bit from a Bernoulli distribution $y_1,\dots,y_n\sim \Ber(1/4)$; keep the sample $\mathbf{y}$ if it lies inside the cut subspace $H_C^\perp$, and sample again otherwise.
\end{fact}
We observe that, while the HSP distribution is uniform over the cut subspace $H_C^\perp$, the distribution \eqref{eq:TwoCopyHaarOutcomeDistribution} derived above is still supported inside the cut subspace $H_C^\perp$, however it is non-uniform since it skews towards smaller-weight outcomes. It remains to show that this modification does not significantly impact the number of samples required to accumulate a complete basis of the cut subspace $H_C^\perp$:
\begin{customtheorem}{\ref{thm:haar}}[Hidden cut algorithm with Haar-random states --- restated]
    When the input state $\ket{\psi}=\ket{\phi_1}_C\otimes\ket{\phi_2}_{\overline{C}}$ is a product of two $(n/2)$-qubit Haar-random factor states $\ket{\phi_1},\ket{\phi_2}$, a variation of the hidden cut algorithm finds the hidden cut with constant probability and using only $O(n)$ copies of the state, involving circuits of constant depth which coherently access only two state copies at a time.
\end{customtheorem}
{\bf\em Proof outline.}
    We delegate the full proof to \Cref{app:NonUniformSimons}, but we outline the argument here. The standard form of Simon's algorithm, including the modification relevant for \Cref{thm:main} above, relies on output distributions which are uniformly supported inside the hidden subspace. This uniformity condition makes it easy to compute the probabilities involved in the ``basis coupon collection'' process, showing that $n-k$ independent random samples can form a complete basis for an $(n-k)$-dimensional hidden subspace with constant probability. However, we are interested in the special case of purity Fourier sampling with a hidden cut state made of Haar-random factor states, which produces $n$-bit string outputs from the $\mathrm{P}_{\mathrm{StateHSP}_{\psi,2}}$ distribution defined above in \eqref{eq:TwoCopyHaarOutcomeDistribution}. As mentioned in \Cref{fact:HaarTwoCopyOutcomeDistr}, the resulting samples are not uniformly distributed inside the subspace, but they are skewed towards shorter-weight strings; the distribution is equivalent to rejection-sampling from an entrywise Bernoulli with $1/4$ probability of returning one, and keeping the sample if it lies inside the cut subspace $H_C^\perp$. The simple mathematics of Simon's coupon collection does not work anymore since the uniform assumption is violated. Our goal is to show that nonetheless, a similar conclusion still holds in this non-uniform case, such that a basis for the cut subspace can be collected in $O(n)$ samples.
    
    In fact, we will prove a slightly weaker form of the necessary basis coupon collection, but one which places us within an $O(1)$ distance to the full answer. Recall that the hidden cut subspace $H_C^\perp\subset \Z_2^n$ is of dimension $n-2$. Specifically, we will show that $n-3$ independent samples from the desired distribution \eqref{eq:TwoCopyHaarOutcomeDistribution} are linearly independent with a probability of at least one half. This is only one false cut direction away from finding the true cut. Specifically, the nullspace of the matrix $Y\in\Z_2^{(n-3)\times n}$ whose rows are the linearly independent samples has dimension three, and will contain eight vectors; two of them are the trivial cuts $0^n$ and $1^n$, leaving a number of six non-trivial candidate cuts. Two of the six candidate cuts are the two equivalent cut strings $1^C0^{\overline{C}}$ and $1^{\overline{C}}0^{C}$. Checking the six candidate cuts can be done by cut-specific SWAP tests, each SWAP test requiring a constant number of copies for a constant success probability guarantee. This is enough to show that the true cut can be found with a constant probability using $O(n)$ state copies.

    We remark that numerical evidence strongly suggests that $n-2$ i.i.d. samples from the distribution \eqref{eq:TwoCopyHaarOutcomeDistribution} form a complete basis for the cut subspace with constant probability, such that in practice the standard Simon's basis coupon collection routine succeeds in this case as well without the need to explicitly find the last basis vector via SWAP tests. We leave a formal proof of this technical conjecture about Boolean random matrix theory to future work.
\hfill\qed

\section{The many-cut case}\label{sec:HiddenManyCut}

The hidden cut algorithms from the previous section also apply naturally to the more general setting in which the input state is separable into two unequal subsystems, or indeed into more than two subsystems across an arbitrary set partition of the qubits, which we will refer to as the `hidden many-cut problem'. Finding the cut means identifying the set partition $C_1\sqcup \dots \sqcup C_m=[n]$ (where the $\sqcup$ symbol stands for disjoint union). Just as before, the number of copies used coherently to generate each Fourier sample will be chosen such that the corresponding StateHSP problem produces a similar outcome as the benchmark HSP distribution, up to negligible relative corrections in $n$. The Simon-like target distribution will now be supported on an $(n-m)$-dimensional subspace of $\Z_2^n$ orthogonal to the cut strings, where $m$ is the number of parts in the cut:
\begin{align}\label{eq:SimonManyCut}
        \mathrm{P}_{\mathrm{HSP}}[\mathbf{y}_1^{C_1}\dots\mathbf{y}_m^{C_m}]=\left\{
        \begin{array}{lr}
            2^{-n+m} & \text{if }\mathbf{y}_k\cdot1^{\abs{C_k}}=0\;\mathrm{mod}\;2,\;\text{for all}\;k\in[m] \\
            0 & \text{otherwise.}
        \end{array}
        \right.
\end{align}
Here, the notation $\mathbf{y}_1^{C_1}\dots\mathbf{y}_m^{C_m}$ represents the $n$-bit string with $\mathbf{y}_1\in\Z_2^{\abs{C_1}}$ in the positions in $C_1$, $\mathbf{y}_2\in\Z_2^{\abs{C_2}}$ in the positions in $C_2$, etc.\  A similar analysis bounding the contributions from possible false cuts applies at the level of the $m$ factor states.

Additionally, we note that the case of interest for our algorithm involves many-cuts for which all the parts are bigger than a constant. Otherwise, a naïve brute-force approach involving sequential SWAP tests of all possible qubit combinations of constant size can be peformed to iteratively discover the parts of the cut in polynomial time.
\begin{customcorollary}{\ref{cor:manycut}}[Algorithm for the many-cut problem --- restated]
Assume an $n$-qubit state $\ket{\psi}$ is separable across an unknown set partition into $m$ parts $C_1\sqcup\dots\sqcup C_m=[n]$:
\begin{equation}
    \ket{\psi} = \ket{\phi_1}_{C_1}\otimes\dots\otimes\ket{\phi_m}_{C_m}\,.
\end{equation}
Then the set partition/`many-cut' $\{C_1,\dots,C_m\}$ can be identified in polynomial time:
\begin{enumerate}
    \item[(a)] If the factor states are promised to be $\epsilon$-far in trace distance from separable, then \Cref{alg:adaptive} can identify the many-cut with $O(n/\epsilon^2)$ total number of state copies, with constant success probability. The algorithm runs circuits of depth $O(n^2)+O(\log\epsilon^{-1})$ on $O(1/\epsilon^2)$ state copies at a time.
    \item[(b)] If the factor states are Haar-random, then \Cref{alg:main} can identify the many-cut with $O(n)$ state copies with high probabulity, provided the additional constraint that the cut parts are superlogarithmic in size: $\min_{k\in[m]}\abs{C_k}> \omega(\log n)$. The algorithm runs circuits of constant depth, acting on two state copies at a time.
\end{enumerate}
\end{customcorollary}
{\bf\em Proof.}
    The proof follows the same logic as \Cref{thm:main} (see \Cref{sec:AdaptiveSubspaceAlgorithm}) and \Cref{thm:haar} (see \Cref{sec:Haar}). The outcome probability of the non-adaptive Fourier sampling circuit (i.e. the updated version of \eqref{eq:NonAdaptiveOutcomeDistribution} in the case of $m$ partitions) becomes:
    \begin{align}
        \mathrm{P}_{\mathrm{StateHSP}_{\psi,t}}[\mathbf{y}_1^{C_1}\dots\mathbf{y}_m^{C_m}] =  \mathrm{P}_{\mathrm{HSP}}[\mathbf{y}_1^{C_1}\dots\mathbf{y}_m^{C_m}]\prod_{k\in[m]}\left[\sum_{\substack{\mathbf{x}_k\in\Z_2^\abs{C_k}\\x_{k,1}=0}}(-1)^{\mathbf{y}_k\cdot\mathbf{x}_k}\Tr[\phi_{k,\mathbf{x}_k}^2]^{t/2}\right]\,,
    \end{align}
    where the benchmark HSP distribution is the one defined in \eqref{eq:SimonManyCut}. A similar triangle inequality as in \Cref{fact:HCPOutput} gives us that:
    \begin{equation}
        \abs{\mathrm{P}_{\mathrm{StateHSP}_{\psi,t}}[\mathbf{y}_1^{C_1}\dots\mathbf{y}_m^{C_m}] - \mathrm{P}_{\mathrm{HSP}}[\mathbf{y}_1^{C_1}\dots\mathbf{y}_m^{C_m}]}\leq\frac{1}{2^{n-m}}\left[\prod_{k\in[m]}\left(1 + \Delta_{\phi_k,t}\right)-1\right]\,,
    \end{equation}
    where we define, as before:
    \begin{equation}
        \Delta_{\phi_k,t} \equiv \sum_{\substack{\mathbf{x}_k\in\Z_2^\abs{C_k}\\\mathbf{x}_k\neq0^{\abs{C_k}}\\x_{k,1}=0}}\Tr[\phi_{k,\mathbf{x}_k}^2]^{t/2}\,.
    \end{equation}
    For the non-adaptive \Cref{alg:main} to efficiently find the cut, it is sufficient to choose $t$ such that all $\Delta_{\phi_k,t}$ are negligible in $n$.
    \begin{enumerate}
        \item[(a)]  If $\ket{\phi}_k$ is at least $\epsilon$-far from all separable states, then \Cref{fact:epsilonpurity} gives us that $\Tr[\phi_{k,\mathbf{x}}^2]\leq 1-\epsilon^2$ for any nontrivial cut $\mathbf{x}$, leading to an upper bound via triangle inequality:
        \begin{align}
            \Delta_{\phi_k,t}&\leq(2^{\abs{C_k}}-1)(1-\epsilon^2)^{t/2}\\
            &\leq2^{\abs{C_*}-O(t\epsilon^2)}\,.
        \end{align}
        Here, we defined the largest cut component size as $\abs{C_*}\equiv \max_{k\in[m]}\abs{C_k}$. The above can be made negligible in $n$ if $t=O(\abs{C_*}/\epsilon^2)+O(\log^2 n)$. This suffices to show how the non-adaptive \Cref{alg:main} applies to finding the many-cut with high probability with $O(n/\epsilon^2)$.

        Applying the adaptive \Cref{alg:adaptive} to the many-cut case is similarly straightforward. Just as in section \Cref{sec:AdaptiveSubspaceAlgorithm}, the adaptive algorithm will accumulate $n-m$ linearly independent vectors in the $(n-m)$-dimensional cut subspace $H_C^\perp$ by the adaptive Fourier sampling method. New samples are always in the cut subspace by the StateHSP subgroup symmetry; a number of copies $t=O(1/\epsilon^2)$ suffices to lower-bound the probability that a new sample is linearly independent with respect to previously collected samples by a constant. This in turn is enough to make sure that the algorithm terminates and succeeds to find the cut with a constant overall success probability.
        
        \item[(b)] The results of \Cref{sec:Haar} apply in this case as well, since each part $C_k$ of the partition incurs a relative error of size $O(2^{-\abs{C_k}})$. Since there are at most $n$ parts, the overall corrections remain negligible as long as each individual correction remains negligible, i.e.\  if all parts are more than logarithmic in size, i.e.\  $\min_{k\in[m]}\abs{C_k}>\omega(\log n)$.
        \hfill\qed
    \end{enumerate}

We conjecture that the stricter requirement of superlogarithmic part size in the case of Haar-random factor states can be removed with a more careful accounting of the concentration properties of Haar-random states.

Finally, we remark that one does not necessarily need to  know the number of unentangled parts which make up the input state a priori, since this is not a parameter in our algorithms. In fact, one can efficiently infer the number of parts in the many-cut by analyzing the linear independence of the obtained Fourier samples: if the input state is a product of $m$ factor states, then the rank of the accumulated samples will plateau at a value of $n-m$.

\section{Discussion, applications, and open questions}\label{sec:Discussion}

\subsection{Applications: cryptography and pseudorandomness}

As mentioned in the introduction, our hidden cut algorithm provides a no-go result for certain recursive constructions of pseudorandom states. In particular, our algorithm shows that a product of pseudorandom states across a random cut is not itself pseudorandom.
More broadly, our hidden cut algorithm prohibits pseudorandom state constructions with zero entanglement across any partition of the qubits.
However, certain generalizations of these constructions are not ruled out by our algorithm.
For example, consider a nonzero-entropy hidden cut state, for example a rank-two state of the form:
\begin{equation}
\ket{\psi} \approxeq \frac{1}{\sqrt{2}} \left(\ket{\alpha_1}_C\otimes \ket{\beta_1}_{\overline{C}} +  \ket{\alpha_2}_C\otimes \ket{\beta_2}_{\overline{C}} \right),
\end{equation}
where the tensor products are taken across a random cut $C \subset [n]$, and the sub-states $\ket{\alpha_{1,2}}$ and $\ket{\beta_{1,2}}$ are pairs of orthogonal pseudorandom states.
Such a state would have constant, but nonzero entanglement entropy across the cut; the cut remains information-theoretically detectable, and verifying the cut can still be achieved with only a constant number of copies via a standard SWAP test. Interestingly, running our algorithm on such an input state would fail to identify the hidden cut $C$. Specifically, the resulting Fourier sampling distribution obtained becomes a noisy version of the Simon's problem, such that the samples now have a constant, nonzero probability of lying outside of the cut subspace $H_C^\perp$. In other words, solving for the cut subspace would require solving a noisy system of linear equations over $\Z_2^n$ with a constant noise rate, i.e. it would require solving a version of the learning parity with noise ($\mathsf{LPN}$) problem, which is conjectured to be cryptographically hard. We leave open the question of designing an algorithm for finding the hidden cut in the nonzero cut entropy scenario, or alternatively of producing further evidence that the problem is computationally hard.

We also note the hidden cut problem might be useful for constructions of quantum money.
Here, the goal is to produce `banknote' states which are difficult to copy but easy to verify. 
One could imagine a quantum money scheme based on the hidden cut problem, in which the cut serves as the secret key, and the banknote would be composed of only a constant number of copies of the separable state.
This means our algorithm cannot be used to find the cut, since it would require a linear number of copies.
Therefore, it is possible the cut could be cryptographically protected.
On the other hand, verifying the cut only requires only a constant number of copies by a standard fixed-cut SWAP test. It remains an open question whether such a quantum money construction can be made public-key compatible. We remark that the factor states themselves, not just the location of the cut, could potentially serve a cryptographic function.

\subsection{The StateHSP framework}

Motivated by the hidden cut problem, in \Cref{sec:StateHSP} we introduced a state version of the hidden subgroup problem as a flexible framework for problems with state input which feature a hidden symmetry subgroup. 
An natural question is whether the StateHSP framework can be used to derive quantum algorithms for other quantum information tasks. 
One source of inspiration could be tasks in unitary complexity theory \cite{rosenthal2021interactive,metger2023stateqip,bostanci2023unitary}.
Alternatively, in the other direction there is the question as to whether StateHSP can give rise to cryptographic primitives via information-computation gaps.
As a corollary of our work is that StateHSP is information-theoretically solvable with enough copies and orthogonality allowance (see \Cref{cor:StateHSPInformationTheoretical}), which opens the possibility of information-computation gaps in the general case. 
Another potential avenue is to consider cases where the mechanism for orthogonality amplification by preparing poly-many state copies is not available --- in such cases, the StateHSP might become hard, while also resisting the reduction to HSP via orthogonality amplification.

\subsection{Entanglement features}

We note that a combinatorial view considering all internal purities of a given state $\psi$ organized as a so-called {\em entanglement feature} vector $\ket{W_\psi}\propto \sum_{\mathbf{x}\in\Z_2^n}\Tr[\psi_\mathbf{x}^2]\ket{\mathbf{x}}$ appears in the condensed matter literature\footnote{We thank Matteo Ippoliti for pointing out this connection.} \cite{you2018machine,you2018entanglement,fan2021self}.
While constructing the entanglement feature state $\ket{W_\psi}$ from the input state $\ket{\psi}$ seems to be generally hard, our hidden cut algorithm is able to indirectly manipulate this quantity. In particular, we remark that our \Cref{alg:main} is able to perform Fourier sampling on the moments of the entanglement feature vector (see equation \eqref{eq:entanglementfeature}). We leave it as an open question to further explore this connection.

\section*{Acknowledgments}

We thank Roozbeh Bassirian, Bill Fefferman, Soumik Ghosh, Patrick Hayden, Matteo Ippoliti, Fernando Jeronimo, Benjamin Jones, Ashley Montanaro, Henry Yuen, Chenyi Zhang, and Jack Zhou for insightful discussions.
A.B. and T.G.T. were supported in part by the U.S. DOE Office of Science under Award Number DE-SC0020377. A.B. was supported in part by the DOE QuantISED grant DE-SC0020360 and by the AFOSR under grants FA9550-21-1-0392 and FA9550-24-1-0089. J.W. is supported in part by NSF CAREER award CCF-2339711.

\newpage
\bibliographystyle{alpha}
\bibliography{arxiv}

\newcommand{\etalchar}[1]{$^{#1}$}
\begin{thebibliography}{GHMW13}

\bibitem[ABF{\etalchar{+}}24]{aaronson2024quantum}
Scott Aaronson, Adam Bouland, Bill Fefferman, Soumik Ghosh, Umesh Vazirani, Chenyi Zhang, and Zixin Zhou.
\newblock {Quantum Pseudoentanglement}.
\newblock In {\em 15th Innovations in Theoretical Computer Science Conference (ITCS 2024)}, 2024.
\newblock \href{https://arxiv.org/abs/2211.00747}{\nolinkurl{arXiv:2211.00747}}.

\bibitem[BDJ99]{baik1999distribution}
Jinho Baik, Percy Deift, and Kurt Johansson.
\newblock On the distribution of the length of the longest increasing subsequence of random permutations.
\newblock {\em Journal of the American Mathematical Society}, 12(4):1119--1178, 1999.
\newblock \href{https://arxiv.org/abs/math/9810105}{\nolinkurl{arXiv:math/9810105}}.

\bibitem[Bea97]{beals1997quantum}
Robert Beals.
\newblock {Quantum computation of Fourier transforms over symmetric groups}.
\newblock In {\em Proceedings of the twenty-ninth annual ACM symposium on Theory of computing}, pages 48--53, 1997.

\bibitem[BEM{\etalchar{+}}23]{bostanci2023unitary}
John Bostanci, Yuval Efron, Tony Metger, Alexander Poremba, Luowen Qian, and Henry Yuen.
\newblock Unitary complexity and the {Uhlmann} transformation problem.
\newblock {\em arXiv preprint \href{https://arxiv.org/abs/2306.13073}{\nolinkurl{arXiv:2306.13073}}}, 2023.

\bibitem[BO20]{badescu2020lower}
Costin B{\u{a}}descu and Ryan O’Donnell.
\newblock Lower bounds for testing complete positivity and quantum separability.
\newblock In {\em Latin American Symposium on Theoretical Informatics}, pages 375--386. Springer, 2020.
\newblock \href{https://arxiv.org/abs/1905.01542}{\nolinkurl{arXiv:1905.01542}}.

\bibitem[CVD10]{childs2010quantum}
Andrew~M Childs and Wim Van~Dam.
\newblock Quantum algorithms for algebraic problems.
\newblock {\em Reviews of Modern Physics}, 82(1):1--52, 2010.
\newblock \href{https://arxiv.org/abs/0812.0380}{\nolinkurl{arXiv:0812.0380}}.

\bibitem[Dia88]{diaconis1988group}
Persi Diaconis.
\newblock {\em {Group Representations in Probability and Statistics}}, volume~11.
\newblock Institute of Mathematical Statistics, 1988.

\bibitem[EHK04]{ettinger2004quantum}
Mark Ettinger, Peter H{\o}yer, and Emanuel Knill.
\newblock The quantum query complexity of the hidden subgroup problem is polynomial.
\newblock {\em Information Processing Letters}, 91(1):43--48, 2004.
\newblock \href{https://arxiv.org/abs/quant-ph/0401083}{\nolinkurl{arXiv:quant-ph/0401083}}.

\bibitem[FO24]{flammia2024quantum}
Steven~T Flammia and Ryan O'Donnell.
\newblock Quantum chi-squared tomography and mutual information testing.
\newblock {\em Quantum}, 8:1381, 2024.
\newblock \href{https://arxiv.org/abs/2305.18519}{\nolinkurl{arXiv:2305.18519}}.

\bibitem[FVVY21]{fan2021self}
Ruihua Fan, Sagar Vijay, Ashvin Vishwanath, and Yi-Zhuang You.
\newblock Self-organized error correction in random unitary circuits with measurement.
\newblock {\em Physical Review B}, 103(17):174309, 2021.
\newblock \href{https://arxiv.org/abs/2002.12385}{\nolinkurl{arXiv:2002.12385}}.

\bibitem[Gao15]{Gao15}
Jingliang Gao.
\newblock Quantum union bounds for sequential projective measurements.
\newblock {\em Physical Review A}, 92(5):052331, 2015.

\bibitem[GC01]{gottesman2001quantum}
Daniel Gottesman and Isaac Chuang.
\newblock Quantum digital signatures.
\newblock {\em arXiv preprint \href{https://arxiv.org/abs/quant-ph/0105032}{\nolinkurl{quant-ph/0105032}}}, 2001.

\bibitem[Gha08]{gharibian2008strong}
Sevag Gharibian.
\newblock {Strong NP-hardness of the quantum separability problem}.
\newblock {\em arXiv preprint \href{https://arxiv.org/abs/0810.4507}{\nolinkurl{arXiv:0810.4507}}}, 2008.

\bibitem[GHMW13]{gutoski2013quantum}
Gus Gutoski, Patrick Hayden, Kevin Milner, and Mark~M Wilde.
\newblock {Quantum interactive proofs and the complexity of separability testing}.
\newblock {\em arXiv preprint \href{https://arxiv.org/abs/1308.5788}{\nolinkurl{arXiv:1308.5788}}}, 2013.

\bibitem[GSVV04]{GSVV04}
Michelangelo Grigni, Leonard Schulman, Monica Vazirani, and Umesh Vazirani.
\newblock Quantum mechanical algorithms for the nonabelian hidden subgroup problem.
\newblock {\em Combinatorica}, 1(24):137--154, 2004.

\bibitem[Har13]{harrow2013church}
Aram~W Harrow.
\newblock The church of the symmetric subspace.
\newblock {\em arXiv preprint \href{https://arxiv.org/abs/1308.6595}{\nolinkurl{arXiv:1308.6595}}}, 2013.

\bibitem[HLM17]{harrow2017sequential}
Aram~W Harrow, Cedric Yen-Yu Lin, and Ashley Montanaro.
\newblock Sequential measurements, disturbance and property testing.
\newblock In {\em Proceedings of the Twenty-Eighth Annual ACM-SIAM Symposium on Discrete Algorithms}, pages 1598--1611. SIAM, 2017.
\newblock \href{https://arxiv.org/abs/1607.03236}{\nolinkurl{arXiv:1607.03236}}.

\bibitem[HM13]{harrow2013testing}
Aram~W Harrow and Ashley Montanaro.
\newblock Testing product states, quantum {Merlin-Arthur} games and tensor optimization.
\newblock {\em Journal of the ACM (JACM)}, 60(1):1--43, 2013.
\newblock \href{https://arxiv.org/abs/1001.0017}{\nolinkurl{arXiv:1001.0017}}.

\bibitem[HRTS03]{hallgren2003hidden}
Sean Hallgren, Alexander Russell, and Amnon Ta-Shma.
\newblock The hidden subgroup problem and quantum computation using group representations.
\newblock {\em SIAM Journal on Computing}, 32(4):916--934, 2003.

\bibitem[JLS18]{ji2018pseudorandom}
Zhengfeng Ji, Yi-Kai Liu, and Fang Song.
\newblock Pseudorandom quantum states.
\newblock In {\em Advances in Cryptology--CRYPTO 2018: 38th Annual International Cryptology Conference, Santa Barbara, CA, USA, August 19--23, 2018, Proceedings, Part III 38}, pages 126--152. Springer, 2018.
\newblock \href{https://eprint.iacr.org/2018/544}{\nolinkurl{iacr:2018/544}}.

\bibitem[JM24]{jones2024testing}
Benjamin~DM Jones and Ashley Montanaro.
\newblock Testing multipartite productness is easier than testing bipartite productness.
\newblock {\em arXiv preprint \href{https://arxiv.org/abs/2406.16827}{\nolinkurl{arXiv:2406.16827}}}, 2024.

\bibitem[LG17]{lockhart2017quantum}
Joshua Lockhart and Carlos E~Gonz{\'a}lez Guill{\'e}n.
\newblock Quantum state isomorphism.
\newblock {\em arXiv preprint \href{https://arxiv.org/abs/1709.09622}{\nolinkurl{arXiv:1709.09622}}}, 2017.

\bibitem[LMW24]{lombardi2024one}
Alex Lombardi, Fermi Ma, and John Wright.
\newblock A one-query lower bound for unitary synthesis and breaking quantum cryptography.
\newblock In {\em Proceedings of the 56th Annual ACM Symposium on Theory of Computing}, pages 979--990, 2024.
\newblock \href{https://arxiv.org/abs/2310.08870}{\nolinkurl{arXiv:2310.08870}}.

\bibitem[LRW23]{laborde2023testing}
Margarite~L LaBorde, Soorya Rethinasamy, and Mark~M Wilde.
\newblock Testing symmetry on quantum computers.
\newblock {\em Quantum}, 7:1120, 2023.
\newblock \href{https://arxiv.org/abs/2105.12758}{\nolinkurl{arXiv:2105.12758}}.

\bibitem[MdW13]{montanaro2013survey}
Ashley Montanaro and Ronald de~Wolf.
\newblock A survey of quantum property testing.
\newblock {\em arXiv preprint \href{https://arxiv.org/abs/1310.2035}{\nolinkurl{arXiv:1310.2035}}}, 2013.

\bibitem[MRR06]{moore2006generic}
Cristopher Moore, Daniel Rockmore, and Alexander Russell.
\newblock {Generic quantum Fourier transforms}.
\newblock {\em ACM Transactions on Algorithms (TALG)}, 2(4):707--723, 2006.

\bibitem[MRS08]{moore2008symmetric}
Cristopher Moore, Alexander Russell, and Leonard~J Schulman.
\newblock {The symmetric group defies strong Fourier sampling}.
\newblock {\em SIAM Journal on Computing}, 37(6):1842--1864, 2008.
\newblock \href{https://arxiv.org/abs/quant-ph/0501056}{\nolinkurl{quant-ph/0501056}}.

\bibitem[MY23]{metger2023stateqip}
Tony Metger and Henry Yuen.
\newblock {$\mathsf{stateQIP}=\mathsf{statePSPACE}$}.
\newblock In {\em 2023 IEEE 64th Annual Symposium on Foundations of Computer Science (FOCS)}, pages 1349--1356. IEEE, 2023.
\newblock \href{https://arxiv.org/abs/2301.07730}{\nolinkurl{arXiv:2301.07730}}.

\bibitem[NC10]{nielsen2010quantum}
Michael~A Nielsen and Isaac~L Chuang.
\newblock {\em {Quantum Computation and Quantum Information}}.
\newblock Cambridge University Press, 2nd edition, 2010.

\bibitem[OW16]{o2016efficient}
Ryan O'Donnell and John Wright.
\newblock Efficient quantum tomography.
\newblock In {\em Proceedings of the forty-eighth annual ACM symposium on Theory of Computing}, pages 899--912, 2016.
\newblock \href{https://arxiv.org/abs/1508.01907}{\nolinkurl{arXiv:1508.01907}}.

\bibitem[Reg04]{regev2004quantum}
Oded Regev.
\newblock Quantum computation and lattice problems.
\newblock {\em SIAM Journal on Computing}, 33(3):738--760, 2004.
\newblock \href{https://arxiv.org/abs/cs/0304005}{\nolinkurl{arXiv:cs/0304005}}.

\bibitem[RLW23]{rethinasamy2023quantum}
Soorya Rethinasamy, Margarite~L LaBorde, and Mark~M Wilde.
\newblock {Quantum Computational Complexity and Symmetry}.
\newblock {\em arXiv preprint \href{https://arxiv.org/abs/2309.10081}{\nolinkurl{arXiv:2309.10081}}}, 2023.

\bibitem[Roi96]{roichman1996upper}
Yuval Roichman.
\newblock Upper bound on the characters of the symmetric groups.
\newblock {\em Inventiones mathematicae}, 125:451--485, 1996.

\bibitem[RY22]{rosenthal2021interactive}
Gregory Rosenthal and Henry Yuen.
\newblock {Interactive Proofs for Synthesizing Quantum States and Unitaries}.
\newblock In {\em 13th Innovations in Theoretical Computer Science Conference (ITCS 2022)}, 2022.
\newblock \href{https://arxiv.org/abs/2108.07192}{\nolinkurl{arXiv:2108.07192}}.

\bibitem[Sim97]{simon1997power}
Daniel~R Simon.
\newblock On the power of quantum computation.
\newblock {\em SIAM Journal on Computing}, 26(5):1474--1483, 1997.

\bibitem[SW22]{soleimanifar2022testing}
Mehdi Soleimanifar and John Wright.
\newblock Testing matrix product states.
\newblock In {\em Proceedings of the 2022 Annual ACM-SIAM Symposium on Discrete Algorithms (SODA)}, pages 1679--1701. SIAM, 2022.
\newblock \href{https://arxiv.org/abs/2201.01824}{\nolinkurl{arXiv:2201.01824}}.

\bibitem[Wat00]{watrous2000succinct}
John Watrous.
\newblock Succinct quantum proofs for properties of finite groups.
\newblock In {\em Proceedings 41st Annual Symposium on Foundations of Computer Science}, pages 537--546. IEEE, 2000.
\newblock \href{https://arxiv.org/abs/cs/0009002}{\nolinkurl{arXiv:cs/0009002}}.

\bibitem[YG18]{you2018entanglement}
Yi-Zhuang You and Yingfei Gu.
\newblock {Entanglement features of random Hamiltonian dynamics}.
\newblock {\em Physical Review B}, 98(1):014309, 2018.
\newblock \href{https://arxiv.org/abs/1803.10425}{\nolinkurl{arXiv:1803.10425}}.

\bibitem[YYQ18]{you2018machine}
Yi-Zhuang You, Zhao Yang, and Xiao-Liang Qi.
\newblock Machine learning spatial geometry from entanglement features.
\newblock {\em Physical Review B}, 97(4):045153, 2018.
\newblock \href{https://arxiv.org/abs/1709.01223}{\nolinkurl{arXiv:1709.01223}}.

\bibitem[Zha24]{zhandry2023quantum}
Mark Zhandry.
\newblock {Quantum Money from Abelian Group Actions}.
\newblock In {\em 15th Innovations in Theoretical Computer Science Conference (ITCS 2024)}, 2024.
\newblock \href{https://arxiv.org/abs/2307.12120}{\nolinkurl{arXiv:2307.12120}}.

\end{thebibliography}

\newpage
\appendix

\section{Internal purity covariance of Haar-random states and the self-averaging of Fourier sampling distributions}\label{app:selfaverage}

\subsection{Proof of \Cref{lemma:fourierselfconcentrate}}

    The proof of \Cref{lemma:fourierselfconcentrate} will make use of the two auxiliary results \Cref{fact:Haar1} and \Cref{fact:FourierProbStats}, detailed below. Here, we import these facts to show how they lead to the conclusion of \Cref{lemma:fourierselfconcentrate}.

    The key quantity is the collection of purity Fourier sampling probabilities induced by a state $\ket{\phi}$ \eqref{eq:FourierPurityProbabilities}, defined as:
    \begin{equation}\label{eq:FourierPurityProbabilities2}
        \mathrm{P}(\mathbf{y};\,\phi) \equiv \frac{1}{2^n}\sum_{\mathbf{x}\in\Z_2^n}(-1)^{\mathbf{y}\cdot \mathbf{x}}\Tr[\phi_\mathbf{x}^2],\quad\text{ where }\mathbf{y}\in\Z_2^n\,.
    \end{equation}
    The goal is to study how these probabilities self-average when the state $\ket{\phi}$ is fixed to a typical sample from the Haar measure on $n$-qubit states, which we will denote by $\Haar_n$ to condense notation.

    \Cref{fact:FourierProbStats} gives us explicit expression for the mean and variance of a single purity Fourier sampling probability $\mathrm{P}(\mathbf{y};\,\phi)$, when $\ket{\phi}$ is sampled from the Haar measure on $n$-qubit states. Chebyshev's inequality then allows us to show how one of these quantities concentrates:
    \begin{align}
         \Pr_{\phi\sim\Haar_n}\left[\abs{\mathrm{P}(\mathbf{y};\,\phi) - \E_{\psi\sim\Haar_n} \mathrm{P}(\mathbf{y};\,\psi)} \geq \beta \E_{\psi\sim\Haar_n} \mathrm{P}(\mathbf{y};\,\psi)\right]  &\leq \frac{\Var_{\psi\sim\Haar_n}\mathrm{P}(\mathbf{y};\,\psi)}{\beta^2  \left(\E_{\psi\sim\Haar_n}\mathrm{P}(\mathbf{y};\,\psi)\right)^2}\\
         &\leq \frac{3^{\abs{\mathbf{y}}-n} \;(1+2^{-n})\;-2^{-2n+1}}{\beta^2(6 + 5\cdot 2^n + 4^n)}\,.
    \end{align}
    Here, the states are implicitly understood to be sampled from the Haar measure on $n$-qubit states. The notation $\abs{\mathbf{y}}$ denotes the Hamming weight of the bit-string $\mathbf{y}\in\Z_2^n$. We want a typicality statement about all of the probabilities $(P_\phi(\mathbf{y}))_{\mathbf{y}\in\Z_2^n}$ induced by a single state $\ket{\phi}$ sampled from the Haar measure. This can be obtained by a simple union bound over all the outcomes $\mathbf{y}\in\Z_2^n$, resulting in:
    \begin{align}
        \Pr_{\phi\sim\Haar_n}\left[\exists\, \mathbf{y}\in\Z_2^n\,:\,\abs{\mathrm{P}(\mathbf{y};\,\phi) - \E_{\psi\sim\Haar_n} \mathrm{P}(\mathbf{y};\,\psi)} \geq \beta \E_{\psi\sim\Haar_n} \mathrm{P}(\mathbf{y};\,\psi)\right] &\leq \frac{(2/3)^{n} + (4/3)^{n} - 2^{1-n}}{\beta^2(6 + 5 \cdot 2^n + 4^n)}\\
        &\leq \frac{3^{-n + o(1)}}{\beta^2}\,,
    \end{align}
    from which the conclusion follows. \hfill\qed

\subsection{The covariance of internal purities of Haar-random states}

The proof above invokes the following two facts involving properties of the second and fourth moments of the Haar measure. The first fact involves calculating the mean and covariance of the internal purities of a Haar-random state:

\begin{fact}\label{fact:Haar1} Let $\ket{\phi}$ be a state sampled from the $n$-qubit Haar measure. Let $\Tr[\phi_\mathbf{x}^2]$ denote the purity of $\ket{\phi}$ across the cut determined by the bit-string $\mathbf{x}$, i.e.\  by tracing out the qubits in $S_\mathbf{x}\equiv\{i\in[n]\;:\;x_i=1\}$. Then we have that the average purity across a cut is:
\begin{align}
    \mathfrak{p}_\mathbf{x} &\equiv \E_{\phi\sim\Haar_n}\Tr[\phi_\mathbf{x}^2]\\
    &=\frac{2^{-\abs{\mathbf{x}}} + 2^{-n+\abs{\mathbf{x}}}}{1+2^{-n}}\,.
\end{align}
Also, the covariance between pairs of purities is given by:
\begin{align}
     \Sigma_{\mathbf{x},\mathbf{x'}} &\equiv \E_{\phi\sim\Haar_n}\Tr[\phi_\mathbf{x}^2]\Tr[\phi_{\mathbf{x'}}^2] - \E_{\phi\sim\Haar_n}\Tr[\phi_\mathbf{x}^2]\E_{\phi\sim\Haar_n}\Tr[\phi_{\mathbf{x'}}^2]\\
     &= 2\frac{2^{\abs{\mathbf{x}\oplus \mathbf{x'}}} + 2^{n-\abs{\mathbf{x} \oplus \mathbf{x'}}}}{(2^n+1)(2^n+2) (2^n+3)} - \frac{2}{(2^n+2) (2^n+3)}\mathfrak{p}_\mathbf{x}\mathfrak{p}_\mathbf{x'}\,.\label{eq:covariancepurity}
\end{align}
\end{fact}
\begin{proof}
    Averages involving state purities can be turned into averages over the corresponding symmetric subspace starting with the following reformulation:
    \begin{equation}\label{eq:reformulationSym}
        \Tr[\phi_\mathbf{\mathbf{x}}^2] = \Tr[\bigotimes_{i:\,x_i=1}R_i((1\;2))\cdot \dyad{\phi}^{\otimes 2}]\,,
    \end{equation}
    which follows the notation from \Cref{sec:HiddenCut}. Specifically, $R_i(\sigma)$ applied the permutation $\sigma\in\bbS_t$ to the $i$-th qubit across the $t$ copies. The above equality involving the simple purity $\Tr[\phi_\mathbf{\mathbf{x}}^2]$ corresponds to the special case $t=2$. Next, we use the well-known equality between the Haar $t$-fold ensemble and the maximally mixed state over the $t$-fold symmetric subspace (see for example \cite{harrow2013church}), which in our notation translates to:
    \begin{equation}
        \E_{\phi\sim\Haar_n}\dyad{\phi}^{\otimes t} = \frac{(2^n - 1)!}{(2^n + t - 1)!}\sum_{\sigma \in \bbS_t}R_1(\sigma)\otimes \dots\otimes R_n(\sigma).
    \end{equation}
    Applying this fact, we have that the average purity is:
    \begin{align}
       \E_{\phi\sim\Haar_n}\Tr[\phi_\mathbf{\mathbf{x}}^2] &= \Tr[\bigotimes_{i:\,x_i=1}R_i((1\;2))\cdot  \E_{\phi\sim\Haar_n}\dyad{\phi}^{\otimes 2}] \\
       &= \frac{1}{2^n(2^n+1)} \sum_{\sigma\in\bbS_2} \Tr[\bigotimes_{i:\,x_i=1}R_i((1\;2)\cdot\sigma)\otimes\bigotimes_{j:\,x_j=0}R_j(\sigma)] \\
       &= \frac{1}{2^n(2^n+1)} \sum_{\sigma\in\bbS_2} \prod_{i:\,x_i=1}\Tr[R_i((1\;2)\cdot\sigma)]\prod_{j:\,x_j=0}\Tr[R_j(\sigma)] \\
       &= \frac{1}{2^n(2^n+1)}\sum_{\sigma\in\bbS_2}2^{\abs{\mathbf{x}}\cyc((1\;\;2)\cdot \sigma) + (n-\abs{\mathbf{x}})\cyc(\sigma)}\,,
    \end{align}
    where in the final line we used the fact that $\Tr[R_i(\sigma)]=2^{\cyc(\sigma)}$. Here, $\cyc(\sigma)$ denotes the number of cycles in the permutation $\sigma$. The explicit sum over $\sigma \in \bbS_2 = \{(1)(2),\;(1\;2)\}$ leads to the closed-form result:
    \begin{align}
    \mathfrak{p}_\mathbf{x} &\equiv \E_{\phi\sim\Haar_n}\Tr[\phi_\mathbf{x}^2]\\
    &=\frac{2^{-\abs{\mathbf{x}}} + 2^{-n+\abs{\mathbf{x}}}}{1+2^{-n}}\,.
    \end{align}
    To calculate the covariance entries we will need a fourth-moment calculation. This is because the same technique used in \eqref{eq:reformulationSym} above can be used to rewrite the product of two purities in terms of four copies of the state:
    \begin{equation}
         \Tr[\phi_\mathbf{\mathbf{x}}^2]\Tr[\phi_\mathbf{\mathbf{x'}}^2] = \Tr[\bigotimes_{i:\,x_i=1}R_i((1\;2))\bigotimes_{j:\,x'_j=1}R_i((3\;4))\cdot \dyad{\phi}^{\otimes 4}]\,,
    \end{equation}
    where in this case $t=4$ because the column-wise permutations $R_i(\sigma)$ represent $\sigma\in\bbS_4$. We apply the same workflow from above to translate the average over 4 copies of the Haar-random state $\ket{\phi}$ to a combinatorial sum over the symmetric group $\bbS_4$, leading to:
    \begin{equation}
        {\scriptstyle \E\limits_{\phi\sim\Haar_n} \Tr[\phi_\mathbf{x}^2]\Tr[\phi_{\mathbf{x'}}^2] = \frac{1}{2^n(2^n+1)(2^n+2)(2^n+3)}\sum\limits_{\sigma\in\bbS_4}2^{\abs{\mathbf{x}\setminus \mathbf{x'}}\cyc((1\;\;2)\sigma) + \abs{\mathbf{x}\cap \mathbf{x'}}\cyc((1\;\;2)(3\;\;4)\sigma)+\abs{\mathbf{x'}\setminus \mathbf{x}}\cyc((3\;\;4)\sigma)+(n-\abs{\mathbf{x}\cup \mathbf{x'}})\cyc(\sigma)}}\,.
    \end{equation}
    In the above, we introduced set notation for binary strings in a natural sense, meaning that $\mathbf{x}\setminus \mathbf{x'}\equiv\{i\in[n]\;:\;x_i=1\;\text{and}\;x'_i=0\}$, also $\mathbf{x}\cap \mathbf{x'}\equiv\{i\in[n]\;:\;x_i=1\;\text{and}\;x'_i=1\}$, as well as $\mathbf{x}\cup \mathbf{x'}=\{i\in[n]\;:\;x_i=1\;\text{or}\;x'_i=1\}$. From here, the closed-form expressions for the covariance matrix entries $\Sigma_{\mathbf{x},\mathbf{x'}}$ follow from explicit calculation by summing over the permutations $\sigma \in \bbS_4$.
\end{proof}

\begin{fact}\label{fact:FourierProbStats}
The average and variance of the purity Fourier probabilities $\mathrm{P}(\mathbf{y};\,\phi)$ when the state $\ket{\phi}$ is sampled from the Haar measure are:
    \begin{align}
        \E_{\phi\sim\Haar_n} \mathrm{P}(\mathbf{y};\,\phi) &=  \frac{3^{n - \abs{\mathbf{y}}} (1 + (-1)^y)}{2^n(2^n+1)} = \frac{2\cdot 3^{n-\abs{\mathbf{y}}}}{2^n(2^n+1)}\delta_{\abs{\mathbf{y}}\text{ even}}\\
        \Var_{\phi\sim\Haar_n} \mathrm{P}(\mathbf{y};\,\phi) &= \frac{2^{2-4 n} \cdot 3^{n-2 \abs{\mathbf{y}}} \left(2^n \left(2^n+1\right) 3^{\abs{\mathbf{y}}}-2\cdot 3^n\right)}{\left(2^n+1\right)^2\left(2^n+2\right) \left(2^n+3\right)}\delta_{\abs{\mathbf{y}}\text{ even}}.
    \end{align}
\end{fact}
\begin{proof}
    The proof follows from explicit calculations with the results of \Cref{fact:Haar1}, applied to the definition \eqref{eq:FourierPurityProbabilities2} of the probabilities $\mathrm{P}(\mathbf{y};\,\phi)$.
\end{proof}

\begin{remark}
    While numerical evidence suggests a much stronger self-averaging result applies at the level of each purity $\Tr[\phi_\mathbf{x}^2]$ of a Haar-random sample, applying the above Chebyshev tail + naïve union bound to individual purities fails to confirm this. However, one reason why this simple approach succeeds in \Cref{lemma:fourierselfconcentrate} to bound the Fourier probabilities $\mathrm{P}(\mathbf{y};\,\phi)$ is that, in fact, the purity covariance matrix $\Sigma_{\mathbf{x},\mathbf{x'}}=\Cov_{\phi}[\Tr\phi_\mathbf{x}^2,\;\Tr\phi_{\mathbf{x'}}^2]$ computed in \Cref{fact:Haar1} above is approximately diagonalized by the Boolean Fourier transform. Letting $F=H^{\otimes n}$ be the Boolean Fourier transform unitary:
    \begin{equation}
        F_{\mathbf{y},\mathbf{x}} = \frac{1}{2^{n/2}}(-1)^{\mathbf{y}\cdot\mathbf{x}}\,,
    \end{equation}
    then we have that in the Fourier basis the purity covariance is almost diagonal:
    \begin{equation}
        (F\Sigma F^\dagger)_{\mathbf{y},\mathbf{y'}} = \frac{4\cdot 3^{n-\abs{\mathbf{y}}}}{(2^n+1)(2^n+2)(2^n+3)}\left(\delta_{\mathbf{y},\mathbf{y'}} - \frac{2}{2^n(2^n+1)}3^{n-\abs{\mathbf{y'}}}\right)\delta_{\abs{\mathbf{y}}\text{ even}}\delta_{\abs{\mathbf{y'}} \text{ even}}.
    \end{equation}
    In other words, the covariance matrix $\Sigma_{\mathbf{x},\mathbf{x'}}$ is dominated by the first term in equation \eqref{eq:covariancepurity}.
\end{remark}

\newpage
\section{Simon's algorithm with non-uniform samples and the hidden cut algorithm for Haar product states}\label{app:NonUniformSimons}

Recall that sampling $\mathbf{y}\in\Z_2^n$ from the distribution $\mathrm{P}_{\mathrm{StateHSP}_{\psi,2}}$ (which we will denote $\mathrm{P}_{\psi,2}$ for ease of notation) defined in \eqref{eq:TwoCopyHaarOutcomeDistribution} is equivalent to sampling each of the $n$ entries i.i.d. from a Bernoulli with $1/4$ probability of yielding 1 and $3/4$ probability of yielding 0, and afterwards keeping the sample if it has even weight on both sides of the cut. To restate the explicit form of the distribution \eqref{eq:TwoCopyHaarOutcomeDistribution}:
\begin{equation}\label{eq:TwoCopyHaarOutcomeDistribution2}
    \mathrm{P}_{\psi,2}[\mathbf{y}_1^{C}\mathbf{y}_2^{\overline{C}}] = \delta_{\abs{\mathbf{y}_1}\text{ even}}\,\delta_{\abs{\mathbf{y}_2}\text{ even}}\left(\frac34\right)^n 3^{-\abs{\mathbf{y}_1}-\abs{\mathbf{y}_2}}\;(1+O(2^{-n/2}))\,.
\end{equation}
Let us first prove two helpful statements which will build towards \Cref{thm:haar}. The first fact is a convenient simplification:

\begin{fact}\label{fact:decoupledBernoullis}
    Assume $p\leq O(n)$. Averaging over $p$ samples from the $\mathrm{P}_{\psi,2}$ cut-specific distribution can be replaced with averaging over the simpler distribution in which the entries are independent Bernoullis, up to a negligible relative correction:
    \begin{align}
    \E_{\mathbf{y}^{(1)},\dots,\mathbf{y}^{(p)}\sim\mathrm{P}_{\psi,2}}\mathrm{P}_{\psi,2}\left[\mathbf{y}^{(1)}\oplus\dots\oplus\mathbf{y}^{(p)}\right] &= \E_{\mathbf{y}^{(1)},\dots,\mathbf{y}^{(p)}\sim \Ber(1/4)^n}\mathrm{P}_{\psi,2}\left[\mathbf{y}^{(1)}\oplus\dots\oplus\mathbf{y}^{(p)}\right]\left(1+O(2^{-n/2})\right)\\
        &=q_p\left(1+O(2^{-n/2})\right),
    \end{align}
    where we define the useful quantity:
    \begin{equation}\label{eq:q_p}
        q_{p}\equiv 4\left(\frac{2+2^{-p}}{4}\right)^n\,.
    \end{equation}
\end{fact}

\begin{proof}
    Taking a fixed $\mathbf{z}=\mathbf{z}_1^C\mathbf{z}_2^{\overline{C}}\in\Z_2^n$ of even weight on both sides of the cut (i.e.\  $\abs{\mathbf{z}_1},\abs{\mathbf{z}_2}$ even), consider the quantity:
    \begin{align}
        \E_{\mathbf{y}\sim\mathrm{P}_{\psi,2}} 3^{-\abs{\mathbf{y}\oplus \mathbf{z}}} &=\sum_{0\leq d_1,d_2\leq n/2}\sum_{\substack{0\leq b\leq \abs{\mathbf{z}_1}\\0\leq b\leq \abs{\mathbf{z}_2}}}\Pr_{\mathbf{y}_1^C\mathbf{y}_2^{\overline{C}}}\left[\abs{\mathbf{y}_1}=2b_1+d_1-\abs{\mathbf{z}_1}\;\cap\;\abs{\mathbf{y}_2}=2b_2+d_2-\abs{\mathbf{z}_2}\right]\;3^{-d_1-d_2},
    \end{align}
    where we define $d_1=\abs{\mathbf{y}_1\oplus \mathbf{z}_1}$ and $b_1=\abs{\mathbf{y}_1\cap\mathbf{z}_1}=\abs{\{i\in[n/2]\,:\,y_{1,i}=z_{1,i}=1\}}$ (and similarly for $d_2$, $b_2$). Notice from \eqref{eq:TwoCopyHaarOutcomeDistribution2} that the distribution over $\mathbf{y}$ completely factorizes over the two sides of the cut $\mathbf{y}_1$ and $\mathbf{y}_2$, and so does the above sum, so it is enough to study only one side of the cut. We can therefore express explicitly, using \eqref{eq:TwoCopyHaarOutcomeDistribution2}:
    \begin{align}
        \sum_{\substack{0\leq d_1\leq n/2\\0\leq b\leq \abs{\mathbf{z}_1}}}\Pr_{\mathbf{y}_1}\left[\abs{\mathbf{y}_1}=2b_1+d_1-\abs{\mathbf{z}_1}\right]\;3^{-d_1} &= 2\frac{3^{n/2-\abs{\mathbf{z}_1}}}{2^{n/2}\left(2^{n/2}+1\right)}\sum_{\substack{0\leq d_1\leq n/2\;\text{even}\\0\leq b_1\leq \abs{\mathbf{z}_1}}}\binom{\abs{\mathbf{z}_1}}{b_1}\binom{\frac{n}{2} - \abs{\mathbf{z}_1}}{d_1-\abs{\mathbf{z}_1}+b_1}3^{-2b_1}.
    \end{align}
    There are two observations. First, since $\abs{\mathbf{y}_1}=2b_1+d_1-\abs{\mathbf{z}_1}$ has to be even, and since we assume $\abs{\mathbf{z}_1}$ is also even, then this restricts $d_1$ to only take even values. Second, the $d_1$ in the exponent cancels out, and the sum over $d_1$ becomes easy, since it is exactly half of a sum over a full set of binomial coefficients, leading to:
    \begin{align}
        \sum_{\substack{0\leq d_1\leq n/2\\0\leq b\leq \abs{\mathbf{z}_1}}}\Pr_{\mathbf{y}_1}\left[\abs{\mathbf{y}_1}=2b_1+d_1-\abs{\mathbf{z}_1}\right]\;3^{-d_1} &= \frac{3^{n/2-\abs{\mathbf{z}_1}}2^{-\abs{\mathbf{z}_1}}}{2^{n/2}\left(2^{n/2}+1\right)}\sum_{0\leq b_1\leq \abs{\mathbf{z}_1}}\binom{\abs{\mathbf{z}_1}}{b_1}3^{-2b_1}\\
        &=\frac{3^{n/2}}{2^{n/2}\left(2^{n/2}+1\right)}\left(\frac{5}{27}\right)^{\abs{\mathbf{z}_1}}.
    \end{align}
    It is helpful to notice that this is negligibly close to the result obtained when $\mathbf{y}$ is sampled not from the $\mathrm{P}_{\mathrm{StateHSP}_{\psi,2}}$ distribution, but from the much simpler distribution in which each entry is i.i.d. sampled from a $\Ber(1/4)$ distribution, i.e.\  removing the parity constraints. In that case, the result is:
    \begin{align}
        \sum_{\substack{0\leq d_1\leq n/2\\0\leq b\leq \abs{\mathbf{z}_1}}}\Pr_{\mathbf{y}_1\sim\Ber(1/4)^{n/2}}\left[\abs{\mathbf{y}_1}=2b_1+d_1-\abs{\mathbf{z}_1}\right]\;3^{-d_1}
        &=\left(\frac{3}{4}\right)^{n/2}\left(\frac{5}{27}\right)^{\abs{\mathbf{z}_1}}.
    \end{align} 
    This means that:
    \begin{align}
         \E_{\mathbf{y}\sim\mathrm{P}_{\mathrm{StateHSP}_{\psi,2}}} 3^{-\abs{\mathbf{y}\oplus \mathbf{z}}} = \E_{\mathbf{y}\sim\Ber(1/4)^n} 3^{-\abs{\mathbf{y}\oplus \mathbf{z}}}\,\left(1+O(2^{-n/2})\right).
    \end{align}
    We can apply this result to estimate the more useful quantity:
    \begin{align}
        \E_{\mathbf{y}^{(1)},\dots,\mathbf{y}^{(p)}\sim \mathrm{P}_{\mathrm{StateHSP}_{\psi,2}}}3^{-\abs{\mathbf{y}^{(1)}\oplus\dots\oplus\mathbf{y}^{(p)}}}.
    \end{align}
    Replacing the $\mathrm{P}_{\mathrm{StateHSP}_{\psi,2}}$ distribution with independent entrywise Bernoullis simplifies the calculation significantly, since:
    \begin{equation}
        \abs{\mathbf{y}^{(1)}\oplus\dots\oplus\mathbf{y}^{(p)}} = \sum_{i=1}^n \mathds{1}\left[y^{(1)}_i+\dots+y^{(p)}_i=1\;\mathrm{mod}\;2\right].
    \end{equation}
    This means that, with the unrestricted Bernoulli distribution, the average from above can be factored across the indices $\{1,\dots,n\}$, leading to the closed-form expression:
    \begin{align}
        \E_{\mathbf{y}^{(1)},\dots,\mathbf{y}^{(p)}\sim\Ber(1/4)^n}3^{-\abs{\mathbf{y}^{(1)}\oplus\dots\oplus\mathbf{y}^{(p)}}} &= \left(\E_{\mathbf{u}\sim\Ber(1/4)^p}3^{-\left(\mathbf{u}\cdot1^{p}\;\mathrm{mod}\;2\right)}\right)^n \\
        &=\left(\frac34\right)^{pn}\left(\sum_{\substack{0\leq u \leq p\\u\text{ even}}}\binom{p}{u}3^{-u}+\sum_{\substack{0\leq u \leq p\\u\text{ odd}}}\binom{p}{u}3^{-u-1}\right)^n\\
        &=\left(\frac{2+2^{-p}}{3}\right)^n
    \end{align}
    Equivalently, let us define:
    \begin{align}
        \E_{\mathbf{y}^{(1)},\dots,\mathbf{y}^{(p)}\sim\mathrm{P}_{\mathrm{StateHSP}_{\psi,2}}}\mathrm{P}_{\mathrm{StateHSP}_{\psi,2}}\left[\mathbf{y}^{(1)}\oplus\dots\oplus\mathbf{y}^{(p)}\right] &= 4\left(\frac{2+2^{-p}}{4}\right)^n\left(1+O(2^{-n/2})\right)\\
        &= q_p\left(1+O(2^{-n/2})\right)\,.
    \end{align}
\end{proof}

Second, let us prove a lower bound on the relevant quantity we are aiming to estimate:

\begin{fact}\label{fact:ConstantProbabilityLinearIndependent}
    Let us define the probability that $k$ samples are linearly independent as vectors over $\Z_2^n$:
    \begin{align}
        \pi_k\equiv \Pr_{\mathbf{y}^{(1)},\dots,\mathbf{y}^{(k)}\sim\mathrm{P}_{\psi,2}}\left[(\mathbf{y}^{(j)})_{j\in[k]}\,\text{lin. indep.}\right],
    \end{align}
    where $\{\mathbf{y}^{(1)},\dots,\mathbf{y}^{(k)}\}\subset H_C^\perp$ are understood to be $k$ independent samples from the $\mathrm{P}_{\psi,2}$ distribution \eqref{eq:TwoCopyHaarOutcomeDistribution2}. Then we have the lower bound:
    \begin{equation}
        \pi_k \geq 1 - \sum_{p=1}^k \binom{k}{p}q_{p-1}\left(1+O(2^{-n/2})\right)\,.
    \end{equation}
    In particular, it follows that that $\pi_{n-3}\geq 1/2 + O(\negl(n))$.
\end{fact}

\begin{proof}  Explicitly write the probability $\pi_k$ as the sum over linear independent combinations of vectors:
    \begin{align}
        \pi_k &= \sum_{\substack{\mathbf{y}^{(1)},\dots,\mathbf{y}^{(k)}\in\Z_2^n\\\mathbf{y}^{(1)},\dots,\mathbf{y}^{(k)}\,\text{ lin.indep.}}} \mathrm{P}_{\psi,2}\left[\mathbf{y}^{(1)}\right]\dots \mathrm{P}_{\psi,2}\left[\mathbf{y}^{(k)}\right] \\
        & = \sum_{\substack{\mathbf{y}^{(1)},\dots,\mathbf{y}^{(k-1)}\in\Z_2^n\\\mathbf{y}^{(1)},\dots,\mathbf{y}^{(k-1)}\,\text{ lin.indep.}}} \mathrm{P}_{\psi,2}\left[\mathbf{y}^{(1)}\right]\dots \mathrm{P}_{\psi,2}\left[\mathbf{y}^{(k-1)}\right]\sum_{\substack{\mathbf{y}^{(k)}\in\Z_2^n\\\mathbf{y}^{(k)}\notin\mathrm{span}\left\{\mathbf{y}^{(j)}\right\}_{j=1}^{k-1}}}\mathrm{P}_{\psi,2}\left[\mathbf{y}^{(k)}\right]\\
        &= \sum_{\substack{\mathbf{y}^{(1)},\dots,\mathbf{y}^{(k-1)}\in\Z_2^n\\\mathbf{y}^{(1)},\dots,\mathbf{y}^{(k-1)}\,\text{ lin.indep.}}} \mathrm{P}_{\psi,2}\left[\mathbf{y}^{(1)}\right]\dots \mathrm{P}_{\psi,2}\left[\mathbf{y}^{(k-1)}\right]\left(1-\sum_{a\in\B^{k-1}}\mathrm{P}_{\psi,2}\left[a_1\mathbf{y}^{(1)}\oplus\dots\oplus a_{k-1}\mathbf{y}^{(k-1)}\right]\right)\\
        &=\pi_{k-1} - \sum_{p=0}^{k-1}\binom{k-1}{p}\sum_{\substack{\mathbf{y}^{(1)},\dots,\mathbf{y}^{(k-1)}\in\Z_2^n\\\mathbf{y}^{(1)},\dots,\mathbf{y}^{(k-1)}\,\text{ lin.indep.}}}\mathrm{P}_{\psi,2}\left[\mathbf{y}^{(1)}\right]\dots\mathrm{P}_{\psi,2}\left[\mathbf{y}^{(k-1)}\right]\mathrm{P}_{\psi,2}\left[\mathbf{y}^{(1)}\oplus\dots\oplus\mathbf{y}^{(p)}\right]\\
        &\geq \pi_{k-1} - \sum_{p=0}^{k-1}\binom{k-1}{p}q_p\left(1+O(2^{-n/2})\right)\,.\label{eq:pi_k_lowerbound}
    \end{align}
    The fourth line comes from symmetry, and the final line comes from bounding:
    \begin{align}
        & \sum_{\substack{\mathbf{y}^{(1)},\dots,\mathbf{y}^{(k-1)}\in\Z_2^n\\\mathbf{y}^{(1)},\dots,\mathbf{y}^{(k-1)}\,\text{ lin.indep.}}}\mathrm{P}_{\psi,2}\left[\mathbf{y}^{(1)}\right]\dots\mathrm{P}_{\psi,2}\left[\mathbf{y}^{(k-1)}\right]\mathrm{P}_{\psi,2}\left[\mathbf{y}^{(1)}\oplus\dots\oplus\mathbf{y}^{(p)}\right] &\\
        &\leq \sum_{\mathbf{y}^{(1)},\dots,\mathbf{y}^{(k-1)}\in\Z_2^n}\mathrm{P}_{\psi,2}\left[\mathbf{y}^{(1)}\right]\dots\mathrm{P}_{\psi,2}\left[\mathbf{y}^{(k-1)}\right]\mathrm{P}_{\psi,2}\left[\mathbf{y}^{(1)}\oplus\dots\oplus\mathbf{y}^{(p)}\right]&\\
        &= \E_{\mathbf{y}^{(1)},\dots,\mathbf{y}^{(p)}\sim\mathrm{P}_{\psi,2}}\mathrm{P}_{\psi,2}\left[\mathbf{y}^{(1)}\oplus\dots\oplus\mathbf{y}^{(p)}\right]&\\
        &= q_p \left(1+O(2^{-n/2})\right). &\text{(using \Cref{fact:decoupledBernoullis})}
    \end{align}
    Therefore, unfolding the recursion \eqref{eq:pi_k_lowerbound} down to the base $\pi_0= 1$, we get the desired lower bound:
    \begin{align}
        \pi_k &\geq 1 - \sum_{\l=0}^{k-1}\sum_{p=0}^{\l}\binom{\l}{p}q_{p}\left(1+O(2^{-n/2})\right) \\
        &= 1 - \sum_{p=0}^{k-1} \binom{k}{p+1}q_{p}\left(1+O(2^{-n/2})\right).\label{eq:pi_klowerbound}
    \end{align}
    Given the explicit form for the $q_p$ quantities derived previously in \eqref{eq:q_p}, we have that the choice of $k = \Theta(n)$ means the above lower bound for $\pi_k$ is well-approximated by taking $q_p\approx 2^{2-n}$. Specifically, assume $k=\Theta(n)$, and upper bound the relevant sum as:
    \begin{align}
        \sum_{p=0}^{k-1} \binom{k}{p+1}q_{p} &= 2^{2-n}\sum_{p=1}^{k} \binom{k}{p}\left(1+2^{-p}\right)^n\\
        &\leq 2^{2-n}\left(\frac{3}{2}\right)^n\sum_{p=0}^{\lfloor\log^2(n)\rfloor} \binom{k}{p} + 2^{2-n}\left(1+\frac{1}{2^{\lfloor\log^2(n)\rfloor}}\right)^n\sum_{p=\lfloor\log^2(n)\rfloor+1}^k \binom{k}{p},
    \end{align}
    where we have simply split the sum at an appropriate term of order polylogarithmic in $n$. The first contribution can be upper bounded by a typical H\"{o}ffding tail bound of the binomial distribution, and the second contribution can be bounded by a standard binomial sum. This results in:
    \begin{align}
        \sum_{p=0}^{k-1} \binom{k}{p+1}q_{p} &\leq 2^{2+k-n}\left[\exp(-\frac{k}{2}+n\ln\frac{3}{2}+O(\polylog(n))) + \left(1+\frac{1}{2^{\lfloor\log^2(n)\rfloor}}\right)^n\right]\,.
    \end{align}
    The second term above is negligibly close to one. Also, when $k=n-b$ for a constant $b>0$, the first term is exponentially decaying in $n$, since $\frac12 - \ln\frac32\approx 0.0945\,>\,0$. This means that in this case, the lower bound derived above \eqref{eq:pi_klowerbound} becomes:
    \begin{align}
        \pi_{n-b}\geq 1 - 2^{2-b}\left(1 + \negl(n))\right),
    \end{align}
    which in particular means that $\pi_{n-3} \geq 1/2$ up to a negligible correction.
\end{proof}

We have the necessary ingredients to assemble the proof of \Cref{thm:haar}:

{\bf\em Proof of \Cref{thm:haar}.} In the main text, we have shown in \Cref{fact:HaarTwoCopyOutcomeDistr} that, with high probability over the Haar-random factor states, purity Fourier sampling with two copies of the input state with a hidden cut yields samples from the $\mathrm{P}_{\psi,2}$ probability distribution \eqref{eq:TwoCopyHaarOutcomeDistribution2}. This distribution is supported (non-uniformly) inside the $(n-2)$-dimensional hidden cut subspace $H_C^\perp$, which is defined in \eqref{eq:Sigma_C_subspace_def} as the subspace orthogonal to the equivalent cut strings $1^C0^{\overline{C}}$ and $0^{C}1^{\overline{C}}$:
\begin{equation}
    H_C^\perp = \left\{\mathbf{y}\in\Z_2^n\,:\,\mathbf{y}\cdot 1^C0^{\overline{C}} = \mathbf{y}\cdot0^{C}1^{\overline{C}}=0\;\mathrm{mod}\;2\right\}.
\end{equation}
\Cref{fact:ConstantProbabilityLinearIndependent} proven above gives us that, with probability at least one half (up to negligible corrections), $n-3$ i.i.d. samples $\mathbf{y}^{(1)},\dots,\mathbf{y}^{(n-3)}$ from the $\mathrm{P}_{\psi,2}$ distribution will be linearly independent as vectors in $\Z_2^n$. Assembling these samples as the rows of a Boolean matrix $Y=(\mathbf{y}^{(1)},\dots,\mathbf{y}^{(n-3)})^T \in \Z_2^{(n-3)\times n}$, we have that with probability at least one half the nullspace of $Y$ is the three-dimensional Boolean subspace:
\begin{equation}
    \mathrm{nullspace}\; Y = \mathrm{span}\left\{1^C0^{\overline{C}},\,1^n,\,\mathbf{b}\right\}\,,
\end{equation}
spanned by the two equivalent cut strings and an additional arbitrary vector $\mathbf{b}\in\Z_2^n$. Given the samples, the nullspace can be determined in polynomial time by simple Boolean linear algebra methods. Excluding the trivial vectors $0^n$ and $1^n$, finding the correct nullspace therefore is equivalent to narrowing down the possible hidden cuts to six non-trivial candidates. Each of these candidates can be individually tested up to a constant confidence interval with a constant number of copies of the input state by standard single-cut SWAP tests. In other words, a constant number of copies is required at the end to find the final vector which completes the basis for the cut subspace. This suffices to show that the correct hidden cut can be found with $O(n)$ total number of copies with constant probability, with only two copies at a time consumed either by the purity Fourier sampling or by the final standard SWAP tests.
\hfill\qed\\

\begin{remark}
    Numerical evidence strongly indicates that a complete basis for the cut subspace $H_C^\perp$ can be obtained with constant probability directly  from $n-2$ i.i.d. samples from the $\mathrm{P}_{\psi,2}$ distribution \eqref{eq:TwoCopyHaarOutcomeDistribution2}. This means that, in practice, the standard Simon's protocol (for example, as used in \Cref{thm:main}) will still return a correct answer when applied to this non-uniform distribution. This does not change the $O(n)$ requirement in terms of number of copies, but would remove the need for additional SWAP tests. However, proving this version of the result would likely require a more involved Boolean random matrix analysis, which we leave to future work.
\end{remark}
\end{document}